\documentclass[12pt,a4paper]{article}

\usepackage{preamble} 

\begin{document}

\title{Flow Metrics on Graphs
\footnote{This paper is the MSc Thesis of Lior Kalman under the supervision of Prof. Robert Krauthgamer at the Weizmann Institute of Science.}}
\author{{Lior Kalman \qquad Robert Krauthgamer}
\\ Weizmann Institute of Science
\\ \texttt{\{lior.kalman,robert.krauthgamer\}@weizmann.ac.il} 
}

\maketitle

\begin{abstract}
    Given a graph with non-negative edge weights,
    there are various ways
    to interpret the edge weights and induce a metric on the vertices of the graph.
    A few examples are shortest-path,
    when interpreting the weights as lengths;
    resistance distance,
    when thinking of the graph as an electrical network
    and the weights are resistances;
    and the inverse of minimum $st$-cut,
    when thinking of the weights as capacities.
    
    It is known that the 3 above-mentioned
    metrics can all be derived from flows,
    when formalizing them as convex optimization problems.
    This key observation
    led us to studying a family of metrics that are derived from flows,
    which we call \textit{flow metrics},
    that gives a natural  interpolation
    between the above metrics using a parameter $p$.
    
    We make the first steps in studying  the flow metrics,
    and mainly focus on two aspects:
    (a) understanding basic properties of the flow metrics,
    either as an
    optimization problem 
    (e.g. finding relations between the flow problem
    and the dual potential problem)
    and as a metric function
    (e.g. understanding their structure and geometry);
    and (b) considering 
    methods for reducing the size of graphs,
    either by removing vertices or edges
    while approximating the 
    flow metrics, and thus attaining
    a smaller instance that
    can be used to accelerate running time
    of algorithms and  reduce their storage
    requirements.
    
    Our main result is a lower bound for the number of edges
    required for a resistance sparsifier in the worst case.
    Furthermore, we present a  method
    for reducing the number of edges in a graph
    while approximating the flow metrics,
    by utilizing a method of \cite{cohen2015lp}
    for reducing the size of matrices.
    In addition, we show that the flow metrics 
    satisfy a stronger version of the triangle inequality,
    which gives some information about their structure
    and geometry.
\end{abstract}

\newpage 

\tableofcontents

\newpage 

\section{Introduction}\label{chapter - introduction}

Given a  graph with non-negative edge weights,
there are various ways 
to interpret the weights and derive a metric on the
vertices.
Two  famous examples
are 
to interpret the weights as lengths
or as capacities,
and the derived metric
on the graph
would be the shortest-path metric
or the  inverse of  minimum cut
(known to be an ultrametric),
respectively.
Another example  is to
think of the graph as an electrical network
of resistors,
and interpret the weights
as conductance (inverse resistance), 
which yields the
effective resistance
(also called resistance distance).

It turns out that one can express all  the above scenarios
via flows.
A flow on an edge-weighted graph $G=(V,E,w)$
is a real function $f$ over the edges,
such that at each vertex, the incoming flow
 equals the outgoing flow,
except for some set of boundary vertices,
usually referred to as sources and targets.
We fix an arbitrary orientation
for the edges, and then the sign of $f(e)$
determines the direction of the flow on the edge $e$
(so informally $f(-e)=-f(e)$ by convention).
Let us examine how the abovementioned metrics
are derived using flows.

\paragraph{Shortest Path.}
This is attained by considering a flow
that ships one unit from a single
source $s\in V$ to a single target $t\in V$,
while using the  edges with the minimum total length.
This  is  formulated  as
\begin{equation}\label{eq: shortest path distance - convex minimization problem}
    d(s,t)=\min\left\{\sum_{e\in E}\abs{f(e)}\cdot w(e)
    \;:\;f\text{ ships 1 unit of flow from }s\text{ to }t \right\}.
\end{equation}
We remark that this is a generalization of the 
shortest path, as
shipping flow through multiple shortest paths is allowed.

\paragraph{Minimum \textit{st}-Cuts.}
Using the well known mincut-maxflow theorem,
the inverse of minimum cut can be viewed as a  problem
of minimizing congestion
(the maximum ``load" on an edge),
formulated by
\begin{equation}\label{eq: min-st-cut - via flows definition}
    \frac{1}{\mincut(s,t)}=\min\left\{\max_{e\in E}\frac{\abs{f(e)}}{w(e)}
    \;:\;f\text{ ships 1 unit of flow from }s\text{ to }t \right\}.
\end{equation}

\paragraph{Effective Resistance.}
There are a couple of equivalent ways to define
effective resistance, denoted $\Reff$.
One way to define it 
is as a flow that obeys some physical rules
(called an \textit{electrical flow}).
Another way is via energy minimization, as
\begin{equation}\label{eq: effective-resistance - energy minimization definition}
    \Reff(s,t)=\min\left\{\sum_{e\in E}\frac{\abs{f(e)}^2}{w(e)}
    \;:\;f\text{ ships 1 unit of flow from }s\text{ to }t \right\}.
\end{equation}
Essentially, each term $\frac{\abs{f(e)}^2}{w(e)}$
is the energy (heat dissipation) 
of an edge $e$ with electrical resistance $\frac{1}{w(e)}$,
and we look for a flow that minimizes the total energy.
The effective resistance is known to
capture a lot of properties of the underlying 
graph,
such as commute time and random spanning trees;
it also has a strong connection to
the Laplacian of the graph.

\paragraph{Flow Metrics.}
A natural question that arises is
how to generalize all the metrics seen above.
Can we define a metric  derived from
flows that captures all three cases,
and perhaps find more metrics in this family?

The definition  we study,
which we call the family of \textit{flow metrics},
is the following.
Given a weighted graph $G=(V,E,w)$, and a parameter $1\leq p < \infty$,
define the $d_p$-distance between $s,t\in V$ to be
\begin{equation}\label{eq: p-metric}
    d_p(s,t)=\min\left\{\prs{\sum_{e\in E}\abs{\frac{f(e)}{w(e)}}^p}^{1/p}
    \;:\;
    f\text{ ships 1 unit of flow from }s\text{ to }t
     \right\}.
\end{equation}
To define $d_\infty$,
we take the limit as $p\rightarrow\infty$,
or equivalently  change
the objective
$\prs{\sum_{e\in E}\abs{\frac{f(e)}{w(e)}}^p}^{1/p}$
in \eqref{eq: p-metric}
to $\max_{e\in E}\abs{\frac{f(e)}{w(e)}}$.
Various works give almost linear time algorithms
for computing the $d_p$-distance,
both in the weighted and unweighted case,
see e.g. \cite{adil2021almost, adil2020faster}
where they refer to it as $p$-norm flows.
It is  immediate  that this definition
yields a metric on $V$.
Moreover, it is easy to see that the case  $p=1$
is in fact the ordinary shortest path metric
(on a graph with edge lengths $\frac{1}{w(e)}$),
and  
$\lim_{p\rightarrow\infty} d_p(s,t)=\frac{1}{\mincut(s,t)}$.
Moreover, in the case  $p=2$, $d_2(s,t)^2$
is just the effective resistance
between $s$ and $t$
in a graph $G'$ with squared edge weights.

Thus,
the family of $d_p$-metrics captures the shortest-path metric
($p=1$),
the effective resistance ($p=2$), and  minimum cuts
($p=\infty$).
Our goal is
to better understand this family
and its properties,
and we make the first steps in this direction.

\paragraph{Connection between different values of $p$.}
The two extreme cases $d_1$ (shortest path) and $d_\infty$ 
(minimum cuts), are informally,  not so well behaved
compared to the resistance distance ($p=2$),
and one might be able to interpolate naturally
between them by using other values of $p$
(e.g. see  \cite{LeeNaor2004embedding,CohenMadrySankowskiVladu2017negative} for such applications).

\paragraph{Capturing properties of the underlying graphs.}
As mentioned earlier, the
effective resistance ($p=2$)
captures key properties of the underlying graph, and a 
natural direction  is to extend this  characterizations
to other values of $p$,
or to find other properties of the underlying graphs
captured by them.

\paragraph{Understanding the geometry of the flow metrics.}
An important tool for understanding the structure
of the $d_p$-metrics is embeddings, 
i.e. mapping $G$ into a normed space
while preserving the $d_p$ metric -
in which case the mapping is called
an isometry,
or up to some error
 - in which case we say that the mapping has
distortion $>1$,
see e.g. \cite{MatousekDiscreteGeometry,matouvsek1997embedding, matouvsek2013lectureEmbeddings}.
Some results are known regarding the three special cases,
for example shortest-path ($d_1$)
embeds isometrically into $\ell_\infty$;
effective resistance
($d_2^2$) embeds isometrically  into $\ell_2^2$;
and $d_\infty$ embeds isometrically into $\ell_1$.
We would like to find and compute such embeddings
for other values of $p$,
and furthermore,
we would like to find the best trade-off between
dimension and approximation 
of such embeddings
(e.g. the Johnson-Lindenstrauss Lemma
\cite{johnsonLindenstrauss1984extensions}).

\paragraph{Small sketches.}
Once we understand which metric spaces the
flow metrics embed into,
and reduce the dimension, 
we can easily design small sketches and 
exploit them  to improve running time and storage requirements
of algorithms.
Other examples for graph problems that are naturally solved by 
such embedding techniques
are multicommodity flows problems, cut sparsifiers,
as well as spectral sparsifiers.

\paragraph{Computing all-pairs distances.}
Another important line of research 
is to compute the distance between all pairs of
vertices simultaneously, or to
construct a data structure
that given a query
of a pair of vertices, returns the exact $d_p$-distance
(or an approximation to it) between them.
Such constructions are known for 
the three special cases,
e.g. Gomory-Hu tree
\cite{GomoryHu1961multi}
for $p=\infty$;
distance oracles
\cite{awerbach1993NearLinearDO,thorupZwick2005approximate, chechik2015ApproximateDO},
All-Pairs Shortest-Path
\cite{chan2010more, seidel1995all},
and spanners
\cite{peleg1989graph, althofer1993sparse}
for $p=1$;
and \cite{spielmanSrivastava2011graphSparsificationReff, jambulapati2018efficient}
for $p=2$,
and it is an interesting direction to 
extend these approaches for other values of $p$.

\paragraph{Reducing the size of the graph.}
Techniques for reducing the size
of the graph while preserving exactly
or approximately
a given metric, are an important tool
that could 
improve the running time
and memory usage of algorithms.
One such technique
is the well known Delta-Wye transform \cite{kennelly1899equivalence},
which removes a vertex of degree 3 from the graph
and forms a triangle from its neighbors.
It is known that for each  of the special cases
$p=1,2,\infty$
(shortest path, effective resistance, and minimum cuts)
there exists such a transform
that preserve $d_p$ and depends only on the 3 edges
incident to the vertex being removed
(i.e. oblivious to the rest of the graph).
Thus, it is interesting 
to examine whether this could hold in general
for other values of $p$.

Another method for reducing the size of the graph is via
edge sparsification. 
This topic is very well studied in the literature
and has various  applications and results.
The most noticeable ones are 
spanners \cite{peleg1989graph, althofer1993sparse}
($d_1$-sparsifiers),
resistance sparsifiers
\cite{dinitz2015towards,jambulapati2018efficient, chu2020shortCycleLongPaper}
($d_2$-sparsifiers),
and cut sparsifiers 
\cite{karger1993global, benczur1996approximating}
($d_\infty$-sparsifiers).
There is also a stronger notion of spectral sparsifiers 
\cite{spielman2004nearly, spielmanSrivastava2011graphSparsificationReff, batsonSpielmanSrivastava2012twice},
which preserve the quadratic form of the Laplacian
of the graph, and in particular
preserve both effective resistance and cuts.
Hence, it is natural to try
to achieve such results for other values of $p$,
as well as give lower bounds to this problem.

\subsection{Results for Graph-Size Reduction}

\paragraph{Lower Bound on Resistance Sparsifiers.}
For $p=1$ and $p=\infty$ there are known upper bounds
and matching lower bounds, but
for $p=2$ there is only an upper bound
and no known lower bound.
Our main result is
the first lower bound for resistance sparsifiers ($p=2$).
Formally,
an $\varepsilon$-resistance sparsifier of a graph
$G$, is a graph $G'$ on the same vertex set $V$ 
such that
\begin{equation}
    \forall s,t\in V,\quad
    \Reff_{,G'}(s,t)\in (1\pm \varepsilon)\Reff_{,G}(s,t).
\end{equation}
We remark that $G'$ does not have to be a subgraph of $G$,
and moreover it can have edge weights.
Chu et al \cite{chu2020shortCycleLongPaper}
show that every graph $G$ with $n$ vertices
admits an  $\varepsilon$-resistance sparsifier
with $\widetilde{O}(n/\varepsilon)$ edges.
We conjecture that this is tight (up to 
the polylog factors), and prove a weaker lower bound.
\begin{conjecture}\label{intro: conjecture: lower bound on resistance sparsifiers}
For every $n\geq 2$ and every $\varepsilon>\frac{1}{n}$,
there exists a graph $G$ with $n$ vertices, 
such that every $\varepsilon$-resistance sparsifier of $G$
has $\Omega\prs{n/\varepsilon}$ edges.
\end{conjecture}
\begin{theorem}\label{intro: thm: Omega(n/sqrt(eps)) lower bound on resistance sparsifiers}
For every $n\geq 2$ and every $\varepsilon>\frac{1}{n}$,
there exists a graph $G$ with $n$ vertices, 
such that every $\varepsilon$-resistance sparsifier of $G$
has $\Omega\prs{n/\sqrt{\varepsilon}}$ edges.
\end{theorem}
In fact,
Theorem \ref{intro: thm: Omega(n/sqrt(eps)) lower bound on resistance sparsifiers}
is an easy consequence of the following
bound regarding resistance sparsifiers of the clique.
\begin{lemma}\label{intro: lemma: 1+1/O(n^2) LB on max/min res ratio}
Let $G=(V,E,w)$ be a graph with $|V|=n$ 
and $|E|<\binom{n}{2}$.
Then,
\begin{equation}\label{intro: eq: 1+1/O(n^2) LB on max/min res ratio from lemma}
    \frac{\max_{x\neq y\in V}\Reff(x,y)}
    {\min_{x\neq y\in V}\Reff(x,y)}
    \geq 1+\frac{1}{O(n^2)}.
\end{equation}
\end{lemma}
Moreover, improving the bound in 
\eqref{intro: eq: 1+1/O(n^2) LB on max/min res ratio from lemma}
to $1+\frac{1}{O(n)}$,  would immediately prove 
Conjecture \ref{intro: conjecture: lower bound on resistance sparsifiers}.
Thus, we focus on studying sparsifiers for the clique.
In some cases, we can prove
the stronger  $1+\frac{1}{O(n)}$ bound.
One very interesting case
is of regular graphs,
even when allowing arbitrary edge weights, including 0.
\begin{theorem}\label{intro: theorem: regular graphs cannot approximate resistance distance of clique}
Let $G=(V,E,w)$ be a $k$-regular graph with $|V|=n$ 
and $k<n-1$.
Then,
\begin{equation}
    \frac{\max_{x\neq y\in V}\Reff(x,y)}
    {\min_{x\neq y\in V}\Reff(x,y)}
    \geq 1+\frac{1}{O(n)}.
\end{equation}
\end{theorem}
Intuitively,
the graphs that seem the best fit
for sparsifying the clique
are regular expanders,
and these graphs are captured by Theorem \ref{intro: theorem: regular graphs cannot approximate resistance distance of clique}.
Thus, we believe that our proof
can be generalized to any non-complete graph 
(i.e. with at least one missing edge),
which would prove Conjecture \ref{intro: conjecture: lower bound on resistance sparsifiers}.
We discuss all these results in subsection \ref{section - lower bound on resistance sparsifiers}.

\paragraph{Flow-Metric Sparsifiers.}
Cohen and Peng \cite{cohen2015lp} show that 
by sampling  rows of a matrix $A$
via an importance sampling approach,
one can preserve up to some error 
the term $\norm{Ax}_p$ for every vector $x$,
with high probability.
We use this result to
sparsify dense graphs while preserving
the flow metrics up to some error.
\begin{theorem}\label{intro: thm: main theorem - existence of d_p sparsifiers}
Let $G=(V,E,w)$ be a graph, fix $p\in \left(\frac{4}{3},\infty\right]$
with H\"older conjugate $q$
(i.e. $\frac{1}{p}+\frac{1}{q}=1$),
and let $\varepsilon>0$.
Then there exists a graph $G'=(V,E',w')$ 
that is a $d_p$-sparsifier of $G$, i.e.
\begin{equation}\label{intro: eq: sparsifier concentration guarantee in our main thm}
    \forall s,t\in V,\quad
    d_{p,G'}(s,t)\in\prs{1\pm \varepsilon}d_{p,G}(s,t),
\end{equation}
and has $|E'|=f(n,\varepsilon,p)$ edges,
where
\begin{equation}
    f(n,\varepsilon,p) = 
    \begin{cases}
        n-1     & \text{if }p=\Omega\prs{\varepsilon^{-1} \log n},\\
        \widetilde{O}\prs{n\varepsilon^{-2}} & 
        \text{if }2<p<\infty,\\
        \widetilde{O}\prs{n^{q/2}\varepsilon^{-5}} & 
        \text{if }\frac{4}{3} < p < 2.
    \end{cases}
\end{equation}
\end{theorem}
We discuss this result in subsection \ref{section - flow metric sparsifiers}.
Note that the case of $p=2$ is in fact
the case of  resistance sparsifiers, for which
\cite{chu2020shortCycleLongPaper}
show a better upper bound of $\widetilde{O}\prs{n\varepsilon^{-1}}$ edges.
We remark that it is an open question to give lower
bounds for this problem for $p\neq 2$.
In particular, our proof of Theorem \ref{intro: thm: Omega(n/sqrt(eps)) lower bound on resistance sparsifiers} does not
extend to other values of $p$.

\paragraph{Delta-Wye Transform for Flow Metrics.}
The well known 
$Y-\Delta$ transform \cite{kennelly1899equivalence}
(or in general - Schur Complement 
\cite{haynsworth1968schur})
is an example for a transform
that removes a vertex of degree 3 from the graph
in  a way that preserves the effective resistance
among all other vertices.
An important aspect of the $Y$-$\Delta$ transform
is that it is ``local",
namely, its change to the graph
depends only on the weights of the
edges incident to the vertex being removed,
and is oblivious to the rest of the graph.
It is known that such transforms exist also
for shortest-path and for minimum cuts.
We show that such transforms do not
exist for the family of flow metric
for other values of $p$, stated
informally as follows.
The formal definitions are given in subsection
\ref{section - transforms for flow metrics}
\begin{theorem}\label{intro: thm: non existence of Y-Delta transform}
Let $p\in[1,\infty]$.
There exists a
local transform 
that removes a vertex of degree 3
and forms a triangle from its neighbors
that preserves the $d_p$ 
metric among all other vertices,
if and only if $p= 1,2,\infty$.
\end{theorem}
In addition, we study the interesting case of $p=\infty$
(minimum cuts),
and show that such a  transform
does not exist for removing a vertex
of degree strictly larger than 3,
even though there does
exist one for shortest-path
and effective resistance.
\begin{theorem}\label{intro: thm: non existence of star-mesh transform for p=infinity and k>3}
For every $k>3$,
there does not exist a local
transform that removes a vertex of degree $k$
and forms a $k$-clique from its neighbors
that preserves $d_\infty$ among
all other vertices.
\end{theorem}
Similarly to the $Y$-$\Delta$ transform
for effective resistance, there also
exist local transforms that remove
a vertex of degree 2 or parallel edges,
and preserve the effective resistance.
We show that this transforms
extend naturally to other values of $p$.
We discuss these results in subsection \ref{section - transforms for flow metrics}.

\subsection{Additional Properties of the Flow Metrics}
\paragraph{The Flow Metrics are \text{p}-strong.}
Towards understanding the geometry of the flow metrics,
we show that  they are $p$-strong,
i.e. satisfy a stronger version of the 
triangle inequality.
\begin{theorem}\label{intro: thm: d_p is p-strong}
Let $G$ be a graph, and let $p\in[1,\infty)$.
Then,
\begin{equation}\label{intro: eq: p-strong triangle inequality}
    \forall s,t,v\in V,\quad
    d_p(s,t)^p \leq d_p(s,v)^p+d_p(v,t)^p.
\end{equation}
\end{theorem}
This is known to hold for the special cases
of $p=1,2$, and also for $p=\infty$,
where
\eqref{intro: eq: p-strong triangle inequality}
becomes (in the limit after raising both sides to power $1/p$)
the ultrametric inequality.
We further show that \eqref{intro: eq: p-strong triangle inequality} is tight,
namely for all $p'>p$, the metric $d_p$
is not always $p'$-strong.
We present it in subsection \ref{section - p-strong triangle inequality}.

\paragraph{An Extension of Foster's Theorem.}
Foster's Theorem
\cite{foster1949average, foster1961extension}
states that for every connected
graph $G=(V,E,w)$,
\begin{equation}
    \sum_{xy\in E}w(xy)\Reff(x,y) = |V|-1.
\end{equation}
Since $d_2^2$ is 
the effective resistance
in a graph with squared weights,
it is immediate that
\(
    \sum_{xy\in E}w(xy)^2d_2(x,y)^2 = |V|-1
\).
We  extend this bound to all $p>1$.
\begin{proposition}\label{intro: proposition: extension of fosters thm}
Let $G=(V,E,w)$ be a connected graph,
and let $p>1$.
Then,
\begin{align}
    \text{if }p\geq 2,&&
    \frac{|V|}{2}&\leq \sum_{xy\in E}\prs{w(xy)d_p(x,y)}^{\frac{p}{p-1}}
    \leq |V|-1,\\
    \text{if }p\leq 2,&&
    |V|-1 & \leq \sum_{xy\in E}\prs{w(xy)d_p(x,y)}^{\frac{p}{p-1}}
    \leq |E|.
\end{align}
\end{proposition}
We remark that on trees, the above sum equals $|V|-1=|E|$
for every $p\in(1,\infty]$,
including in the limit $p\rightarrow 1$.
Furthermore, 
on unweighted cycles it holds that
\begin{align}
    \lim_{p\rightarrow \infty}\sum_{xy\in E}d_p(x,y)^{\frac{p}{p-1}}
    = \frac{|V|}{2}.
\end{align}
Thus, all four bounds are existentially
tight, and cannot be strengthened.
We discuss this in subsection \ref{section - monotonicity properties}.

\subsection{Related Work}
The literature contains several variations
of the flow metrics
that have found applications.
One example
is $p$-resistance, 
 formulated as
\begin{equation}
    R_p(s,t) = \min\left\{\sum_{e\in E}\frac{1}{w(e)}\cdot\abs{f(e)}^p
    \;:\;
    f\text{ ships 1 unit of flow from }s\text{ to }t
     \right\}.
\end{equation}
For fixed $p$,
one can express $d_p^p$ as $R_p$ on a 
related graph $G'$ with $p$-powered weights.
However, as $p\rightarrow\infty$,
the effect of the weights on $R_p$
becomes negligible;
a phenomenon that $d_p$ does not suffer from.
We utilize the connection between $R_p$
and $d_p$ to establish a connection
between $d_p$ and a dual
problem of vertex potentials,
as done in \cite{alamgir2011phaseTransitionPresistance}
where such a connection (between
flow and potentials)  is shown.
\cite{alamgir2011phaseTransitionPresistance}
also shows a transition in the ``behavior" of the $p$-resistance
when moving from small values of $p$ to large ones.
The same $p$-resistance is shown
in 
\cite{herbster2010triangle}
to satisfy a variation 
of the triangle inequality,
and we use their technique  to show
that the flow metrics satisfy a stronger
version of the triangle inequality as well.
The $p$-resistance was later extended in 
\cite{pmlr-v51-nguyen16a}
by adding some penalty terms to the objective
function,
which help understanding the structure
of the underlying graph, both local and global
properties of it,
mainly for clustering purposes.

Another variant,
called $p$-norm flow,
was studied in
\cite{henzinger2021cut, fountoulakis2020p,adil2019iterative, adil2020faster, adil2021almost}.
Given a graph $G=(V,E,w)$
with signed edge-vertex incidence matrix $B$,
and a demands vector $d\in \R^V$,
the goal is to find a flow $f\in \R^E$
that satisfies the demands and minimizes some cost function.
This is formulated as
\begin{equation}
    \min_{B^Tf=d}\texttt{cost}(f)
\end{equation}
The formal definitions are given in
subsection \ref{chapter - notations and preliminaries}.
They studied different options for the cost function
and for the constraints,
such as $\norm{f}_p^p$ (unweighted),
$g^Tf+\norm{W_1f}_2^2+\norm{W_2f}_p^p$
(where $g, W_1$ and $W_2$ are some more variables of the problem
that represent a gradient, 2-norm weights, and $p$-norm weights
respectively),
or even allowing the constraint $B^Tf\leq d$
or considering cases where $B$ is some general matrix
(i.e. regression problem).
Their focus 
is on fast computation of the $p$-norm
flows with respect to the proposed cost function
and the constraints,
either exactly or approximately,
and use it to
design fast algorithms for approximating
maximum flow (see e.g. \cite{adil2020faster}),
as well as an applications for graph clustering
(see e.g. \cite{fountoulakis2020p}).
Our focus is on objective $\texttt{cost}(f)=\norm{W^{-1}f}_p$
and  demand vector $d$ corresponding  to a single source and a single target, 
i.e. has $1$ in some entry (the source
vertex), $-1$ in another entry (the target vertex)
and $0$ everywhere else,
the formal definitions are given in subsection \ref{chapter - notations and preliminaries}.

\subsection{Notations and Problem Definition}\label{chapter - notations and preliminaries}

Unless 
stated otherwise,
we assume that graphs are connected, and have
non-negative edge weights.
For a weighted graph $G=(V,E,w)$,
we fix an arbitrary orientation to the edges,
i.e. for every edge, one of its endpoints
is said to be the ``head" of the edge
and the other is said to be its ``tail".
Let $B_{m\times n}$ be a signed edge-vertex  incidence
matrix of $G$
with respect to the given orientation,
namely $B_{e,x}=\begin{cases}
1   &   \text{if }x\text{ is the head of }e,\\
-1   &  \text{if }x\text{ is the tail of }e,\\
0   &   \text{o.w.};
\end{cases}$.
Let $W_{m\times m}$
be a diagonal matrix where $W_{e,e}=w(e)$.
For every vertex $x\in V$ we denote by
$N(x)\subseteq V$ the set of neighbors of $x$.
We also denote by $\chi_x\in \R^V$
the  unit basis vector with 1 in the entry corresponding to $x$
and 0 everywhere else.
In addition, we denote by $\OneVec\in \R^V$
the all ones vector.
A flow that ships one unit from a source vertex $s\in V$
to a target vertex $t\in V$
is a function $f:E\rightarrow \R$
such that $B^Tf = \chi_s-\chi_t$,
where the sign of $f(e)$ represents the direction
of the flow over the edge $e$ (plus sign meaning
the flow goes from the head to the tail,
and minus sign represent flow that goes from the
tail to the head).
It is easy to verify that the condition
$B^Tf=\chi_s-\chi_t$
implies that for every vertex other than
$s$ and $t$,
the incoming flow equals the outgoing flow.

For fixed $p\in[1,\infty]$,
the $\ell_p$ cost of a flow $f\in \R^E$ is defined to be
\begin{equation}
    \norm{W^{-1}f}_{p} = \prs{\sum_{e\in E}\abs{\frac{f(e)}{w(e)}}^p}^{1/p}.
\end{equation}
Using this notation we can rewrite
\eqref{eq: p-metric} as
\begin{equation}\label{eq: d_p metric as convex optimization problem}
        d_{p,G}(s,t) = \min \curprs{ \norm{W^{-1}f}_p 
         \,:\, B^Tf=\chi_s-\chi_t}.
\end{equation}
We will omit the subscript $G$ when it is clear from the context.
We say that $q\in[1,\infty]$ is the H\"older
conjugate of $p\in[1,\infty]$ if it is the (unique) parameter
such that $1/p+1/q=1$,
namely $q=\frac{p}{p-1}$,
and by convention, 1 and $\infty$ are H\"older conjugates
of each other.

\section{Properties of the Flow Metrics}\label{chapter - properties of flow metrics}
In this section we discuss some properties of the flow metrics.
We begin by presenting in subsection \ref{section - Basic Properties}
the dual problem
of \eqref{eq: d_p metric as convex optimization problem}, 
and using it to derive
equivalent definitions of $d_p$  
that will be  useful in subsequent sections.
Moreover, we present a closed-form
solution for $d_p$, which can be viewed as a generalization
of Ohm's law and electrical flows.

In subsection \ref{section - monotonicity properties}
we  discuss some monotonicity properties of the flow metrics.
Specifically, we show that for fixed $p$,
the $d_p$ distance is monotone in the edge weights,
and for fixed edge weights, 
it is monotone in $p$.
Moreover, we show a stronger monotonicity property
of $d_p$, which leads to a generalization of Foster's Theorem
(proving 
Proposition \ref{intro: proposition: extension of fosters thm}).

Finally, in subsection \ref{section - p-strong triangle inequality}
we show that for fixed $p\in[1,\infty)$,
$d_p$ satisfies a stronger version of the triangle inequality
(proving Theorem \ref{intro: thm: d_p is p-strong}),
which gives some information about the geometry of the flow metrics.

\subsection{Basic Properties}\label{section - Basic Properties}

In this subsection 
we show the dual problem of \eqref{eq: d_p metric as convex optimization problem},
and use it to provide characterizations
of the flow metrics, which will be useful
in different contexts.

\paragraph{Dual Problem.}
It is well known that \eqref{eq: d_p metric as convex optimization problem}  
has a dual  problem, that
asks to optimize vertex potentials.
Let $G=(V,E,w)$ be a graph, and fix $p\in[1,\infty]$,
having H\"older conjugate $q$.
For a vector $\varphi\in \R^V$,
viewed as vertex potentials, 
define its  $\ell_q$ cost  to be
\begin{equation}
    \norm{WB\varphi}_{q}= \prs{\sum_{xy\in E}w(xy)^q\abs{\varphi_x-\varphi_y}^q}^{1/q}.
\end{equation}
Similarly to \cite{alamgir2011phaseTransitionPresistance},
for every $s\neq t\in V$ we can 
consider the  optimization
problem,
\begin{equation}\label{eq: bar(d)_p metric as convex optimization problem}
        \bar{d}_p(s,t) = \min \curprs{ \norm{WB\varphi}_q
         \,:\,\prs{\chi_s-\chi_t}^T\varphi=1}.
\end{equation}
Proposition 4 in
\cite{alamgir2011phaseTransitionPresistance} implies that
the optimization problems \eqref{eq: d_p metric as convex optimization problem}
and \eqref{eq: bar(d)_p metric as convex optimization problem}
are equivalent in the following sense.
\begin{claim}\label{claim: d_p = 1 / bar(d)_p}
Let $p\in(1,\infty)$.
Let $G=(V,E,w)$ be a weighted graph.
Then,
\begin{equation}
    \forall s\neq t\in V,\quad
    d_{p}(s,t) = \prs{ \bar{d}_{p}(s,t)}^{-1}.
\end{equation}
\end{claim}
We show how it is derived  formally in Appendix 
\ref{appendix: ommited proofs from basic props}.
We remark that $\bar{d}_p$ is strongly related to 
$\mincut(s,t)$ because
in the limit $p\rightarrow \infty$,
\begin{equation}
    \begin{split}
         \lim_{p\rightarrow\infty} \bar{d}_{p}(s,t)  & = \lim_{p\rightarrow\infty} \min_{\varphi_s-\varphi_t=1}\prs{\sum_{xy\in E}\abs{\varphi_x-\varphi_y}^{1+\frac{1}{p-1}}w(xy)^{1+\frac{1}{p-1}}}^{1-\frac{1}{p}}\\
        & = \min_{\varphi_s-\varphi_t=1}\sum_{xy\in E}\abs{\varphi_x-\varphi_y}\cdot w(xy)\\
        & = \mincut(s,t).
    \end{split}
\end{equation}
Moreover, this shows that Claim
\ref{claim: d_p = 1 / bar(d)_p}
extends to $p=\infty$, because
$d_\infty(s,t)
= (\text{maxflow}(s,t))^{-1}$.

Furthermore, by swapping
the constraints and the objectives
of problem \eqref{eq: bar(d)_p metric as convex optimization problem},
we get the optimization problem,
\begin{equation}\label{eq: tilde(d_p) - dual of the flow problem}
        \widetilde{d}_p(s,t) =  \max \curprs{ \prs{\chi_s-\chi_t}^T\varphi
         \,:\,\norm{WB\varphi}_q=1},
\end{equation}
and establish another connection.
\begin{claim}\label{claim: bar(d_p) = inverse of tilde(d_p)}
Let $G=(V,E,w)$ be a graph, let $p\in[1,\infty]$.
Then,
\begin{equation}
    \forall s\neq t\in V,\quad
    \widetilde{d}_p(s,t)=\prs{\bar{d}_p(s,t)}^{-1}.
\end{equation}
\end{claim}
\begin{proof}
Indeed, let $\varphi^*$ be a solution to
\eqref{eq: bar(d)_p metric as convex optimization problem},
and define 
\(
    \widetilde{\varphi} 
    =
    \frac{\varphi^*}{\norm{WB\varphi^*}_q}
\).
Thus,
\[
    \widetilde{d}_p(s,t)\geq \prs{\chi_s-\chi_t}^T\widetilde{\varphi}
    =\prs{\chi_s-\chi_t}^T \frac{\varphi^*}{\norm{WB\varphi^*}_q}
    = \frac{1}{\norm{WB\varphi^*}_q}=\prs{\bar{d}_p(s,t)}^{-1}.
\]
We can similarly show that $\widetilde{d}_p(s,t)\leq\prs{\bar{d}_p(s,t)}^{-1}$,
and the claim follows.
\end{proof}
Combining   Claims 
\ref{claim: d_p = 1 / bar(d)_p} and
\ref{claim: bar(d_p) = inverse of tilde(d_p)} ,
we  conclude that $d_p=\widetilde{d}_p$.
Next, we use this connection to show another characterization
of the flow metrics, which will
be very useful in later sections.
\begin{claim}\label{claim: d_p as importance sampling}
Let $G=(V,E,w)$ be a graph, let $p\in[1,\infty]$,
with  H\"older conjugate $q$.
Then,
\begin{equation}\label{eq: d_p as importance sampling opt problem}
    \forall s,t\in V,\quad
    d_p(s,t) = \max\curprs{
    \frac{\prs{\chi_s-\chi_t}^T\varphi}{\norm{WB\varphi}_q}
    \;:\;\varphi\in \R^V,\varphi\notin\text{span}\curprs{\OneVec}
    }.
\end{equation}
\end{claim}
\begin{proof}
Since $d_p=\widetilde{d}_p$,
it suffices to show it for $\widetilde{d}_p$.
Fix some $s\neq t\in V$.
First, 
let $\phi^*\in\argmax{\norm{WB\varphi}_q=1}\prs{\chi_s-\chi_t}^T\varphi$, and
since $\norm{WB\phi^*}_q=1$,
it holds that $\phi^*\notin\text{span}\curprs{\OneVec}$
(otherwise the norm would have been 0 since
$\text{ker}B=\text{span}\curprs{\OneVec}$)
and $\prs{\chi_s-\chi_t}\phi^* = \frac{\prs{\chi_s-\chi_t}\phi^*}{\norm{WB\phi^*}_q}$.
Hence, it is immediate that
\begin{equation}
    \widetilde{d}_p(s,t)
    = \prs{\chi_s-\chi_t}^T\phi^*
    \leq \max_{\varphi\in \R^V:\varphi\notin\text{span}\curprs{\OneVec}}
    \frac{\prs{\chi_s-\chi_t}^T\varphi}{\norm{WB\varphi}_q}.
\end{equation}
In the other direction, let
\(
    \varphi^*\in\argmax{\varphi\in \R^V:\varphi\notin\text{span}\curprs{\OneVec}}\frac{\prs{\chi_s-\chi_t}^T\varphi}{\norm{WB\varphi}_q},
\)
and define $\widetilde{\varphi} = \frac{\varphi^*}{\norm{WB\varphi^*}_q}$.
It is easy to see that
$\widetilde{\varphi}\notin\text{span}\curprs{\OneVec}$,
$\norm{WB\widetilde{\varphi}}_q=1$,
and moreover,
\begin{equation}
    \begin{split}
        \widetilde{d}_p(s,t) & \geq \prs{\chi_s-\chi_t}^T\widetilde{\varphi}
        = \frac{\prs{\chi_s-\chi_t}^T\varphi^*}{\norm{WB{\varphi^*}}_q}
        =  \max_{\varphi\in \R^V:\varphi\notin\text{span}\curprs{\OneVec}}
    \frac{\prs{\chi_s-\chi_t}^T\varphi}{\norm{WB\varphi}_q}.
    \end{split}
\end{equation}
\end{proof}

\paragraph{Closed-Form Solution.}
We remark that via the KKT-conditions
we can get a closed-form solution
for the flow metrics,
which can be viewed as a generalization
of Ohm's law and electrical flows,
similarly to \cite{henzinger2021cut}.
\begin{fact}[KKT conditions for $d_p$]\label{fact: closed form solution for d_p}
Let $G=(V,E,w)$ be a graph,
let $p\in (1,\infty)$ with H\"older conjugate $q$, 
let $s,t\in V$, and let $f\in\R^E$
such that $B^Tf = \chi_s-\chi_t$.
Then, $f$ is an optimal flow for $d_p(s,t)$,
if and only if
there exists a potentials vector $\varphi\in \R^V$
such that for every edge \(xy\in E\),
    \(\varphi_x-\varphi_y  = \frac{f(xy)\abs{f(xy)}^{p-2}}{w(xy)^p}\),
    or equivalently,
    \(f(xy)  = w(xy)^q\prs{\varphi_x-\varphi_y}\abs{\varphi_x-\varphi_y}^{q-2}\).
\end{fact}
We remark that
using Fact \ref{fact: closed form solution for d_p},
we can deduce a connection between
the flow metrics and the $p$-Laplacian of the graph,
as well as its second smallest eigenvalue
(see \cite{buhler2009spectralPLaplacianClustering}
for details about the graph $p$-Laplacian
and its second smallest eigenvalue).
We present it in appendix \ref{appendix: connection to the laplacian}.

\subsection{Monotonicity Properties}\label{section - monotonicity properties}

In this subsection we show 
that the flow metrics are ``monotone"
in two manners
- in the edge  weights and in the parameter $p$.
We begin by showing that for fixed $p$,
as the weights increase,
the $d_p$-distance decrease.
We then show that as $p$ increases,
the $d_p$-distance decreases.
Moreover, we show even stronger monotonicity  in $p$,
that as $p$ increases even $d_p^q$ decreases,
where $q$ is the  H\"older conjugate
of $p$
(and changes with $p$).
We then use these properties 
to show a generalization of Foster's Theorem.

\paragraph{Monotonicity in the edge weights.}
\begin{claim}\label{claim: bigger weights means closer metric}
Let $G=(V,E)$ be a graph
and let $w,w':E\rightarrow\R_+$ be two weight functions
on $E$,
such that for every edge $e\in E$, $w(e)\leq w'(e)$.
Fix $p\in[1,\infty]$, and let $d_p,d'_p$ be the corresponding
flow metrics on $G$ with  weight functions $w$ and $w'$
respectively.
Then, 
\begin{equation}
    \forall s,t\in V,\quad d_p(s,t)\geq d'_p(s,t).
\end{equation}
\end{claim}
\begin{proof}
Fix some $s\neq t\in V$,
and let $f^*$ be some minimizing $st$ flow
with respect to the weight function $w$, i.e. 
\(
    f^*\in \argmin{B^Tf=\chi_s-\chi_t}
    \norm{W^{-1}f}_p
\).
Thus, we can see that
\begin{align*}
    d_p(s,t) 
        & = \prs{\sum_{e\in E}\abs{\frac{f^*(e)}{w(e)}}^p}^{1/p}
        && \text{(by definition of }f^*\text{)}\\
        & \geq \prs{\sum_{e\in E}\abs{\frac{f^*(e)}{w'(e)}}^p}^{1/p}
        && \text{(by } w'\geq w\text{)}\\
        & \geq  d'_p(s,t).
\end{align*}
\end{proof}

\paragraph{Monotonicity in \textit{p}.}
We recall 
that for all  $1\leq p \leq {p'} \leq \infty$,
we have
\begin{equation}\label{eq: connection between different l_p norms}
    \forall x\in \R^n,\quad 
    \norm{x}_{p'}\leq\norm{x}_p\leq n^{1/p-1/{p'}}\norm{x}_{p'}.
\end{equation}
Using this bound, we get the following.
\begin{claim}\label{claim: simple connection between different d_p metrics}
Let $G=(V,E,w)$ be a graph.
Let $1\leq p \leq {p'} \leq\infty$, and let $s,t\in V$.
Then
\begin{equation}
    d_{p'}(s,t)\leq d_p(s,t)\leq \abs{E}^{1/p-1/{p'}}d_{p'}(s,t).
\end{equation}
\end{claim}
\begin{proof}
Let
\(
    f_{p}^{*}\in \argmin{B^Tf=\chi_s-\chi_t}
    \norm{W^{-1}f}_p
\)
and
\(
    f_{p'}^{*}\in \argmin{B^Tf=\chi_s-\chi_t}
    \norm{W^{-1}f}_{p'}
\).
Then,
\begin{align*}
    d_p(s,t) 
    &\leq \norm{W^{-1}{f^*_{{p'}}}}_p
    &&( \text{by definition of }d_{p})\\
    & \leq \abs{E}^{1/p-1/{p'}}\norm{W^{-1}{f^*_{{p'}}}}_{p'}
    &&(\text{by }\eqref{eq: connection between different l_p norms})\\
    & = \abs{E}^{1/p-1/{p'}}d_{p'}(s,t)
    &&( \text{by definition of }f^*_{{p'}}).
\end{align*}
By similar calculations we conclude that
$d_{p'}(s,t)\leq d_p(s,t)$.
\end{proof}
The following corollary
 will come  handy
in subsection \ref{section - flow metric sparsifiers},
where we discuss flow metric sparsifiers.
\begin{corollary}\label{corl: large p is essentialy infinity}
Let $G=(V,E,w)$ be a graph on $n$ vertices, 
let $0<\varepsilon<1$, and let
$p\geq 4\varepsilon^{-1} \log n$.
Then,
\begin{equation}
    \forall s,t\in V,\quad
    d_\infty(s,t)\leq d_p(s,t)\leq (1+\varepsilon)d_\infty(s,t).
\end{equation}
\end{corollary}
\begin{proof}
By Claim \ref{claim: simple connection between different d_p metrics},
we have
\begin{equation}\label{eq: we know from connection between d_p and d_inf}
    d_\infty\leq d_p\leq |E|^{1/p}d_\infty.
\end{equation}
Hence, for $p\geq 4\varepsilon^{-1} \log n$ we get,
\begin{equation}
    |E|^{1/p}
    = e^{\frac{1}{p}\log |E|}
    \leq e^{\varepsilon/2}\leq  1+\varepsilon.
\end{equation}
Thus, plugging this into (\ref{eq: we know from connection between d_p and d_inf}) gives the desired result.
\end{proof}

\subsubsection{A Generalization of Foster's Theorem}
Next, we  present another monotonicity
property for the flow metrics, and use it to conclude
lower and upper bounds on the sum 
$\sum_{xy\in E}w(xy)^qd_p(x,y)^q$
(proving Proposition \ref{intro: proposition: extension of fosters thm}).
These bounds can be viewed as a generalization of Foster's Theorem.
We restate it here for clarity.
\begin{proposition}\label{proposition: extension of fosters thm}
Let $G=(V,E,w)$ be a connected graph,
and fix $p>1$ with H\"older conjugate $q$.
Then,
\begin{align}
    \text{if }p\geq 2,&&
    \frac{|V|}{2}&\leq \sum_{xy\in E}w(xy)^qd_p(x,y)^q
    \leq |V|-1,\\
    \label{align: bounds on the sum}
    \text{if } p\leq 2,&&
    |V|-1 & \leq \sum_{xy\in E}w(xy)^qd_p(x,y)^q
    \leq |E|.
\end{align}
\end{proposition}
We remark that for defining the sum for $p=1$,
we take power $1/q$ on both sides and take the limit
$p\rightarrow1$.
Thus, (\ref{align: bounds on the sum})
turns into $\max_{xy\in E}\curprs{w(xy)d_1(x,y)} = 1$.
Proposition \ref{proposition: extension of fosters thm}
is a  consequence of the following lemma.
\begin{lemma}\label{lemma: stronger monotonicity of the flow metric}
Let $G=(V,E,w)$ be a graph (possibly with parallel edges),
fix $p,p'\in(1,\infty]$ with H\"older conjugates
$q,q'$   respectively.
Define $w'=w^{q/q'}$,
and denote $G'=(V,E,w')$.
If $p\leq p'$ then
\begin{equation}\label{eq: what we want to show}
    \forall s,t\in V,\quad
    d_{p,G}(s,t)^q \geq d_{p', G'}(s,t)^{q'}.
\end{equation}
Otherwise, the inequality is reversed (by symmetry).
\end{lemma}

\begin{proof}
Let $s\neq t\in V$,
and recall that by Claim \ref{claim: d_p = 1 / bar(d)_p}
\[
    d_{p,G}(s,t)=\prs{\Bar{d}_{p,G}(s,t)}^{-1}
    = \prs{\min_{\varphi_s-\varphi_t=1}{\prs{\sum_{xy\in E}w(xy)^q\abs{\varphi_x-\varphi_y}^q}^{1/q}}}^{-1}.
\]
Denote a minimizer $\varphi$ by
\[
    \varphi^*\in\argmin{\varphi_s-\varphi_t=1}{\prs{\sum_{xy\in E}w'(xy)^{q'}\abs{\varphi_x-\varphi_y}^{q'}}^{1/q'}}.
\]
Then,
\begin{align*}
    d_{p',G'}(s,t)^{-q'} & =  
        \sum_{xy\in E} w'(xy)^{q'}\abs{
        \varphi^*_x-\varphi^*_y}^{q'}
        && (\text{by definition of }\varphi^*)\\
        & \geq  \sum_{xy\in E}w'(xy)^{q'}\abs{
        \varphi^*_x-\varphi^*_y}^{q}
        &&(\text{by }q'\leq q \text{ and }\abs{
        \varphi^*_x-\varphi^*_y}\leq 1)\\
        & = \sum_{xy\in E}w(xy)^{q}\abs{
        \varphi^*_x-\varphi^*_y}^{q}
        && (\text{by }w^q=(w')^{q'})\\
        & \geq  
        d_{p,G}(s,t)^{-q}.
\end{align*}
\end{proof}
 Lemma \ref{lemma: stronger monotonicity of the flow metric}
implies two of the  bounds in
Proposition \ref{proposition: extension of fosters thm}
as easy corollaries.
\begin{corollary}\label{cor: monotonicity implies sum of w*d_p is at most n-1 for unweighted graphs}
Let $G=(V,E,w)$ be a connected graph 
(possibly with parallel edges)
with $|V|=n$ vertices.
Fix $p\in(1,\infty]$ with H\"older conjugate $q$.
Then,
\begin{equation}\label{eq: from corollary - monotonicity implies bound}
    \sum_{xy\in E} w(xy)^qd_p(x,y)^q
    \leq n-1
    \iff p\geq 2.
\end{equation}
\end{corollary}

\begin{proof}
For $p\geq 2$ with H\"older conjugate
$q$,
denote $w'=w^{q/2}$ and let 
$d'_2$ is the corresponding flow metric
on $G'=(V,E,w')$.
By Lemma \ref{lemma: stronger monotonicity of the flow metric},
\begin{equation}
    \sum_{xy\in E} w(xy)^qd_p(x,y)^q  \leq   \sum_{xy\in E}w'(xy)^2d'_2(x,y)^2
         = n-1.
\end{equation}
where the last equality is exactly Foster's Theorem
about effective resistance.
For $p\in(1,2]$ we use the symmetric case of 
Lemma \ref{lemma: stronger monotonicity of the flow metric}
and Foster's Theorem again.
\end{proof}
We remark that the  bound in (\ref{eq: from corollary - monotonicity implies bound}) is tight
(for every $p>1$), 
since on trees the sum always equals to $n-1$.

Next, we give tight upper and lower bounds for the remaining cases.
We begin by showing the following claim.
\begin{claim}\label{claim: trivial upper bound by 1}
For every graph $G=(V,E,w)$, 
and every $p>1$ with H\"older conjugate $q$,
\begin{equation}\label{eq: bound w^qd_p^q<=1}
    \forall uv\in E,\quad     w(uv)^{q}d_p(u,v)^q\leq 1,
\end{equation}
which implies the
upper bound
$\sum_{xy\in E}w(xy)^qd_p(x,y)^q\leq |E|$.
\end{claim}
We remark that these bounds hold also when
taking the limit $p\rightarrow 1$.
\begin{proof}
Consider first the case $p>1$.
Fix an edge $uv\in E$, and fix $p>1$ with H\"older conjugate $q$.
Then,
\begin{equation}
    \begin{split}
        w(uv)^{-q}d_p(u,v)^{-q}& = w(uv)^{-q}\cdot \min_{\varphi_u-\varphi_v=1}
        \prs{\sum_{xy\in E}w(xy)^q\abs{\varphi_x-\varphi_y}^q}\\
        & = \min_{\varphi_u-\varphi_v=1}
        \prs{\sum_{xy\in E}\prs{\frac{w(xy)}{w(uv)}}^q\abs{\varphi_x-\varphi_y}^q}\\
        & = \min_{\varphi_u-\varphi_v=1}
        \prs{1+\sum_{xy\in E\backslash\curprs{uv}}
        \prs{\frac{w(xy)}{w(uv)}}^q\abs{\varphi_x-\varphi_y}^q}\\
        & \geq 1.
    \end{split}
\end{equation}
The case $p=1$  follows by taking the limit.
\end{proof}
In fact, we have something stronger for several graphs.
\begin{claim}\label{claim: sum is tight in the limit p=1}
    For every unweighted graph $G=(V,E)$ with no parallel edges,
    \begin{equation}
        \forall uv\in E,\quad \lim_{p\rightarrow 1}
        d_p(u,v)^q=1.
    \end{equation}
    and thus the upper bound
    $\sum_{xy\in E}w(xy)^qd_p(x,y)^q\leq |E|$ is 
    existentially tight in the limit $p\rightarrow 1$.
\end{claim}
\begin{proof}
Similarly to the proof of Claim \ref{claim: trivial upper bound by 1},
we have 
\begin{equation}\label{eq: from claim about trivial upper bound by 1}
    d_p(u,v)^{-q}
         = \min_{\varphi_u-\varphi_v=1}
        \prs{1+\sum_{xy\in E\backslash\curprs{uv}}
        \abs{\varphi_x-\varphi_y}^q}.
\end{equation}
Define a potential function $\varphi$ by
\(
    \varphi_x=\begin{cases}
    1   &   \text{if }x=u,\\
    0   &   \text{if }x=v,\\
    \frac{1}{2}   &   \text{o.w.};
    \end{cases}
\).
Plugging this into 
(\ref{eq: from claim about trivial upper bound by 1})
we get 
\begin{equation}
    \begin{split}
        d_p(u,v)^{-q} & \leq 
        1+\sum_{xy\in E\backslash\curprs{uv}}
        \abs{\varphi_x-\varphi_y}^q\\
        & = 1+\sum_{x\in N(u)\backslash\{v\}}\prs{\frac{1}{2}}^q
        + \sum_{y\in N(v)\backslash\{u\}}\prs{\frac{1}{2}}^q\\
        & \leq 1+\frac{2\Delta}{2^q}.
    \end{split}
\end{equation}
where $\Delta=\max_{v\in V}\deg(v)$.
Together with the bound from
Claim \ref{claim: trivial upper bound by 1},
we have 
\[
    1 \leq d_p(u,v)^{-q} \leq  1+\frac{2\Delta}{2^q}\xrightarrow{q\rightarrow\infty}1.
\]
and the claim follows.
\end{proof}
We conclude the case $p\leq 2$
in Proposition \ref{proposition: extension of fosters thm},
by using Corollary \ref{cor: monotonicity implies sum of w*d_p is at most n-1 for unweighted graphs}
and Claim \ref{claim: trivial upper bound by 1}.
Moreover, these bounds are tight - 
the lower bound is tight for $p=2$,
and the upper bound is tight for $p\rightarrow 1$
(Claim \ref{claim: sum is tight in the limit p=1}).
In particular, on trees
the sum always equals $|V|-1=|E|$.

Finally, we give a lower bound for the sum
\(
    \sum_{xy\in E}w(xy)^qd_p^q(x,y)
\) when $p>2$,
which will conclude Proposition \ref{proposition: extension of fosters thm}.
Due to Lemma \ref{lemma: stronger monotonicity of the flow metric},
it suffices to  give a lower bound for the case
$p=\infty$.
Indeed, we have the following bound.
\begin{claim}\label{claim: lower bound on the sum for p geq 2}
    Let $G=(V,E,w)$ be a graph with $|V|=n$ vertices.
    Then
    \begin{equation}
        \sum_{xy\in E}w(xy)d_\infty(x,y)\geq \frac{n}{2}.
    \end{equation}
\end{claim}

\begin{proof}
For every vertex $x\in V$,
\begin{align*}
    \sum_{y\in N(x)}w(xy)d_\infty(x,y)
    & = \sum_{y\in N(x)}\frac{w(xy)}{\mincut(x,y)}\\
    & \geq \sum_{y\in N(x)}\frac{w(xy)}{\text{weighted-deg}(x)}\\
    & = 1.
\end{align*}
By rearranging the sum over the edges into a sum over the vertices,
we obtain
\begin{align*}
    \sum_{xy\in E}w(xy)d_\infty(x,y) & = \frac{1}{2}\sum_{x\in V}\sum_{y\in N(x)}w(xy)d_\infty(x,y)\\
    &  \geq \frac{n}{2}.
\end{align*}
\end{proof}
We remark that the above lower bound is  tight too,
for example  for an $n$-cycle
the sum is indeed $n/2$.

Using Claim \ref{claim: lower bound on the sum for p geq 2},
we conclude the case of $p\geq 2$ 
in Proposition \ref{proposition: extension of fosters thm}.

\subsection{\textit{p}-strong Triangle Inequality for Flow Metrics}\label{section - p-strong triangle inequality}

In this subsection we show that
the flow-metrics satisfy a stronger version of
the triangle inequality,
and discuss some of its properties.
\begin{definition}[$p$-strong metric]
Let $(X,d)$ be a metric space, and fix $p\geq 1$.
We say that $d$ is $p$-strong if it satisfies
the $p$-strong triangle inequality,
\begin{equation}\label{eq: strong triangle inequality that we want to prove}
    \forall x,y,z\in X,\quad d(x,y)^p\leq d(x,z)^p+d(z,y)^p.
\end{equation}
\end{definition}
To extend the definition to $p=\infty$,
we take power $1/p$ of both side of the inequality
and let $p\rightarrow\infty$.
Inequality
(\ref{eq: strong triangle inequality that we want to prove})
then turns to $d(x,y)\leq \max\curprs{d(x,z),d(z,y)}$.
\begin{theorem}\label{theorem: p-strong tirangle inequality}
Let $G=(V,E,w)$ be a graph, and fix $p\in[1,\infty)$.
Then $d_p$ is $p$-strong.
\end{theorem}
For $p=1,2,\infty$ this was known to be true:
For $p=1$ it is trivial,
as (\ref{eq: strong triangle inequality that we want to prove}) is just the regular triangle
inequality.
For $p=\infty$ this is known to be true since
$d_\infty$ is an ultrametric.
For $p=2$,
since $d_2(s,t)^2$
is  the effective resistance
between $s$ and $t$,
it is related to the commute time.
Namely, if we denote $w_2(E)=\sum_{e\in E}w(e)^2$,
then 
$commute(s,t)=2w_2(E)d_2(s,t)^2$.
It is easy to see that for every $s,t,v$
it holds that 
$commute(s,t)
\leq commute(s,v)+commute(v,t) $,
which leads to
$d_2(s,t)^2\leq d_2(s,v)^2+d_2(v,t)^2$.

The proof of Theorem \ref{theorem: p-strong tirangle inequality}
uses the same technique as
in \cite{herbster2010triangle},
who showed a variation of the triangle
inequality for the family of $p$-resistance.
\begin{proof}(of Theorem \ref{theorem: p-strong tirangle inequality})
Let $s,t,v\in V$, and define a new graph $\widetilde{G}=\prs{\widetilde{V},
\widetilde{E}, \widetilde{w}}$
that consists of two
copies of $G$, and a single copy of $v$.
Namely, a copy  $G_1=\prs{V_1,E_1,w}$ and 
another copy $G_2=\prs{V_2,E_2,w}$,
where the two copies of  $v$ from $V_1$ and $V_2$
are identified with the same vertex $v$.
An illustration is given
in figure \ref{fig: new graph}.
\begin{figure}[htb]
    \centering

\tikzset{every picture/.style={line width=0.75pt}} 

\begin{tikzpicture}[x=0.75pt,y=0.75pt,yscale=-1,xscale=1]

\draw  [fill={rgb, 255:red, 251; green, 3; blue, 3 }  ,fill opacity=1 ] (65,85.6) .. controls (65,82.62) and (67.42,80.2) .. (70.4,80.2) .. controls (73.38,80.2) and (75.8,82.62) .. (75.8,85.6) .. controls (75.8,88.58) and (73.38,91) .. (70.4,91) .. controls (67.42,91) and (65,88.58) .. (65,85.6) -- cycle ;
\draw  [fill={rgb, 255:red, 4; green, 27; blue, 249 }  ,fill opacity=1 ] (338,88.6) .. controls (338,85.62) and (340.42,83.2) .. (343.4,83.2) .. controls (346.38,83.2) and (348.8,85.62) .. (348.8,88.6) .. controls (348.8,91.58) and (346.38,94) .. (343.4,94) .. controls (340.42,94) and (338,91.58) .. (338,88.6) -- cycle ;
\draw  [fill={rgb, 255:red, 0; green, 0; blue, 0 }  ,fill opacity=1 ] (178,41.6) .. controls (178,38.62) and (180.42,36.2) .. (183.4,36.2) .. controls (186.38,36.2) and (188.8,38.62) .. (188.8,41.6) .. controls (188.8,44.58) and (186.38,47) .. (183.4,47) .. controls (180.42,47) and (178,44.58) .. (178,41.6) -- cycle ;
\draw  [fill={rgb, 255:red, 0; green, 0; blue, 0 }  ,fill opacity=1 ] (179.4,79.6) .. controls (179.4,76.62) and (181.82,74.2) .. (184.8,74.2) .. controls (187.78,74.2) and (190.2,76.62) .. (190.2,79.6) .. controls (190.2,82.58) and (187.78,85) .. (184.8,85) .. controls (181.82,85) and (179.4,82.58) .. (179.4,79.6) -- cycle ;
\draw  [fill={rgb, 255:red, 0; green, 0; blue, 0 }  ,fill opacity=1 ] (238,107.6) .. controls (238,104.62) and (240.42,102.2) .. (243.4,102.2) .. controls (246.38,102.2) and (248.8,104.62) .. (248.8,107.6) .. controls (248.8,110.58) and (246.38,113) .. (243.4,113) .. controls (240.42,113) and (238,110.58) .. (238,107.6) -- cycle ;
\draw    (75.8,85.6) -- (178,41.6) ;
\draw    (75.8,85.6) -- (179.4,79.6) ;
\draw    (75.8,85.6) -- (172.8,141.6) ;
\draw    (172.8,141.6) -- (184.8,79.6) ;
\draw    (275.8,46.6) -- (338,88.6) ;
\draw    (172.8,141.6) -- (243.4,107.6) ;
\draw  [fill={rgb, 255:red, 0; green, 0; blue, 0 }  ,fill opacity=1 ] (265,46.6) .. controls (265,43.62) and (267.42,41.2) .. (270.4,41.2) .. controls (273.38,41.2) and (275.8,43.62) .. (275.8,46.6) .. controls (275.8,49.58) and (273.38,52) .. (270.4,52) .. controls (267.42,52) and (265,49.58) .. (265,46.6) -- cycle ;
\draw    (184.8,79.6) -- (338,88.6) ;
\draw    (243.4,107.6) -- (338,88.6) ;
\draw    (183.4,41.6) -- (243.4,102.2) ;
\draw    (183.4,41.6) -- (270.4,46.6) ;
\draw    (184.8,79.6) -- (270.4,46.6) ;
\draw  [fill={rgb, 255:red, 251; green, 3; blue, 3 }  ,fill opacity=1 ] (63.68,210.67) .. controls (63.56,214.38) and (65.88,217.39) .. (68.86,217.39) .. controls (71.85,217.39) and (74.36,214.38) .. (74.48,210.67) .. controls (74.6,206.95) and (72.28,203.94) .. (69.3,203.94) .. controls (66.31,203.94) and (63.8,206.95) .. (63.68,210.67) -- cycle ;
\draw  [fill={rgb, 255:red, 4; green, 27; blue, 249 }  ,fill opacity=1 ] (336.8,206.93) .. controls (336.68,210.64) and (339,213.65) .. (341.98,213.65) .. controls (344.97,213.65) and (347.48,210.64) .. (347.6,206.93) .. controls (347.72,203.21) and (345.4,200.2) .. (342.42,200.2) .. controls (339.43,200.2) and (336.92,203.21) .. (336.8,206.93) -- cycle ;
\draw  [fill={rgb, 255:red, 0; green, 0; blue, 0 }  ,fill opacity=1 ] (174.92,265.47) .. controls (174.8,269.19) and (177.12,272.2) .. (180.1,272.2) .. controls (183.08,272.2) and (185.6,269.19) .. (185.72,265.47) .. controls (185.84,261.76) and (183.52,258.75) .. (180.53,258.75) .. controls (177.55,258.75) and (175.04,261.76) .. (174.92,265.47) -- cycle ;
\draw  [fill={rgb, 255:red, 0; green, 0; blue, 0 }  ,fill opacity=1 ] (177.84,218.14) .. controls (177.72,221.85) and (180.04,224.87) .. (183.02,224.87) .. controls (186.01,224.87) and (188.52,221.85) .. (188.64,218.14) .. controls (188.76,214.42) and (186.44,211.41) .. (183.46,211.41) .. controls (180.47,211.41) and (177.96,214.42) .. (177.84,218.14) -- cycle ;
\draw  [fill={rgb, 255:red, 0; green, 0; blue, 0 }  ,fill opacity=1 ] (237.56,183.26) .. controls (237.44,186.98) and (239.76,189.99) .. (242.74,189.99) .. controls (245.73,189.99) and (248.24,186.98) .. (248.36,183.26) .. controls (248.48,179.55) and (246.16,176.53) .. (243.18,176.53) .. controls (240.19,176.53) and (237.68,179.55) .. (237.56,183.26) -- cycle ;
\draw    (74.48,210.67) -- (174.92,265.47) ;
\draw    (74.48,210.67) -- (177.84,218.14) ;
\draw    (74.48,210.67) -- (172.8,141.6) ;
\draw    (172.8,141.6) -- (183.24,218.14) ;
\draw    (272.92,259.25) -- (336.8,206.93) ;
\draw    (173.72,140.91) -- (242.96,183.26) ;
\draw  [fill={rgb, 255:red, 0; green, 0; blue, 0 }  ,fill opacity=1 ] (262.12,259.25) .. controls (262,262.96) and (264.32,265.97) .. (267.3,265.97) .. controls (270.28,265.97) and (272.8,262.96) .. (272.92,259.25) .. controls (273.04,255.53) and (270.72,252.52) .. (267.73,252.52) .. controls (264.75,252.52) and (262.24,255.53) .. (262.12,259.25) -- cycle ;
\draw    (183.24,218.14) -- (336.8,206.93) ;
\draw    (242.96,183.26) -- (336.8,206.93) ;
\draw    (180.32,265.47) -- (242.74,189.99) ;
\draw    (180.32,265.47) -- (267.52,259.25) ;
\draw    (183.24,218.14) -- (267.52,259.25) ;
\draw  [fill={rgb, 255:red, 126; green, 211; blue, 33 }  ,fill opacity=1 ] (167.4,141.6) .. controls (167.4,138.62) and (169.82,136.2) .. (172.8,136.2) .. controls (175.78,136.2) and (178.2,138.62) .. (178.2,141.6) .. controls (178.2,144.58) and (175.78,147) .. (172.8,147) .. controls (169.82,147) and (167.4,144.58) .. (167.4,141.6) -- cycle ;

\draw (44,76.4) node [anchor=north west][inner sep=0.75pt]    {$s_{1}$};
\draw (361,78.4) node [anchor=north west][inner sep=0.75pt]    {$t_{1}$};
\draw (197,133.4) node [anchor=north west][inner sep=0.75pt]    {$v$};
\draw (357,200.4) node [anchor=north west][inner sep=0.75pt]    {$t_{2}$};
\draw (41,204.4) node [anchor=north west][inner sep=0.75pt]    {$s_{2}$};

\end{tikzpicture}
    \caption{Illustration of the construction of the new graph}
    \label{fig: new graph}
\end{figure}
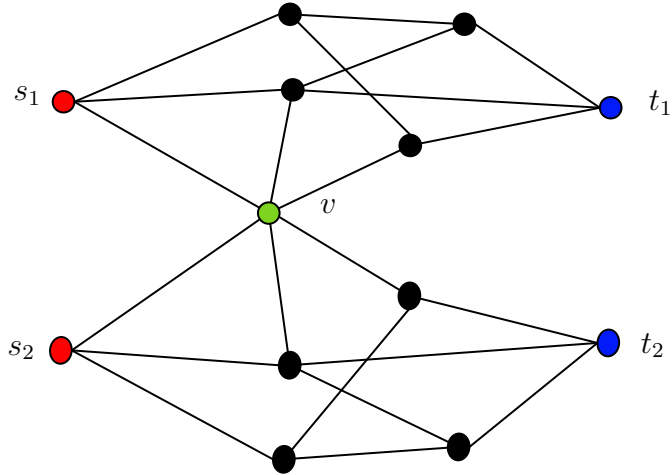

Given a flow $f_1$ that ships 1 unit from
$s$ to $v$ in the original graph $G$, and another flow $f_2$
that ships 1 unit from $v$ to $s$ (also in $G$),
we can define a new flow $\widetilde{f}$ from $s_1$ to $t_2$ in $\widetilde{G}$
that will be the same on the first copy of $G$ as $f_1$,
and on the second copy of $G$ it will be the same as $f_2$.
This is a feasible flow that ships 1 unit from $s_1$ to $t_2$,
and it is easy to see that by taking 
$f_1$ and $f_2$ to be the minimizing flows that attain 
$d_{p,G}(s,v)^p $ and $ d_{p,G}(v,t)^p$ respectively,
we get,
\begin{equation}
    d_{p,\widetilde{G}}\prs{s_1,t_2}^p
    \leq \sum_{e\in E_1}\abs{\frac{f_1(e)}{w(e)}}^p
    +\sum_{e\in E_2}\abs{\frac{f_2(e)}{w(e)}}^p=
    d_{p,G}(s,v)^p+d_{p,G}(v,t)^p.
\end{equation}
Thus, it suffices to prove  that $d_{p,G}\prs{s,t}\leq
d_{p,\widetilde{G}}\prs{s_1,t_2}$.
By Claim \ref{claim: d_p = 1 / bar(d)_p}
 it is enough to prove that
\begin{equation}\label{eq: what we want to prove}
   \overline{d}_{p,G}(s,t)\geq \overline{d}_{p,\widetilde{G}}\prs{s_1,t_2}.
\end{equation}
Let $q$ be the H\"older conjugate of $p$.
Denote by $W_{m\times m}, B_{m\times n}$
the diagonal weight matrix and the signed edge-vertex
incident matrix of $G$,
and denote by $\widetilde{W}, \widetilde{B}$
the diagonal weight matrix and the signed edge-vertex
incident matrix of $\widetilde{G}$.

To prove (\ref{eq: what we want to prove}),
let $\varphi^*\in \argmin{\varphi\in \R^V:
\varphi_s-\varphi_t=1} \norm{WB\varphi}_q$,
be a minimizing potential function on the old graph $G$,
and we will use it to define a new 
potential function $\widetilde{\varphi}$ on $\widetilde{G}$
such that $\norm{WB\varphi^*}_q\geq \norm{\widetilde{W}\widetilde{B}\widetilde{\varphi}}_q$.
This will give,
\begin{equation}
   \overline{d}_{p,G}(s,t)
   = \norm{WB\varphi^*}_q \geq \norm{\widetilde{W}\widetilde{B}\widetilde{\varphi}}_q
   \geq \overline{d}_{p,\widetilde{G}}\prs{s_1,t_2},
\end{equation}
which will conclude the proof.

For every $u_1\in V_1$, define 
$\widetilde{\varphi}\prs{u_1}=\max\curprs{\varphi^*(u),\varphi^*(v)}$
and for every $u_2\in V_2$, define
$\widetilde{\varphi}\prs{u_2}=\min\curprs{\varphi^*(u),\varphi^*(v)}$
(note that $\widetilde{\varphi}(v)=\varphi^*(v)$).
This potential function is at least $\varphi^*(v)$
on the first copy of $G$,
and at most $\varphi^*(v)$ on the second copy of $G$.
Note that in particular, $\widetilde{\varphi}\prs{s_1}-\widetilde{\varphi}\prs{t_2}
=\varphi^*(s)-\varphi^*(t)=1$.

By translation of the potential, we may assume that $\varphi^*(v)=0$.
Fix some edge $uu'\in E$, and note that
its contribution to $\norm{WB{\varphi^*}}_{q}^q$
is $w(e)^{q}\cdot\abs{\varphi^*(u)-\varphi^*(u')}^{q}$,
and the contribution
of its corresponding two edges in $\widetilde{G}$
to $\norm{\widetilde{W}\widetilde{B}\widetilde{\varphi}}_{q}^q$
is 
\begin{equation}\label{eq: subsec 2.3. contribution of corresponging edges to total sum}
  w\prs{u_1u'_1}^{q}\cdot\abs{\widetilde{\varphi}\prs{u_1}-\ \widetilde{\varphi}\prs{u'_1}}^{q}
            +w\prs{u_2u'_2}^{q}\cdot\abs{\widetilde{\varphi}\prs{u_2}-
                \widetilde{\varphi}\prs{u'_2}}^{q}.
\end{equation}
Next, let us examine (\ref{eq: subsec 2.3. contribution of corresponging edges to total sum}) more carefully by separating into four cases.
\begin{itemize}
    \item   \underline{Case 1:} $\varphi^*(u),\varphi^*(u')\geq 0$.
            In this case it holds that:
            \[\widetilde{\varphi}\prs{u_1}=\varphi^*(u),\;\;
                \widetilde{\varphi}\prs{u'_1}=\varphi^*(u'),
            \]
            \[\widetilde{\varphi}\prs{u_2}=
                \widetilde{\varphi}\prs{u'_2}=0.
            \]
            Thus, the contribution of the corresponding edges
            in (\ref{eq: subsec 2.3. contribution of corresponging edges to total sum})
            is the same as in $\norm{WB\varphi^*}_q^q$
            (i.e. 
            $w(uu')^{q}\cdot\abs{\varphi^*(u)-\varphi^*(u')}^{q}
            $).
            
    \item   \underline{Case 2:} $\varphi^*(u),\varphi^*(u')\leq 0$.
            This case is  very similar to the previous case
            since now:
            \[\widetilde{\varphi}\prs{u_1}=
                \widetilde{\varphi}\prs{u'_1}=0,
            \]
            \[\widetilde{\varphi}\prs{u_2}=\varphi^*(u),\;\;
                \widetilde{\varphi}\prs{u'_2}=\varphi^*(u').
            \]
            Thus, again, the contribution of the edges 
            is  
            $w(uu')^{q}\cdot\abs{\varphi^*(u)-\varphi^*(u')}^{q}$.
            
    \item   \underline{Case 3:} $\varphi^*(u)\geq 0\geq \varphi^*(u')$.
            In this case we have:
            \[\widetilde{\varphi}\prs{u_1}=\varphi^*(u),\;
                \widetilde{\varphi}\prs{u'_1}=0,
            \]
            \[\widetilde{\varphi}\prs{u_2}=0,\;\;
                \widetilde{\varphi}\prs{u'_2}=\varphi^*(u').
            \]
            Hence, the contribution is:
            \[
            w(uu')^{q}\cdot\prs{\abs{\varphi^*(u)-0}^{q}
            +\abs{0-\varphi^*(u')}^{q}}
            \leq w(uu')^{q}\cdot\abs{\varphi^*(u)-\varphi^*(u')}^{q}.
            \]
            where we used the fact that $\varphi^*(u)$ and 
            $-\varphi^*(u')$ are non-negative,
            which implies that $\abs{\varphi^*(u)-\varphi^*(u')}
            =\abs{\varphi^*(u)}+\abs{-\varphi^*(u')}$,
            and thus we could apply the inequality
            \[
            \forall \alpha\geq 1,\,
            a,b\geq 0,\quad
            \prs{a+b}^\alpha\geq a^\alpha+b^\alpha.
            \]
            This inequality is proved below
            as Claim 
            \ref{claim: ap+bp leq (a+b)p}.

            To conclude, we got that the contribution of the edges
            in this case is at most the contribution
            of the corresponding edge in $G$.
    \item   \underline{Case 4:} $\varphi^*(u)\leq 0\leq \varphi^*(u')$.
            This case is analogous to the previous case,
            where we use the fact that $\abs{a-b}=\abs{b-a}$
            for any $a,b\in \R$,
            and then repeat the same arguments as in the previous case
            in order to reach the same conclusion.
\end{itemize}
To summarize, we got that for every edge $e=uu'\in E$,
its contribution to $\norm{WB{\varphi^*}}_{q}^q$
is larger than the sum of the contributions of the corresponding edges
$e_1=u_1u'_1\in E_1$ and $e_2=u_2u'_2\in E_2$
to $\norm{\widetilde{W}\widetilde{B}\widetilde{\varphi}}_{q}^q$.
Thus,  it indeed holds that $\norm{WB\varphi^*}_q\geq \norm{\widetilde{W}\widetilde{B}\widetilde{\varphi}}_q$,
which concludes the proof.
\end{proof}
Next, we present two simple properties of the
$p$-strong triangle inequality.
\paragraph{The $p$-strong Triangle Inequality is Monotone in \textit{p}.}
We show that if a metric is $p$-strong for some
$p\geq 1$, then it is also $p'$-strong
for all $p'\in[1, p]$.
\begin{proposition}\label{proposition: p strong triangle inequality is monotone}
Let $(X,d)$ be a metric space that is $p$-strong
for some $p\geq 1$.
Then, for every $p'\in[1, p]$, 
$d$ is $p'$-strong as well.
\end{proposition}
In particular, this  proves that for every
$p\in[1,\infty)$,  $d_p$ is a metric.
In order to prove Proposition \ref{proposition: p strong triangle inequality is monotone},
we will need the following simple claim,
which we prove here for completeness.
\begin{claim}\label{claim: ap+bp leq (a+b)p}
Let $a_1,\dots,a_n\geq 0$ and let $p>0$, then:
\begin{align}
    \text{if }p\leq 1,\quad
    \sum_{i=1}^na_i^p\geq \prs{\sum_{i=1}^na_i}^p.\\
    \text{if }p\geq 1,\quad
    \sum_{i=1}^na_i^p\leq \prs{\sum_{i=1}^na_i}^p.
\end{align}
\end{claim}
\begin{proof}(of Claim \ref{claim: ap+bp leq (a+b)p})
Suppose $p\geq 1$
(the other case is similar),
let $A=\sum_{i=1}^n a_i$, and note that if $A=0$ then the statement clearly
holds, hence we may assume that $A> 0$.
By direct calculations,
\begin{align*}
    \sum_{i=1}^n a_i^p 
    & = A^p\cdot \sum_{i=1}^n\prs{\frac{a_i}{A}}^p\\
    & \leq A^p\cdot \sum_{i=1}^n\frac{a_i}{A}
    &&(\text{by } a_i/A\leq 1\text{ and }p\geq 1)\\
    & = A^p.
\end{align*}
\end{proof}

\begin{proof}(of Proposition \ref{proposition: p strong triangle inequality is monotone})
Let $x,y,z\in X$. Then,
\begin{align*}
        d(x,y)^{p'} 
        & \leq \prs{d(x,z)^p+d(z,y)^p}^{\frac{p'}{p}}
        && (\text{by }p\text{-strong})\\
        & \leq d(x,z)^{p'}+d(z,y)^{p'}
        && (\text{by }p'\leq p \text{ and Claim }\ref{claim: ap+bp leq (a+b)p}).
\end{align*}
\end{proof}

\paragraph{The $p$-strong Triangle Inequality is Tight for $d_p$.}
We have shown that $d_p$ is $p$-strong,
and thus in particular it is also
$p'$-strong for $1\leq p'\leq p$.
Next, we present a simple example that shows
that in the general case, the  power $p$ cannot be strengthened.
\begin{claim}
There exists a graph $G=(V,E)$
such that for every $p\geq 1$
and every $\varepsilon>0$,
the metric $d_p$ is not $(p+\varepsilon)$-strong.
\end{claim}
\begin{proof}
Let $G=(V,E)$ be a 2-path,
i.e. $V=\curprs{s,t,v}$ and $E=\curprs{\curprs{s,v},\curprs{v,t}}$.
It is easy to see that for every $p\geq 1$,
$d_p(s,v)=d_p(v,t)=1$.
Moreover, it holds that
$d_p(s,t)=\prs{1^p+1^p}^{1/p}=2^{1/p}$.
It is easy to verify that indeed the $p$-strong
triangle inequality holds, but for any $\varepsilon>0$,
we can see that
\begin{equation}
    d_p(s,t)^{p+\varepsilon} = \prs{2^{1/p}}^{p+\varepsilon}
    =2^{1+\varepsilon/p}
     > 2 = 1^{p+\varepsilon} + 1^{p+\varepsilon} = d_p(s,v)^{p+\varepsilon}
    +d_p(v,t)^{p+\varepsilon},
\end{equation}
where we used the fact that $2^a$ for $a>0$ is greater than 1.
\end{proof}

\section{Graph-Size Reduction}\label{chapter - flow metric sparsifiers}
In this section we discuss techniques for reducing
the size of a graph while preserving its flow metric
up to some error.
We begin by examining the method of edge sparsification.
\begin{definition}[$d_p$-sparsifier]
Let $G=(V,E,w)$ be a graph, let $p\in[1,\infty]$,
and let $\varepsilon>0$.
A $d_p$-sparsifier of $G$ is a graph $G'=(V,E',w')$,
such that
\begin{equation}
    \forall s,t\in V,\quad
    d_{p,G'}(s,t) \in (1\pm \varepsilon)d_{p,G}(s,t).
\end{equation}
\end{definition}
We remark that for the special cases of $p=1,2,\infty$
there are known upper bounds:
for $p=1$, the definition coincides with multiplicative spanners;
for $p=2$, it coincides with resistance-sparsifiers;
for $p=\infty$, there is the Gomory-Hu tree.
We wish to generalize these 
constructions to other values of $p$.

Moreover, there are matching
lower bounds for $p=1,\infty$, but for $p=2$
no non-trivial lower bound is known.
In subsection \ref{section - lower bound on resistance sparsifiers}
we give the first lower bound for $p=2$ (Theorem
\ref{intro: thm: Omega(n/sqrt(eps)) lower bound on resistance sparsifiers}).
It is essentially a lower bound for resistance sparsifiers of the clique
(Lemma \ref{intro: lemma: 1+1/O(n^2) LB on max/min res ratio}),
and we also discuss some special graph families
(for the sparsifier)
in which the lower bound 
can be strengthened,
including regular graphs
(Theorem \ref{intro: theorem: regular graphs cannot approximate resistance distance of clique}),
which intuitively should be the best fit for
sparsifying the clique.

In subsection \ref{section - flow metric sparsifiers}
we present constructions
of $d_p$-sparsifiers for other values
of $p$ (Theorem \ref{intro: thm: main theorem - existence of d_p sparsifiers}),
that follow easily from 
a Theorem by Cohen and Peng
\cite{cohen2015lp}.
Furthermore, we discuss the relation between
the size of the sparsifier 
and the parameter $p$,
as well as the gaps between the
known constructions for $p=1,2,\infty$
and our construction for other values of $p$.

In subsection \ref{section - transforms for flow metrics}
we discuss a different method to reduce the size
of graphs while preserving exactly the $d_p$-metric,
known as the Delta-Wye transform,
and its generalization the $k$-star-mesh transform
for effective resistance;
the formal definitions are given in subsection \ref{section - transforms for flow metrics}.
We examine for which
values of $p$ and $k$ the
$d_p$ metrics admit such 
transforms.
Specifically we show in Theorem
\ref{intro: thm: non existence of Y-Delta transform}
that for $k=3$
there exists an analogue of the Delta-Wye transform
if and only if $p=1,2,\infty$,
and in Theorem
\ref{intro: thm: non existence of star-mesh transform for p=infinity and k>3}
that for $p=\infty$
there exists an analogue of the $k$-star-mesh transform
if and only if $k\leq 3$.

\subsection{Lower Bound on Resistance Sparsifiers}\label{section - lower bound on resistance sparsifiers}

In this subsection we make partial progress towards proving Conjecture
\ref{intro: conjecture: lower bound on resistance sparsifiers},
which asserts that in the worst case, 
an $\varepsilon$-resistance-sparsifier
of a graph
with $n$ vertices requires at least $\Omega(n/\varepsilon)$
edges.
We begin by showing a weaker lower bound
of $\Omega(n/\sqrt{\varepsilon})$
edges,
and then discuss some special
cases in which we achieve the stronger lower bound.
\begin{definition}[$\varepsilon$-resistance sparsifier \cite{dinitz2015towards}]
Let $G=(V,E,w)$ be a graph.
An $\varepsilon$-resistance sparsifier
of  $G$ is a graph $G'=(V,E',w')$ such that
\begin{equation}
    \forall x,y\in V,\quad
    \Reff_{,G'}(x,y)\in (1\pm \varepsilon)\Reff_{,G}(x,y).
\end{equation}
\end{definition}
Conjecture \ref{intro: conjecture: lower bound on resistance sparsifiers}
is inspired by the following open question.
\begin{oq}\label{oq: is tilde O(n/eps) edges tight?}
Chu et al \cite{chu2020shortCycleLongPaper}
show that every graph with $n$ vertices
admits an $\varepsilon$-resistance
sparsifiers with $\widetilde{O}\prs{n/\varepsilon}$
edges.
Is this  tight?
\end{oq}
It is known that a clique over $n$ vertices
admits an $\varepsilon$-resistance sparsifier
with $O(n/\varepsilon)$ edges.
We present it formally in
\ref{subsection - resistance sparsifier for clique upper bound} 
for completeness.
Thus, the best  lower bound for sparsifying the clique is
of $\Omega(n/\varepsilon)$ edges
(compared to $\widetilde{O}\prs{n/\varepsilon}$
as stated in the question),
which leads the following question,
where we think of $G$ as a possible resistance sparsifier
of the clique.
\begin{oq}\label{oq: does removing an edge from the clique implies a gap in the max/min resistance ratio (even after reweighting)?}
Let $G=(V,E,w)$ be a graph with 
$|V|=n$ and $|E|<\binom{n}{2}$.
Is it true that
\[
    \frac{\max_{x\neq y\in V}\Reff(x,y)}{\min_{x\neq y\in V}\Reff(x,y)}>1+\frac{1}{10n}?
\]
\end{oq}
We show a weaker bound than Question 
\ref{oq: does removing an edge from the clique implies a gap in the max/min resistance ratio (even after reweighting)?},
that 
gives a lower bound of $\Omega\prs{n/\sqrt{\varepsilon}}$
edges for an $\varepsilon$-resistance sparsifier
in the worst case
(Theorem \ref{intro: thm: Omega(n/sqrt(eps)) lower bound on resistance sparsifiers}).
This is in fact 
Lemma \ref{intro: lemma: 1+1/O(n^2) LB on max/min res ratio},
which we restate here for clarity.
\begin{lemma}\label{lemma: 1+1/O(n^2) LB on max/min res ratio}
For any graph $G=(V,E,w)$ with $|V|=n$
and $|E|< \binom{n}{2}$,
it holds that
\[
    \frac{\max_{x\neq y\in V}\Reff(x,y)}{\min_{x\neq y\in V}\Reff(x,y)}>1+\frac{1}{O\prs{n^2}}.
\]
\end{lemma}
Before presenting the proof of Lemma 
\ref{lemma: 1+1/O(n^2) LB on max/min res ratio},
we show how this proves Theorem \ref{intro: thm: Omega(n/sqrt(eps)) lower bound on resistance sparsifiers},
which we restate here for clarity.
\begin{theorem}\label{thm: Omega(n/sqrt(eps)) lower bound on resistance sparsifiers}
For every $n\geq 2$ and every $\varepsilon>\frac{1}{n}$,
there exists a graph $G$ with $n$ vertices, 
such that every $\varepsilon$-resistance sparsifier of $G$
has $\Omega\prs{n/\sqrt{\varepsilon}}$ edges.
\end{theorem}
\begin{proof}(of Theorem \ref{thm: Omega(n/sqrt(eps)) lower bound on resistance sparsifiers})
To see this, take $\Theta(n\sqrt{\varepsilon})$
distinct cliques, each of size $\Theta\prs{{\varepsilon}^{-1/2}}$;
removing even one edge
will fail to achieve $1+\varepsilon$ approximation.
This graph has $\Theta\prs{\varepsilon^{-1}}\cdot \Theta\prs{n\sqrt{\varepsilon}}
= \Theta\prs{n/\sqrt{\varepsilon}}$ edges,
which concludes the lower bound.
\end{proof}
We remark that if Question 
\ref{oq: does removing an edge from the clique implies a gap in the max/min resistance ratio (even after reweighting)?}
is true, it implies that
every $\frac{1}{10n}$-resistance sparsifier of
the clique  must have $\Omega\prs{n^2}$ edges, 
and by following a similar proof as in Theorem
\ref{thm: Omega(n/sqrt(eps)) lower bound on resistance sparsifiers} we conclude  that
in the worst case,
an $\varepsilon$-resistance sparsifier
requires $\Omega\prs{n/\varepsilon}$ edges
(Conjecture \ref{intro: conjecture: lower bound on resistance sparsifiers}).

We proceed to present the proof of Lemma
\ref{lemma: 1+1/O(n^2) LB on max/min res ratio}.
Throughout the proof,
as well as in the next sections,
we use the following equivalent definition
of effective resistance
for a graph $G=(V,E,w)$.
\begin{equation}
        \forall s\neq t\in V,\quad
        \Reff_{, G}(s,t) 
        = \prs{
        \underset{\varphi\in \R^V:
        \varphi_s-\varphi_t=1}{\min}\sum_{xy\in E}w(xy)\prs{\varphi_x-\varphi_y}^2}^{-1}.
\end{equation}
Moreover, we introduce the following notations.
For a subset of edges $F\subseteq E$
we denote $w(F)=\sum_{e\in F}w(e)$,
and for every vertex $x$ we denote its weighted degree
by $\deg_w(x)=\sum_{y\in V}w(xy)$.
\begin{proof}(of Lemma \ref{lemma: 1+1/O(n^2) LB on max/min res ratio})
Without loss of generality, we may assume that
$|E|=\binom{n}{2}-1$ since a missing edge is
an edge of weight $0$.
Let $s,t\in V$ be the pair with a missing edge between them,
denote by $A$ the set of all the edges
that touch $s$ or $t$,
and by $B$ denote the set of the edges that do not 
touch them. Formally,
\begin{align*}
    A & = \curprs{\curprs{x,y}\in E\,:\,
                    \abs{\curprs{x,y}\cap \curprs{s,t}}=1},\\
    B & = \curprs{\curprs{x,y}\in E\,:\,
                    \abs{\curprs{x,y}\cap \curprs{s,t}}=0}. 
\end{align*}
Denote the average edge weight
in each set by
\begin{align*}
    \bar{\alpha} & = \Ex{xy\in A}\brs{w(xy)}
    = \frac{1}{|A|}w(A),\\
    \bar{\beta} & = \Ex{xy\in B}\brs{w(xy)}
    = \frac{1}{|B|}w(B).
\end{align*}
Note that $|A|=2(n-2)$ and $|B|=\binom{n-2}{2}$.
Now, by considering a potential function $\varphi\in \R^V$
given by
            $\varphi_x = \begin{cases}
            1   &   x=s,\\
            0   &   x=t,\\
            1/2   &  x\neq s,t;
            \end{cases}$,
we see that
\begin{equation}\label{eq: upper bound on Reff(s,t)}
    \begin{split}
        \Reff(s,t)^{-1} & \leq
        \sum_{xy\in A}
        w(xy) \frac{1}{4}\\
        & = \frac{|A|\cdot\bar{\alpha}}{4}\\
        & = \frac{n-2}{2}\cdot \bar{\alpha}.
    \end{split}
\end{equation}
Similarly, for every $x,y\in V\backslash\curprs{s,t}$,
by considering a potential function
with values $0,\frac{1}{2},1$
(similarly to the previous case), we get
\begin{equation}\label{eq: upper bound on Reff of an edge}
    \begin{split}
        \Reff(x,y)^{-1} & \leq \frac{\deg_w(x)+\deg_w(y)+2w(xy)}{4}.
    \end{split}
\end{equation}
Let us compute
$\Ex{xy\in B}\brs{\deg_w(x)+\deg_w(y)+2w(xy)}$.
We will first compute the sum,
and then divide by $|B|$.
\begin{equation}
    \begin{split}
        \sum_{xy\in B}\prs{\deg_w(x)+\deg_w(y)+2w(xy)}
        & = 2w(B) + \sum_{xy\in B}\prs{\deg_w(x)+\deg_w(y)}\\
        & = 2w(B) + \sum_{x\in V\backslash\curprs{s,t}}(n-3)\deg_w(x)\\
        & = 2w(B) + (n-3)\cdot\prs{ 2w(E)
        -\deg_w(s)-\deg_w(t)}\\
        & = 2w(B) + (n-3)\cdot\prs{2w(B)+w(A)}\\
        & = 2(n-2)w(B)+(n-3)w(A).
    \end{split}
\end{equation}
Thus,
\begin{equation}\label{eq: expectation of deg(x)+deg(y)+2w(xy) in B}
    \begin{split}
        \Ex{xy\in B}\brs{\deg_w(x)+\deg_w(y)+2w(xy)}
        & = \frac{1}{|B|}\prs{ 2(n-2)w(B)+(n-3)w(A)}\\
        & = 2(n-2)\bar{\beta} + \frac{n-3}{\binom{n-2}{2}}w(A)\\
        & = 2(n-2)\bar{\beta} + \frac{n-3}{\frac{(n-2)(n-3)}{2}}w(A)\\
        & = 2(n-2)\bar{\beta} + \frac{2}{n-2}w(A)\\
        & = 2(n-2)\bar{\beta} + 4\bar{\alpha}.
    \end{split}
\end{equation}
We get that there exists $uv\in B$
such that 
\begin{equation}\label{eq: upper bound on Reff(u,v)}
    \Reff(u,v)^{-1}
    \leq \Ex{xy\in B}\brs{\Reff(x,y)^{-1}}
    \leq \frac{2(n-2)\bar{\beta} + 4\bar{\alpha}}{4}
    = \frac{n-2}{2}\bar{\beta} +\bar{\alpha}.
\end{equation}
where we used (\ref{eq: upper bound on Reff of an edge})
and (\ref{eq: expectation of deg(x)+deg(y)+2w(xy) in B})
in the second transition.

Moreover, note that for all $\tau\in[0,1]$,
\begin{equation}
    \min\curprs{\Reff(s,t)^{-1},\Reff\prs{u,v}^{-1}}
    \leq \tau\cdot \Reff(s,t)^{-1}+ (1-\tau)\cdot \Reff\prs{u,v}^{-1}.
\end{equation}
We set $\tau = \frac{2}{n-2}$,
thus $1-\tau=\frac{n-4}{n-2}$,
which combined with the bounds from
(\ref{eq: upper bound on Reff(s,t)})
and (\ref{eq: upper bound on Reff(u,v)}),
 yields
\begin{equation}
    \begin{split}
        \min_{x\neq y\in V}\Reff(x,y)^{-1}
        & \leq \frac{2}{n-2}\cdot\frac{n-2}{2}\bar{\alpha}
        +\frac{n-4}{n-2}\cdot\prs{\bar{\alpha}+\frac{n-2}{2}\bar{\beta}}\\
        & = \prs{2-\frac{2}{n-2}}\bar{\alpha} + \frac{n-4}{2}\bar{\beta}\\
        & = \prs{\frac{n-3}{(n-2)^2}}\cdot 2(n-2)\bar{\alpha} +\frac{n-4}{2}\bar{\beta}\\
        & = \prs{\frac{n-3}{(n-2)^2}}\cdot w(E)-\prs{\frac{n-3}{(n-2)^2}}\cdot\binom{n-2}{2}\bar{\beta}+ \frac{n-4}{2}\bar{\beta}\\
        & = \prs{\frac{n-3}{(n-2)^2}}\cdot w(E) +\prs{-\frac{n^2-6n+9}{2(n-2)}+\frac{n^2-6n+8}{2(n-2)}}\bar{\beta}\\
        & < \prs{\frac{n-3}{(n-2)^2}}\cdot w(E).
    \end{split}
\end{equation}
This gives us a lower bound on the maximum effective resistance
in the graph.

In order to get an upper bound on the minimum  effective resistance
in the graph,
let us compute the expectation of the 
effective resistance of a random edge $e\in E$
sampled with probability
$\frac{w(e)}{w(E)}$.
By Foster's Theorem,
\begin{equation}
    \begin{split}
        \Ex{xy\in E}\brs{\Reff(x,y)}
        & =\frac{\sum_{xy\in E}w(xy)\Reff(x,y)}{w(E)}\\
        & = \frac{n-1}{w(E)}.
    \end{split}
\end{equation}
There exists an edge $u'v'\in E$
whose effective resistance is at most the expectation, i.e.
\begin{equation}
    \Reff(u',v') \leq \frac{n-1}{w(E)}.
\end{equation}
Altogether, we  see that
\begin{equation}
    \begin{split}
        \frac{\max_{x\neq y\in V}\Reff(x,y)^{-1}}{\min_{x\neq y\in V}\Reff(x,y)^{-1}}
        &> \frac{\frac{w(E)}{n-1}}{\prs{\frac{n-3}{(n-2)^2}}\cdot w(E)}\\
        & = \frac{(n-2)^2}{(n-1)(n-3)}\\
        & = \frac{n^2-4n+4}{n^2-4n+3}\\
        & = 1+\frac{1}{n^2-4n+3}.
    \end{split}
\end{equation}
\end{proof}

\subsubsection{Stronger Bound for Special Cases}
In this subsection we present some special cases in which
we can prove the stronger bound from Question
\ref{oq: does removing an edge from the clique implies a gap in the max/min resistance ratio (even after reweighting)?}.
We begin by showing a technical lemma
stating that if some condition holds,
then the graph cannot guarantee better than
$(1+1/O(n))$-approximation of the clique.
Later, we show that as a matter of fact,
this condition holds in some interesting
special cases.
One very interesting case is when the graph is regular
and arbitrary edge weights are allowed,
including zero 
(Theorem \ref{intro: theorem: regular graphs cannot approximate resistance distance of clique}).

\begin{lemma}\label{lemma: stronger lower bound on special cases with simple conditions}
Let $G=(V,E,w)$ be a graph with $|V|=n$ vertices.
Denote by $\overline{D_w}$ the average weighted degree
of the vertices, i.e. $\overline{D_w} = \frac{1}{n}\sum_{x\in V}\deg_w(x)
=\frac{2w(E)}{n}$.
Suppose that one of the following conditions holds.
\begin{enumerate}
    \item   There exists $v\in V$
            such that $\deg_w(v)\leq
            \frac{\overline{D_w}}{2}
            \cdot\prs{1+\frac{1}{2n}}$.
    \item   There exist 
            $s,t\in V$
            such that
            $\deg_w(s)+\deg_w(t)+2w(st)
            \leq 2\overline{D_w}
            \cdot\prs{1+\frac{1}{2n}}$.
\end{enumerate}
Then,
\begin{equation}
    \frac{\max_{x\neq y\in V}\Reff(x,y)}
    {\min_{x\neq y\in V}\Reff(x,y)}
    \geq 1+\frac{1}{O(n)}.
\end{equation}
\end{lemma}
We remark that the factor $2$
in the denominator (in $1+\frac{1}{2n}$)
is arbitrary, and in fact this can be proved for
any constant $c>1$.
\begin{proof}
Observe that by Foster's Theorem,
we can calculate the expected effective resistance of an edge
when sampling an edge $e$ with probability proportional to its
weight $w(e)$, as follows.
\begin{equation}
    \begin{split}
        \Ex{xy\in E}\brs{\Reff(x,y)} & = \frac{\sum_{xy\in E}w(xy)\Reff(x,y)}{w(E)}\\
        & = \frac{n-1}{\frac{n\overline{D_w}}{2}}\\
        & = \frac{2}{\overline{D_w}}\prs{1-\frac{1}{n}}.
    \end{split}
\end{equation}
In particular, there exists an edge $uv\in E$
such that $\Reff(u,v)\leq \frac{2}{\overline{D_w}}\prs{1-\frac{1}{n}}$.

Next, if the first condition holds,
then by considering a potential function
$\varphi$ such that $\varphi_x=\begin{cases}
1   &   x=v,\\
0   &   \text{o.w.};
\end{cases}$,
we can see that
for any vertex $u\in V\backslash\{v\}$,
\begin{equation}
    \Reff(v,u)^{-1} \leq \deg_w(v) \leq \frac{\overline{D_w}}{2}
    \cdot\prs{1+\frac{1}{2n}}.
\end{equation}
Similarly, if the second condition holds, then 
by considering a potential function
$\varphi$ such that $\varphi_x=\begin{cases}
1   &x=s,\\
0   &x=t,\\
\frac{1}{2} &\text{o.w.};
\end{cases}$,
we see that
\begin{equation}
    \Reff(s,t)^{-1}\leq \frac{\deg_w(s)+\deg_w(t)
    +2w(st)}{4}
    \leq \frac{\overline{D_w}}{2}
    \cdot\prs{1+\frac{1}{2n}}.
\end{equation}
Hence, in either case it follows that,
\begin{equation}\label{eq: calculations between upper and lower bound on resistances}
    \begin{split}
        \frac{\max_{x\neq y\in V}\Reff(x,y)}
        {\min_{x\neq y\in V}\Reff(x,y)}
         & \geq \frac{\frac{2}{\overline{D_w}\cdot
         \prs{1+\frac{1}{2n}}}}
        {\frac{2}{\overline{D_w}}\prs{1-\frac{1}{n}}} \\
         & \geq  \prs{1+\frac{1}{n-1}}\cdot\prs{
         1-\frac{1}{2n}}\\
         & = 1 + \frac{1}{n-1}-\frac{1}{2n}
         -\frac{1}{2n(n-1)}\\
         & = 1+\frac{1}{2(n-1)}.
    \end{split}
\end{equation}
as desired.
\end{proof}

Observe that Lemma \ref{lemma: stronger lower bound on special cases with simple conditions}
gives the following theorem as a  corollary.

\begin{theorem}\label{theorem: stronger lower bound on special cases}
Let $G=(V,E,w)$ be a graph with $|V|=n$
and $|E|=m<\binom{n}{2}$.
For every vertex $x\in V$,
denote $N_x=\abs{N(x)}$ and $\deg_w(x)=\sum_{y\in N(x)}w(xy)$.
Suppose that one of the following conditions holds.
\begin{enumerate}
    \item   The graph is regular, i.e. 
            there is $a>0$ such that
            for every $x\in V$, $N_x = a$.
    \item   All the weighted degrees are the same, 
            i.e.
            there is $b>0$ such that
            for every $x\in V$, $\deg_w(x) = b$.
    \item   All the edge weights are the same, i.e. 
            there is $c>0$ such that
            for every $xy\in E$, $w(xy)=c$.
\end{enumerate}
Then,
\begin{equation}
    \frac{\max_{x\neq y\in V}\Reff(x,y)}
    {\min_{x\neq y\in V}\Reff(x,y)}
    \geq 1+\frac{1}{O(n)}.
\end{equation}
\end{theorem}
We remark that condition \#1 in the above
is in fact Theorem \ref{intro: theorem: regular graphs cannot approximate resistance distance of clique}
presented in the introduction.
\begin{proof}
We will show that if one of the conditions in the theorem
holds, 
then there exists a non-edge pair $s,t\in V$
such that $\deg_w(s)+\deg_w(t)\leq 
2\overline{D_w}$
(i.e. condition \#2 in Lemma \ref{lemma: stronger lower bound on special cases with simple conditions}
holds), which will give the desired result.

First, assume that
\begin{equation}\label{eq: proof of generalized claim with covariance}
    \sum_{x\in V}\deg_w (x)\cdot N_x \geq \frac{4m\cdot w(E)}{n}.
\end{equation}
We will show later that if one of the conditions
in the theorem holds, then (\ref{eq: proof of generalized claim with covariance}) holds as well.
Denote by $F$ the set of all the missing edges in $G$,
namely $F=E\prs{K_n}\backslash E$.
Now, we can see that
when sampling a non-edge pair 
uniformly from $F$,
the following holds.
\begin{equation}
    \begin{split}
        \Ex{st\in F}\brs{\deg_w(s)+\deg_w(t)}
        & = \frac{1}{\binom{n}{2}-m}\prs{\sum_{x\in V}\deg_w(x)\prs{n-1-N_x}}\\
        & = \frac{1}{\binom{n}{2}-m}\prs{(n-1)\cdot 2w(E) - \sum_{x\in V}\deg_w(x)\cdot N_x}\\
        & \leq \frac{1}{\binom{n}{2}-m}\prs{\binom{n}{2}\cdot \frac{4w(E)}{n} - \frac{4m\cdot w(E)}{n}}\\
        & = 2\overline{D_w}.
    \end{split}
\end{equation}
Hence, there exists a pair
$s,t\in V$ such that 
$\deg_w(s)+\deg_w(t)\leq 2\overline{D_w}$
(since $w(st)=0$).
Thus, condition \#2 in Lemma \ref{lemma: stronger lower bound on special cases with simple conditions}
holds as claimed.

All we are left to show is that if one of the three conditions in the theorem
holds, then 
(\ref{eq: proof of generalized claim with covariance}) holds as well.
Indeed,
this can be seen as follows.
\begin{align*}
    \sum_{x\in V}\deg_w(x)\cdot N_x
    & \stackrel{?}{\geq} \frac{4m\cdot w(E)}{n}\\
    \iff 
    \frac{\sum_{x\in V}\deg_w(x)\cdot N_x}{n}
    & \stackrel{?}{\geq} 
    \frac{\sum_{y\in V}N_y}{n}\cdot \frac{\sum_{z\in V}\deg_w(z)}{n} \\
    \iff \Ex{x\in V}\brs{\deg_w(x)\cdot N_x}
    & \stackrel{?}{\geq}
    \Ex{y\in V}\brs{\deg_w(y)}\cdot\Ex{z\in V}\brs{N_z}\\
    \iff \textsc{Cov}\prs{\deg_w(x),N_x} & \stackrel{?}{\geq} 0\\
    \iff
    \Ex{x\in V}\brs{\prs{\deg_w(x)-\overline{D_w}}\prs{N_x-\overline{d}}}
    & \stackrel{?}{\geq} 0
\end{align*}
where we denote by $\overline{d}$ the average number of neighbors
of the vertices, i.e. 
\[\overline{d} = \frac{1}{n}\sum_{x\in V}N_x = \frac{2m}{n}.\]
Observe that in the first two cases, one of the random variables
($\deg_w(x)$ and $N_x$) is constant, and thus
$\prs{\deg_w(x)-\overline{D_w}}\prs{N_x-\overline{d}}=0$
for every vertex $x$.
In the third case, the two random variables
are equal up to scaling by a constant,
and we know $\Ex{}\brs{Z^2}\geq 0$.
Thus, we arrive at the desired result.
\end{proof}
We remark that Theorem \ref{theorem: stronger lower bound on special cases}
does not say that any graph with
maximal degree (number of neighbors)
$\leq (n-2)$ cannot guarantee
better than $1+\frac{1}{O(n)}$ approximation,
as not every such graph can be ``completed"
to form an $(n-2)$-regular graph.
However, it
does immediately give the following 
corollary.
\begin{corollary}
Let $G=(V,E,w)$ be a graph with $|V|=n$ where $n$ is even 
and with maximal degree $\Delta<\frac{n}{2}$.
Then,
\begin{equation}
    \frac{\max_{x\neq y\in V}\Reff(x,y)}
    {\min_{x\neq y\in V}\Reff(x,y)}
    \geq 1+\frac{1}{O(n)}
\end{equation}
\end{corollary}
\begin{proof}
Recall that by  Dirac's Theorem, 
every graph with minimal degree $\geq \frac{n}{2}$,
must be Hamiltonian.
Thus, we may apply it on the complement of $G$,
which leads to the conclusion that
there exists a complete matching
(complete since $n$ is even)
that does not belong to $G$.
Hence, we can refer to $G$
as an $(n-2)$-regular graph,
and by Theorem \ref{theorem: stronger lower bound on special cases} we are done.
\end{proof}
In addition, observe that in  Theorem \ref{theorem: stronger lower bound on special cases},
we essentially showed that
if (\ref{eq: proof of generalized claim with covariance})
holds, then the graph cannot guarantee
better than $\prs{1+\frac{1}{O(n)}}$-approximation
of the clique.
Similarly, we can show the following.
\begin{claim}
Let $G=(V,E,w)$ with $|V| = n$ vertices and $|E|=m$ edges.
For every vertex $x\in V$,
denote $N_x=\abs{N(x)}$
and $\deg_w(x)=\sum_{y\in N(x)}w(xy)$.
If
\begin{equation}\label{eq: variant where inner product is small enough}
    \sum_{x\in V}\deg_w(x)\cdot N_x
    \leq \frac{4m\cdot w(E)}{n}-2w(E),
\end{equation}
Then,
\begin{equation}
    \frac{\max_{x\neq y\in V}\Reff(x,y)}{\min_{x\neq y\in V}\Reff(x,y)}
    \geq 1+\frac{1}{O(n)}
\end{equation}
\end{claim}
\begin{proof}
Similarly to the proof of 
Theorem \ref{theorem: stronger lower bound on special cases},
we will show that there exists an edge
$xy\in E$ such that
$\deg_w(x)+\deg_w(y)+2w(xy)\leq
2\overline{D_w}$
(where 
$\overline{D_w} = \frac{1}{n}\sum_{x\in V}\deg_w(x) = \frac{2w(E)}{n}$),
which by Lemma  \ref{lemma: stronger lower bound on special cases with simple conditions}
will give the desired result.

By similar calculations
as in Theorem \ref{theorem: stronger lower bound on special cases}
we can see that 
by sampling an edge uniformly
the following holds.
\begin{equation}
    \begin{split}
        \Ex{xy\in E}\brs{\deg_w(x)+\deg_w(y)+2w(xy)}
        & = \frac{1}{m}\prs{ 2w(E)+\sum_{x\in V}\deg_w(x)\cdot N_x}\\
        & \leq \frac{1}{m}\cdot \frac{4m\cdot w(E)}{n}\\
        & = 2\overline{D_w}
    \end{split}
\end{equation}
Hence, 
again, by Lemma 
\ref{lemma: stronger lower bound on special cases with simple conditions}
we are done.
\end{proof}

\subsubsection{The Symmetric Case}
If we further add the assumption that
the graph is symmetric,
in the sense that all the edges that touch
$s$ and $t$ (where $st$ is the missing edge)
have the same weight $\alpha$,
and the rest of the edges have the same weight $\beta$,
then we can prove the stronger bound from Question
\ref{oq: does removing an edge from the clique implies a gap in the max/min resistance ratio (even after reweighting)?}.
Formally,
\begin{claim}\label{claim: lower bound on the symmetric case}
Let $G=(V,E,w)$ be a graph with $|E|=\binom{n}{2}-1$.
Let $s,t\in V$ be the pair with a missing edge between them,
and suppose that for every edge $e\in E$,
$w(e)=\begin{cases}
\alpha  &   e\text{ touches }s\text{ or }t,\\
\beta   &   \text{o.w.};
\end{cases}$
where $\alpha,\beta\in \R^+$.
Then,
\[
    \frac{\max_{x\neq y\in V}\Reff(x,y)}{\min_{x\neq y\in V}\Reff(x,y)}>1+\frac{1}{10n}.
\]
\end{claim}
We remark that this case is not necessarily contained
within any other case which we have presented so far.

We present here a proof via direct calculations.
We give an additional proof via the connection
between effective resistance and commute time
in Appendix \ref{appendix sectoin: Proof of Symmetric Case via Commute Time}.
\begin{proof}
Again,
denote by $A$ the set of all the edges
that touch $s$ or $t$,
and by $B$ denote the set of the edges that do not 
touch them.
By symmetry, for every $uv,u'v'\in A$
it holds that $\Reff(u,v)=\Reff(u',v')$,
and the same holds for every pair of edges in $B$.
We will denote by $R_A$ and $R_B$ the resistances
of the edges from $A$ and $B$ respectively.

Let us compute the effective resistance
for each pair according to the sets.
\begin{enumerate}
    \item   $\Reff(s,t)$:
            On the one hand we can suggest a flow
            that splits equally to the neighbors
            of $s$, and then each of the neighbors
            ships the same amount of flow directly to $t$.
            \begin{equation}
                \begin{split}
                    \Reff(s,t) & \leq
                    \frac{1}{(n-2)^2}\cdot\prs{
                    \sum_{x\in N(s)}\frac{1}{\alpha}
                    + \sum_{y\in N(t)}\frac{1}{\alpha}
                    }\\
                    & = \frac{2}{(n-2)\alpha}
                \end{split}
            \end{equation}
            On the other hand we can suggest a potential function
            $\varphi_x = \begin{cases}
            1   &   ,x=s\\
            0   &   ,x=t\\
            1/2   &   ,x\neq s,t
            \end{cases}$.
            \begin{equation}
                \begin{split}
                    \Reff(s,t)^{-1} & \leq
                    \sum_{xy\in E}
                    \alpha\prs{\varphi_x-\varphi_y}^2
                    \\
                    & \leq  2(n-2)\alpha\cdot \frac{1}{4}\\
                    & = \frac{(n-2)\alpha}{2}
                \end{split}
            \end{equation}
            Thus we conclude that
            \begin{equation}
                 \Reff(s,t)
                =
                \frac{2}{(n-2)\alpha}
            \end{equation}

\item   $R_B$:
        Take some vertices $u,v\in V\backslash\{s,t\}$.
        On the one hand we can suggest a flow
        that ships $\tau/2$ amount of flow to
        each of $s$ and $t$, and then ships
        the same amount from $s$ and $t$ to $v$.
        In addition, the flow will send $\sigma$
        on the edge $uv$,
        and then split the rest  of the flow equally 
        between every vertex
        $x\in V\backslash\{s,t,v\}$.
        Thus,
        \begin{equation}
            \begin{split}
                R_B & \leq
                \min_{\tau,\sigma\in [0,1]:\,
                \tau+\sigma\leq 1}
                \prs{2\cdot2\cdot\frac{\prs{\frac{\tau}{2}}^2}{\alpha}
                +\frac{\sigma^2}{\beta}
                +2(n-4)\frac{\prs{\frac{1-\tau-\sigma}{n-4}}^2}{\beta}
                }\\
                & =
                \min_{\tau,\sigma\in [0,1]:\,
                \tau+\sigma\leq 1}
                \prs{\frac{\tau^2}{\alpha}
                +\frac{\sigma^2}{\beta}
                +\frac{2}{(n-4)\beta}(1-\tau-\sigma)^2
                }
            \end{split}
        \end{equation}
        Let us denote the RHS by $f(\tau,\sigma)$, differentiate it
        with respect to $\sigma$, 
        and compare to 0 in order to find a minimum.
        \begin{align*}
            \frac{\partial f}{\partial \sigma}(\tau,\sigma) & = 0\\
            \iff \frac{2\sigma}{\beta}
                -\frac{4}{(n-4)\beta}(1-\tau-\sigma) & = 0\\
            \iff \prs{\frac{n-2}{n-4}}\sigma
            & = \frac{2}{n-4}(1-\tau)\\
            \iff \sigma
            & =\frac{2}{n-2}(1-\tau)
        \end{align*}
        Denote  $\sigma_0 = \frac{2}{n-2}(1-\tau)$,
        and let $f_{\sigma_0}(\tau)
        = f\prs{\tau,\sigma_0}$,
        thus,
        \begin{equation}
            \begin{split}
                f_{\sigma_0}(\tau) & = 
                \frac{\tau^2}{\alpha}
                +\frac{\prs{\frac{2}{n-2}(1-\tau)}^2}{\beta}
                +\frac{2}{(n-4)\beta}
                \prs{1-\tau-\frac{2}{n-2}(1-\tau)}^2\\
                & = \frac{\tau^2}{\alpha}
                +\frac{4}{(n-2)^2\beta}(1-\tau)^2
                +\frac{2}{(n-4)\beta}
                \prs{\prs{1-\tau}\prs{1-\frac{2}{n-2}}}^2\\
                & = \frac{\tau^2}{\alpha}
                +\frac{4}{(n-2)^2\beta}(1-\tau)^2
                +\frac{2(n-4)}{(n-2)^2\beta}
                \prs{1-\tau}^2\\
                & = \frac{\tau^2}{\alpha}
                +\frac{2}{(n-2)^2\beta}(1-\tau)^2\prs{
                2+n-4
                }\\
                & = \frac{\tau^2}{\alpha}
                +\frac{2}{(n-2)\beta}(1-\tau)^2
            \end{split}
        \end{equation}
        Again, let us now differentiate $f_{\sigma_0}$
        and compare to 0 in order to find a minimum.
        \begin{equation}
            \begin{split}
                f_{\sigma_0}'(\tau) & = 0\\
                \iff
                \frac{2}{\alpha}\tau
                -\frac{4}{(n-2)\beta}(1-\tau) & = 0\\
                \iff
                \tau\prs{\frac{1}{\alpha}+\frac{2}{(n-2)\beta}}
                & = \frac{2}{(n-2)\beta}\\
                \iff
                \tau & = \frac{1}{1+\frac{n-2}{2}\cdot \frac{\beta}{\alpha}}
            \end{split}
        \end{equation}
        Thus, we conclude that
        \begin{equation}
            \begin{split}
                R_B
                & \leq \frac{1}{\alpha}\cdot\prs{\frac{1}{1+\frac{n-2}{2}\cdot \frac{\beta}{\alpha}}}^2
                + \frac{2}{(n-2)\beta}\prs{\frac{\frac{n-2}{2}\cdot \frac{\beta}{\alpha}}{1+\frac{n-2}{2}\cdot \frac{\beta}{\alpha}}}^2\\
                & = \frac{\frac{1}{\alpha}\prs{1+\frac{n-2}{2}\cdot \frac{\beta}{\alpha}}}{\prs{1+\frac{n-2}{2}\cdot \frac{\beta}{\alpha}}^2}\\
                & = 
                \frac{1}{ \frac{n-2}{2}\cdot \beta +\alpha}
            \end{split}
        \end{equation}
        On the other hand we can suggest a potential function
        $\varphi_x = \begin{cases}
        1   &   ,x=u\\
        0   &   ,x=v\\
        1/2   &   ,x\neq u,v
        \end{cases}$.
        \begin{equation}
            \begin{split}
                R_B^{-1} & \leq
                \sum_{xy\in E}
                w(xy)\prs{\varphi_x-\varphi_y}^2
                \\
                & =\beta + 2\cdot 2\cdot \alpha\cdot\frac{1}{4} 
                + (2(n-4))\beta\cdot \frac{1}{4}\\
                & = \frac{1}{2}\cdot\prs{2\alpha+(n-2)\beta}
            \end{split}
        \end{equation}
        and thus
        \begin{equation}\label{eq: upper and lower bounds on R_B}
            R_B
            =
            \frac{1}{\alpha+\frac{n-2}{2}\beta}
        \end{equation}

\item   $R_A$:
         Recall that by Foster's Theorem we have
         \begin{equation}
             \sum_{xy\in E}w(xy)\Reff(x,y)=n-1
         \end{equation}
         Thus,
         \begin{align*}
             n-1 & = \sum_{xy\in A}\alpha R_A+\sum_{x'y'\in B}\beta R_B\\
             & = 2(n-2)\alpha R_A+\binom{n-2}{2}\beta R_B
         \end{align*}
         which leads to the conclusion that
         \begin{equation}
             \begin{split}
                 R_A & = \frac{(n-1)-\binom{n-2}{2}\beta R_B}{2(n-2)\alpha}\\
                 & =\frac{1}{2\alpha}\cdot\prs{
                 1+\frac{1}{n-2}-\frac{(n-3)}{2}\beta R_B}
             \end{split}
         \end{equation}
         Hence, by using  (\ref{eq: upper and lower bounds on R_B})
         we get that
         \begin{equation}
                 R_A  = \frac{1}{2\alpha}\cdot\prs{
                 1+\frac{1}{n-2}-  \frac{(n-3)\beta}{2\alpha+(n-2)\beta}
                }
         \end{equation}
\end{enumerate}
Assume towards contradiction that $\frac{\max_{x\neq y}\Reff(x,y)}{\min_{x\neq y}\Reff(x,y)}\leq 1+\frac{1}{10n}$.
On the one hand, by comparing $\Reff(s,t)$ and $R_B$
we can see that
\begin{align*}
    \prs{1+\frac{1}{10n}}& \geq \frac{\frac{2}{(n-2)\alpha}}{\frac{2}{2\alpha+(n-2)\beta}}\\
    & = \frac{2}{n-2}+\frac{\beta}{\alpha}\\
    \implies
    \frac{\beta}{\alpha} & \leq 
     1+\frac{1}{10n}-\frac{2}{n-2}
\end{align*}
On the other hand, by comparing $R_B$
and $R_A$, we see that
\begin{align*}
    \prs{1+\frac{1}{10n}} & \geq \frac{\frac{1}{\alpha+\frac{n-2}{2}\beta}}
    {\frac{1}{2\alpha}\cdot\prs{
    1+\frac{1}{n-2}-  \frac{(n-3)\beta}{2\alpha+(n-2)\beta}}}\\
    & = \frac{\frac{4\alpha}{2\alpha+(n-2)\beta}}
    {1+\frac{1}{n-2}-  \frac{(n-3)\beta}{2\alpha+(n-2)\beta}}\\
    & = \frac{4\alpha}
    {2\alpha+(n-2)\beta+\frac{2\alpha+(n-2)\beta}{n-2}-  (n-3)\beta}\\
    & = \frac{2}{1+\frac{1}{n-2}+\frac{\beta}{\alpha}}
\end{align*}
Thus,
\begin{align*}
    \frac{\beta}{\alpha}
    +\frac{1}{n-2}+1
    \geq 2\prs{1-\frac{1}{10n+1}}\\
    \iff
    \frac{\beta}{\alpha}\geq
    1-\frac{1}{n-2}-\frac{2}{10n+1}
\end{align*}
Hence, we conclude that
\begin{equation}
    1-\frac{1}{n-2}-\frac{2}{10n+1}
    \leq \frac{\beta}{\alpha}
    \leq  1-\frac{2}{n-2}+\frac{1}{10n}
\end{equation}
Which is a contradiction.
\end{proof}

\subsubsection{Discussion.}
We remark that (\ref{eq: variant where inner product is small enough})
can be viewed as follows.
\begin{equation}
    \begin{split}
        \sum_{x\in V}\deg_w(x)\cdot N_x
        & \leq \frac{4m\cdot w(E)}{n}-2w(E)\\
        \iff
        \sum_{x\in V}\frac{\deg_w(x)}{2w(E)}\cdot N_x
        & \leq \frac{2m}{n} - 1\\
        \iff
        \Ex{x\sim \deg_w(x)}\brs{N_x}
        & \leq \Ex{x\in V}\brs{N_x} -1
    \end{split}
\end{equation}
where by $x\sim \deg_w(x)$ we mean that a vertex
$x\in V$ is sampled with probability proportional
to its weighted degree.

Recall that as we mentioned earlier,
the proof of 
Theorem \ref{theorem: stronger lower bound on special cases}
essentially shows that 
if (\ref{eq: proof of generalized claim with covariance}) holds,
then the lower bound of $1+\frac{1}{O(n)}$ holds.
Note that (\ref{eq: proof of generalized claim with covariance})
can be viewed as follows.
\begin{equation}
    \begin{split}
        \sum_{x\in V}\deg_w(x)\cdot N_x
        & \geq \frac{4m\cdot w(E)}{n}\\
        \iff
        \sum_{x\in V}\frac{\deg_w(x)}{2w(E)}\cdot N_x
        & \geq \frac{2m}{n}\\
        \iff
        \Ex{x\sim \deg_w(x)}\brs{N_x}
        & \geq \Ex{x\in V}\brs{N_x}
    \end{split}
\end{equation}
A summary of our results is presented in Table \ref{tab:summary of lower bounds for resistance sparsifiers}.
\begin{table}[bht]
    \centering
    \begin{tabular}{|l|c|}
        \hline
         Condition & Lower Bound on $\frac{\max \Reff}{\min \Reff}$ \\
         \hline
         All graphs except the clique & $1+\frac{1}{O(n^2)}$\\
         \hline
         $\min_{x\in V}\deg_w(x)\leq \frac{\overline{D_w}}{2}
         \prs{1+\frac{1}{(1+O(1))n}}$
         & $1+\frac{1}{O(n)}$\\
         $\min_{x\neq y\in V}\prs{\deg_w(x)+\deg_w(y)+2w(xy)}\leq 2\overline{D_w}
         \prs{1+\frac{1}{(1+O(1))n}}$
         & $1+\frac{1}{O(n)}$\\
         $\Ex{x\sim \deg_w(x)}\brs{N_x}
         \geq \Ex{x\in V}\brs{N_x}$ & $1+\frac{1}{O(n)}$\\
         $\Ex{x\sim \deg_w(x)}\brs{N_x}
         \leq \Ex{x\in V}\brs{N_x}-1$  & $1+\frac{1}{O(n)}$\\
         Symmetric case & $1+\frac{1}{O(n)}$\\
         \hline
    \end{tabular}
    \caption{Summary of lower bounds for resistance sparsifiers
    of the clique.}
    \label{tab:summary of lower bounds for resistance sparsifiers}
\end{table}

\subsubsection{Upper Bound for Resistance Sparsifier of the Clique}\label{subsection - resistance sparsifier for clique upper bound}
In this subsection we complete the discussion about 
Question \ref{oq: is tilde O(n/eps) edges tight?}
by showing the upper bound of $O(n/\varepsilon)$
edges for resistance sparsifier of the clique,
as observed in \cite{dinitz2015towards}.
\begin{claim}\label{claim: clique is not such a tight example for resistance sparsifiers}
The  clique $K_n$
admits an $\varepsilon$-resistance sparsifier with $O\prs{n/\varepsilon}$
edges.
\end{claim}
Claim \ref{claim: clique is not such a tight example for resistance sparsifiers} follows from the fact that $\Reff(s,t)=\frac{2}{n}$
for all pairs of vertices in the clique (Fact \ref{fact: R_eff-2/n in K_n})
and the following theorem presented 
in \cite{von2014hitting}.
\begin{theorem}[Proposition 5 in \cite{von2014hitting}]\label{thm: proposition 5 from VLRH14 - R_eff is essentialy 1/deg+1/deg}
Let $G=(V,E,w)$ be a graph.
For every vertex $x\in V$
denote its weighted degree by $\deg_w(x)=\sum_{y\in N(x)}w(xy)$.
Denote the minimum weighted degree by $d_{min}$
and the maximal edge weight by $w_{max}$.
Denote by $\lambda_2(G)$
the second smallest eigenvalue of the normalized
Laplacian of $G$.
Then
\begin{equation}
    \abs{\Reff(s,t)-\prs{\frac{1}{\deg_w(s)}+\frac{1}{\deg_w(t)}}}
    \leq \frac{w_{max}}{d_{min}}\prs{\frac{1}{\lambda_2(G)}+2}\prs{\frac{1}{\deg_w(s)}+\frac{1}{\deg_w(t)}}
\end{equation}
\end{theorem}
Since for expanders $\lambda_2= \Omega(1)$,
an $\Theta\prs{\varepsilon^{-1}}$-regular expander with
all edges having the same weight $\Theta(\varepsilon\cdot  n)$
is an $\varepsilon$-resistance sparsifier for $K_n$.

\begin{fact}\label{fact: R_eff-2/n in K_n}
Let $G=K_n$ be a clique over $n$ vertices.
Then 
\[
    \forall s,t\in V,\quad
    \Reff(s,t) = \frac{2}{n}.
\]
\end{fact}
\begin{proof}
By symmetry and Foster's theorem,
\begin{equation}
    \Reff(s,t)
    = \frac{\sum_{xy\in E}\Reff(x,y)}{\binom{n}{2}}
    = \frac{n-1}{\binom{n}{2}} = \frac{2}{n}.
\end{equation}
\end{proof}

\subsection{Flow Metric Sparsifiers}\label{section - flow metric sparsifiers}

In this subsection we show that for some values of  $p$
there exists a $d_p$-sparsifier, where the sparsity of the 
graph depends on  $p$ and some additional parameters of the problem.
This is in fact Theorem \ref{intro: thm: main theorem - existence of d_p sparsifiers}, which we restate here for clarity.
\begin{theorem}\label{thm: main theorem - existence of d_p sparsifiers}
Let $G=(V,E,w)$ be a graph, fix $p\in \left(\frac{4}{3},\infty\right]$
having H\"older conjugate $q$,
and let $\varepsilon>0$.
Then there exists a graph $G'=(V,E',w')$ 
that is a $d_p$-sparsifier of $G$, i.e.
\begin{equation}\label{eq: sparsifier concentration guarantee in our main thm}
    \forall s,t\in V,\quad
    d_{p,G'}(s,t)\in\prs{1\pm \varepsilon}d_{p,G}(s,t),
\end{equation}
and has $|E'|=f(n,\varepsilon,p)$ edges,
where
\begin{equation}
    f(n,\varepsilon,p) = 
    \begin{cases}
        n-1     & \text{if }p=\Omega\prs{\varepsilon^{-1} \log n},\\
        O\prs{n\log(n/\varepsilon)\prs{\log\log(n/\varepsilon)}^2 \varepsilon^{-2}}& 
        \text{if }2<p<\infty,\\
        O\prs{n^{q/2}\log(n)\log(1/\varepsilon)\varepsilon^{-5}} & 
        \text{if }\frac{4}{3} < p < 2.
    \end{cases}
\end{equation}
\end{theorem}
We remark that in general, the last row of the table
applies for all $1<p<2$, but for $1<p\leq \frac{4}{3}$
it gives trivial bounds since this implies that $q>4$
and thus we have $n^{q/2}>n^2$.
Moreover, we remark that as $p$ tends to $\infty$,
the $d_p$ metric tends to the inverse of $\mincut(s,t)$ (ultra)metric,
for which there exists cut sparsifiers
\cite{karger1993global, benczur1996approximating},
that preserve \textit{all} of the cuts in the graph,
and for which there exists a lower bound
\cite{carlson2019optimalLowerBoundForSketchingCuts}
of $\Omega(n/\varepsilon^2)$ edges.
However in our case, in order to preserve the $d_\infty$
metric, it suffices to preserve only 
the minimum-$st$-cuts.
In addition, note that the case $p=2$ is the special case where
$d_2^2$ is in fact the resistance distance, for which \cite{chu2020shortCycleLongPaper}
showed the existence of a resistance sparsifier
with $\widetilde{O}\prs{n\varepsilon^{-1}}$ edges.

We now turn to proving Theorem
\ref{thm: main theorem - existence of d_p sparsifiers}
which follows easily from a theorem
of Cohen and Peng \cite{cohen2015lp}.
\begin{theorem}[Theorem 7.1 in \cite{cohen2015lp}]\label{thm: thm 7.1 from cohen and peng - matrix concentration by lp row sampling}
Given a matrix $A\in \R^{m\times n}$ and parameters $q\in (1,\infty),
\varepsilon>0$, there exists a set of scores $\curprs{\tau_i(A,q)}_{i=1}^m$
summing up to at most $n$,
such that for any set of sampling values $\curprs{\sigma_i}_{i=1}^m$
satisfying
\[
    \sigma_i \geq \tau_i(A,q)\cdot g(n,\varepsilon, q)
\]
if we generate a matrix $S$ with $N=\sum_{i=1}^m\sigma_i$
rows, each chosen independently as $\frac{1}{\sigma_i^{1/q}}\cdot \overrightarrow{e_i}$ with probability $\frac{\sigma_i}{N}$
($\overrightarrow{e_i}\in\R^m$ is the $i^{th}$ basis vector),
then with probability at least $1-\frac{1}{n^{\Omega(1)}}$ 
we have
\[
    \forall \varphi\in \R^n,\quad
    \norm{SA\varphi}_q\in (1\pm \varepsilon)\norm{A\varphi}_q
\]
where
\begin{equation}
    g(n,\varepsilon,q) = 
    \begin{cases}
        \log(n) \varepsilon^{-2}     &
        \text{if }q=1,\\
       \log(n/\varepsilon)\prs{\log(\log(n/\varepsilon))}^2 \varepsilon^{-2}& 
        \text{if }1<q<2,\\
        n^{\frac{q}{2}-1}\log(n)\log(1/\varepsilon)\varepsilon^{-5} & 
        \text{if }2 < q.
    \end{cases}
\end{equation}
\end{theorem}
Note that for the proof of existence,
we can set $\sigma_i = \tau_i(A,q)\cdot g(n,\varepsilon, q)$
and thus
a sufficient number of rows in the matrix $S$
will satisfy $N=\sum_{i=1}^m\tau_i(A,q)g(n,\varepsilon,q) \leq n\cdot g(n,\varepsilon,q)$.
We are ready to present the proof of Theorem
\ref{thm: main theorem - existence of d_p sparsifiers}.

\begin{proof}(of Theorem \ref{thm: main theorem - existence of d_p sparsifiers})
We will first deal with the last two rows in the table of the theorem
($\frac{4}{3} < p < \infty$).
Let $A=WB$ where $W_{m\times m}$ is the diagonal weight matrix
and $B_{m\times n}$ is the signed edge-vertex incidence matrix of $G$.
By Theorem \ref{thm: thm 7.1 from cohen and peng - matrix concentration by lp row sampling},
there exists a matrix $S$ with $N$ rows
(same $N$ as defined in Theorem 
\ref{thm: thm 7.1 from cohen and peng - matrix concentration by lp row sampling})
where each row is a reweighted basis vector,
such that
\begin{equation}\label{eq: concentration guarantee from thm 7.1 cp15}
    \forall \varphi\in \R^n,\quad
    \norm{SWB\varphi}_q\in (1\pm \varepsilon)\norm{WB\varphi}_q
\end{equation}
with $q$ being the H\"older conjugate of $p$.
Note that $q=\frac{p}{p-1}$, and thus the case 
$2 < p < \infty$ corresponds to the case $1<q<2$,
and the case $2< q$ corresponds to the case $1<p<2$.

Let $W'=\prs{S^TS}^{1/2}W$, and note that 
it is a  diagonal 
matrix of dimensions $m\times m$
where the entries on the diagonal are 
in fact the weights of the edges corresponding to the entries
in $SW$.
Let $G'$ be the graph defined by the weights $W'$,
and let $d_p'$ be the flow metric on $G'$.
It is important to note that since $S$ consists of
$N$ rows, then $W'$ has at most $N$ 
non-zero entries,
and moreover it is easy to see that
for any $\varphi\in \R^n$
it holds that $\norm{W'B\varphi}_q=\norm{SWB\varphi}_q$.

Next, recall that by Claim \ref{claim: d_p = 1 / bar(d)_p}
\begin{equation}
    \forall s\neq t\in V,\quad
    d_p(s,t)
    = \prs{\min_{\varphi_s-\varphi_t=1}\norm{WB\varphi}_q}^{-1}.
\end{equation}
and similarly for $d'_p$ using $W'$.

Now, we can see that
\begin{equation}
    \begin{split}
    d'_p(s,t)^{-1} &= \min_{\varphi_s-\varphi_t=1}\norm{W'B\varphi}_q\\
    & \leq (1+\varepsilon)\min_{\varphi_s-\varphi_t=1}\norm{WB\varphi}_q\\
    & =  (1+\varepsilon)d_p(s,t)^{-1}
    \end{split}
\end{equation}
where the inequalities follow because
we can take a minimizer $\varphi^*$
of 
$\norm{WB\varphi}_q$,
apply (\ref{eq: concentration guarantee from thm 7.1 cp15}) on it, and 
conclude an upper bound on the minimum of $\norm{W'B\varphi}_q$.
The other direction is similar,
and thus we can conclude that $G'$ is a $d_p$ sparsifier of $G$ that satisfies the guarantee of
(\ref{eq: sparsifier concentration guarantee in our main thm}).

Finally, recalling that $N \leq n\cdot g(n,\varepsilon,q) $, gives the desired bound on the number of the edges
in $G'$.

Next, for the first row in the table (the case $p=\Omega\prs{\varepsilon^{-1}\logn}$) we recall
that by Corollary \ref{corl: large p is essentialy infinity},
\begin{equation}\label{cor-eq: d_p and d_infty are the same for large enough p}
    \forall s,t\in V,\quad
    d_\infty(s,t)\leq d_p(s,t)\leq (1+\varepsilon)d_\infty(s,t).
\end{equation}
Thus, we choose $G'$ to be the Gomory-Hu tree of $G$,
and clearly Corollary \ref{corl: large p is essentialy infinity}
and (\ref{cor-eq: d_p and d_infty are the same for large enough p})
can be applied also to $d'_p$.
It is easy to see that this is indeed a $d_p$ sparsifier for $G$,
since we have 
\begin{align*}
        d'_p(s,t) & \leq (1+\varepsilon) d'_\infty(s,t)
        && (\text{by }(\ref{cor-eq: d_p and d_infty are the same for large enough p})
        \text{ on }G')\\
        & = (1+\varepsilon) d_\infty(s,t) 
        && {\text{by GH guarantee}}\\
        & \leq (1+\varepsilon) d_p(s,t) 
        && (\text{by }(\ref{cor-eq: d_p and d_infty are the same for large enough p})
        \text{ on }G)
\end{align*}
The other direction is similar, and this completes the proof of Theorem \ref{thm: main theorem - existence of d_p sparsifiers}.
\end{proof}

\subsection{Transforms that Preserve the Flow Metrics, and Those that do not Exist}\label{section - transforms for flow metrics}

In this subsection, we  present transformations that
reduce the number of edges/vertices in the graph in some cases,
while preserving the flow metrics on them.
These transformations are closely related to
known  transformations for
effective resistance.
We begin with reductions of parallel edges 
and of sequential edges,
that are natural extensions of corresponding
reductions for effective resistance.
We then proceed to discuss the well known $Y$-$\Delta$
transform for effective resistance,
as well as the more general $k$-star-mesh transform,
and examine for which values of $k$ and $p$
there exists an analogue of it for $d_p$.

\subsubsection{Sequential Edges Reduction}
Think of the case presented  in figure
\ref{fig:Sequential edges to single edge transform.},
where 
a vertex $x$ of degree 2 is
connected with an edge $e_1$ of weight $\alpha$
to a vertex $a$,
and with an edge $e_2$ of weight $\beta$
to another vertex  $b\neq a$.
We wish to remove the vertex $x$ and find a weight $\gamma=\gamma(\alpha,\beta)$ for a new edge
that will now connect  $a$ and $b$,
and will preserve the flow metric of the graph.

\begin{figure}[htb!]
    \centering

\tikzset{every picture/.style={line width=0.75pt}} 

\begin{tikzpicture}[x=0.75pt,y=0.75pt,yscale=-1,xscale=1]

\draw  [fill={rgb, 255:red, 74; green, 90; blue, 226 }  ,fill opacity=1 ] (120,111.53) .. controls (120,106.27) and (124.27,102) .. (129.53,102) .. controls (134.79,102) and (139.06,106.27) .. (139.06,111.53) .. controls (139.06,116.79) and (134.79,121.06) .. (129.53,121.06) .. controls (124.27,121.06) and (120,116.79) .. (120,111.53) -- cycle ;
\draw  [fill={rgb, 255:red, 208; green, 2; blue, 27 }  ,fill opacity=1 ] (490,109.53) .. controls (490,104.27) and (494.27,100) .. (499.53,100) .. controls (504.79,100) and (509.06,104.27) .. (509.06,109.53) .. controls (509.06,114.79) and (504.79,119.06) .. (499.53,119.06) .. controls (494.27,119.06) and (490,114.79) .. (490,109.53) -- cycle ;
\draw  [fill={rgb, 255:red, 74; green, 90; blue, 226 }  ,fill opacity=1 ] (122,248.53) .. controls (122,243.27) and (126.27,239) .. (131.53,239) .. controls (136.79,239) and (141.06,243.27) .. (141.06,248.53) .. controls (141.06,253.79) and (136.79,258.06) .. (131.53,258.06) .. controls (126.27,258.06) and (122,253.79) .. (122,248.53) -- cycle ;
\draw  [fill={rgb, 255:red, 208; green, 2; blue, 27 }  ,fill opacity=1 ] (494,248.53) .. controls (494,243.27) and (498.27,239) .. (503.53,239) .. controls (508.79,239) and (513.06,243.27) .. (513.06,248.53) .. controls (513.06,253.79) and (508.79,258.06) .. (503.53,258.06) .. controls (498.27,258.06) and (494,253.79) .. (494,248.53) -- cycle ;
\draw [line width=3]    (141.06,248.53) -- (494,248.53) ;
\draw  [fill={rgb, 255:red, 255; green, 255; blue, 255 }  ,fill opacity=1 ] (15.06,179.49) .. controls (15.06,156.33) and (45.92,137.56) .. (84,137.56) -- (84,158.24) .. controls (45.92,158.24) and (15.06,177.01) .. (15.06,200.17) ;\draw  [fill={rgb, 255:red, 255; green, 255; blue, 255 }  ,fill opacity=1 ] (15.06,200.17) .. controls (15.06,217.37) and (32.07,232.15) .. (56.42,238.62) -- (56.42,245.51) -- (84,231.76) -- (56.42,211.04) -- (56.42,217.93) .. controls (32.07,211.46) and (15.06,196.68) .. (15.06,179.49)(15.06,200.17) -- (15.06,179.49) ;
\draw  [fill={rgb, 255:red, 126; green, 211; blue, 33 }  ,fill opacity=1 ] (305,110.53) .. controls (305,105.27) and (309.27,101) .. (314.53,101) .. controls (319.79,101) and (324.06,105.27) .. (324.06,110.53) .. controls (324.06,115.79) and (319.79,120.06) .. (314.53,120.06) .. controls (309.27,120.06) and (305,115.79) .. (305,110.53) -- cycle ;
\draw    (139.06,111.53) -- (305,110.53) ;
\draw    (324.06,110.53) -- (490,109.53) ;

\draw (98,104.4) node [anchor=north west][inner sep=0.75pt]  [font=\large]  {$a$};
\draw (527,100.4) node [anchor=north west][inner sep=0.75pt]  [font=\large]  {$b$};
\draw (183,78.4) node [anchor=north west][inner sep=0.75pt]    {$w( e_{1}) =\alpha $};
\draw (367,79.4) node [anchor=north west][inner sep=0.75pt]    {$w( e_{2}) =\beta $};
\draw (100,241.4) node [anchor=north west][inner sep=0.75pt]  [font=\large]  {$a$};
\draw (529,237.4) node [anchor=north west][inner sep=0.75pt]  [font=\large]  {$b$};
\draw (282,210.4) node [anchor=north west][inner sep=0.75pt]    {$w( e') =\gamma $};
\draw (309,122.4) node [anchor=north west][inner sep=0.75pt]  [font=\large]  {$x$};

\end{tikzpicture}
    \caption{Sequential edges to single edge transform.}
    \label{fig:Sequential edges to single edge transform.}
\end{figure}
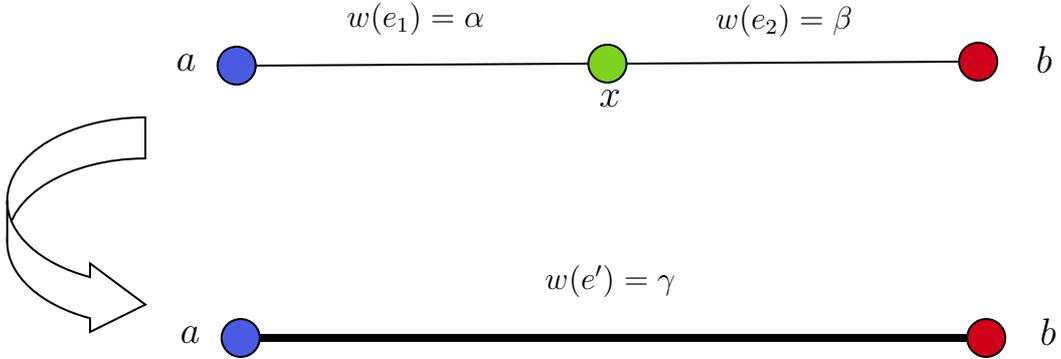

\begin{claim}
Let $p\in[1,\infty]$.
Let $G=(V,E,w)$ be a graph that contains a vertex $x$
that is incident to exactly 2 edges $e_1=\{x,a\}$ and $e_2=\{x,b\}$
of weights $\alpha$ and $\beta$ respectively.
Let  $G'=(V',E',w')$ be a graph
with $V'=V\backslash\{x\}, E'=E\backslash\{e_1,e_2\}\cup\curprs{\{a,b\}}$
and
\[
    w'(e)=\begin{cases}
    \frac{1}{\prs{\frac{1}{\alpha^p}+\frac{1}{\beta^p}}^{1/p}}
    & \text{if }e=\{a,b\},\\
    w(e)    &   \text{o.w.};
    \end{cases}
\]
Then, for every $s,t\in V'$,
$d_{p,G}(s,t)=d_{p,G'}(s,t)$.
\end{claim}
Before we begin the proof,
note that for $p=1$, we have
$\frac{1}{\gamma}    =\frac{1}{\alpha}+\frac{1}{\beta}$,
which is the desired behavior since 
 $d_1$ coincides with the shortest-path metric in the graph
with inverse edge weights.
Moreover, 
note that as $p\rightarrow\infty$,
it holds that $\gamma\rightarrow\min\curprs{\alpha,\beta}$,
which is the desired behavior for the minimum cut in the graph.

\begin{proof}
Note that in the graph $G$,
for every amount of flow $\tau$
that flows from $a$ to $x$,
the contribution
of the flow (without the $1/p$ power)
over the edges that connect
$a$ to $b$ in this case is exactly
$\abs{\frac{\tau}{\alpha}}^p+\abs{\frac{\tau}{\beta}}^p
=|\tau|^p\cdot\prs{\frac{1}{\alpha^p}+\frac{1}{\beta^p}}$,
where in $G'$ it will be $\abs{\frac{\tau}{\gamma}}^p$.
So setting 
\[
    \gamma=\frac{1}{\prs{\frac{1}{\alpha^p}+\frac{1}{\beta^p}}^{1/p}}
\]
indeed satisfies our demands,
since it holds that
$\frac{1}{\gamma^p} = \frac{1}{\alpha^p}+\frac{1}{\beta^p}$.
\end{proof}

\subsubsection{Parallel Edges Reduction}
Think of the case presented  in figure
\ref{fig:Parallel edges to single edge transform},
where two vertices $a$ and $b$
are connected with two edges $e_1$ and  $e_2$ 
of some weight $\alpha$ and $\beta$ respectively.
Again, we wish to find a weight 
$\gamma=\gamma(\alpha,\beta)$
that will replace the two edges and preserve
the flow metric in the graph.

\begin{figure}[htb!]
    \centering

\tikzset{every picture/.style={line width=0.75pt}} 

\begin{tikzpicture}[x=0.75pt,y=0.75pt,yscale=-1,xscale=1]

\draw  [fill={rgb, 255:red, 74; green, 90; blue, 226 }  ,fill opacity=1 ] (120,111.53) .. controls (120,106.27) and (124.27,102) .. (129.53,102) .. controls (134.79,102) and (139.06,106.27) .. (139.06,111.53) .. controls (139.06,116.79) and (134.79,121.06) .. (129.53,121.06) .. controls (124.27,121.06) and (120,116.79) .. (120,111.53) -- cycle ;
\draw  [fill={rgb, 255:red, 208; green, 2; blue, 27 }  ,fill opacity=1 ] (492,111.53) .. controls (492,106.27) and (496.27,102) .. (501.53,102) .. controls (506.79,102) and (511.06,106.27) .. (511.06,111.53) .. controls (511.06,116.79) and (506.79,121.06) .. (501.53,121.06) .. controls (496.27,121.06) and (492,116.79) .. (492,111.53) -- cycle ;
\draw [color={rgb, 255:red, 0; green, 0; blue, 0 }  ,draw opacity=1 ][line width=2.25]    (139.06,111.53) .. controls (139.06,41.56) and (492.06,40.56) .. (492,111.53) ;
\draw [color={rgb, 255:red, 0; green, 0; blue, 0 }  ,draw opacity=1 ][line width=2.25]    (139.06,111.53) .. controls (141.06,191.56) and (490.06,191.56) .. (492,111.53) ;
\draw  [fill={rgb, 255:red, 74; green, 90; blue, 226 }  ,fill opacity=1 ] (122,248.53) .. controls (122,243.27) and (126.27,239) .. (131.53,239) .. controls (136.79,239) and (141.06,243.27) .. (141.06,248.53) .. controls (141.06,253.79) and (136.79,258.06) .. (131.53,258.06) .. controls (126.27,258.06) and (122,253.79) .. (122,248.53) -- cycle ;
\draw  [fill={rgb, 255:red, 208; green, 2; blue, 27 }  ,fill opacity=1 ] (494,248.53) .. controls (494,243.27) and (498.27,239) .. (503.53,239) .. controls (508.79,239) and (513.06,243.27) .. (513.06,248.53) .. controls (513.06,253.79) and (508.79,258.06) .. (503.53,258.06) .. controls (498.27,258.06) and (494,253.79) .. (494,248.53) -- cycle ;
\draw [line width=3]    (141.06,248.53) -- (494,248.53) ;
\draw  [fill={rgb, 255:red, 255; green, 255; blue, 255 }  ,fill opacity=1 ] (15.06,179.49) .. controls (15.06,156.33) and (45.92,137.56) .. (84,137.56) -- (84,158.24) .. controls (45.92,158.24) and (15.06,177.01) .. (15.06,200.17) ;\draw  [fill={rgb, 255:red, 255; green, 255; blue, 255 }  ,fill opacity=1 ] (15.06,200.17) .. controls (15.06,217.37) and (32.07,232.15) .. (56.42,238.62) -- (56.42,245.51) -- (84,231.76) -- (56.42,211.04) -- (56.42,217.93) .. controls (32.07,211.46) and (15.06,196.68) .. (15.06,179.49)(15.06,200.17) -- (15.06,179.49) ;

\draw (98,104.4) node [anchor=north west][inner sep=0.75pt]  [font=\large]  {$a$};
\draw (527,101.4) node [anchor=north west][inner sep=0.75pt]  [font=\large]  {$b$};
\draw (278,24.4) node [anchor=north west][inner sep=0.75pt]    {$w( e_{1}) =\alpha $};
\draw (278,139.4) node [anchor=north west][inner sep=0.75pt]    {$w( e_{2}) =\beta $};
\draw (100,241.4) node [anchor=north west][inner sep=0.75pt]  [font=\large]  {$a$};
\draw (529,238.4) node [anchor=north west][inner sep=0.75pt]  [font=\large]  {$b$};
\draw (282,210.4) node [anchor=north west][inner sep=0.75pt]    {$w( e') =\gamma $};

\end{tikzpicture}
    \caption{Parallel edges to single edge transform.}
    \label{fig:Parallel edges to single edge transform}
\end{figure}
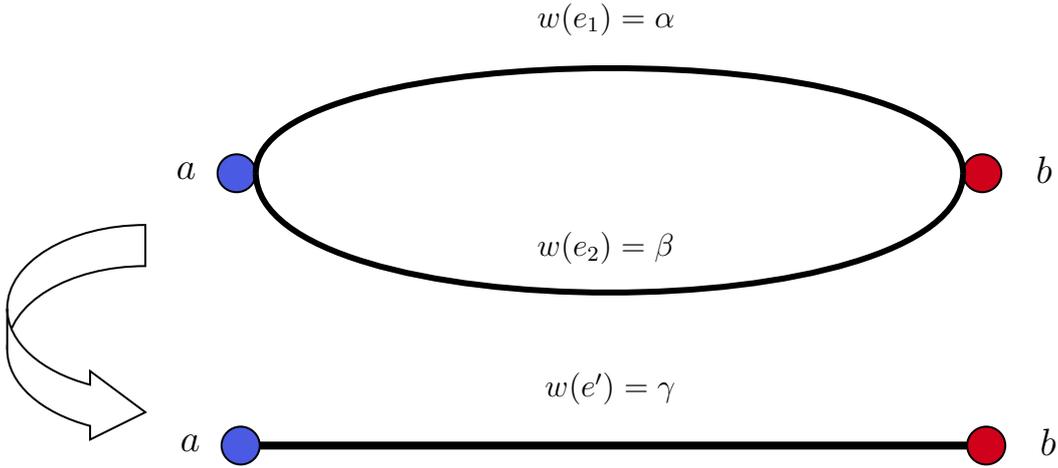

\begin{claim}\label{claim: parallel edges reduction for d_p}
Let $p\in[1,\infty]$, and let $q$
be the H\"older conjugate of $p$.
Let $G=(V,E,w)$ be a graph that contains two
parallel edges $e_1$ and $e_2$
of weights $\alpha$ and $\beta$ respectively
between some vertices $a$ and $b$.
Let $G'=(V,E',w')$ be a graph with
with $E'=E\backslash\{e_1\}$
and 
\[
    w'(e)=\begin{cases}
    \prs{\alpha^q+\beta^q}^{1/q}
    & \text{if }e=e_2,\\
    w(e)    &  \text{o.w.};
    \end{cases}
\]
Then, for every $s,t\in V'$,
$d_{p,G}(s,t)=d_{p,G'}(s,t)$.
\end{claim}

Let us examine what
happens to $\gamma$ as $p\rightarrow\infty$.
In this case, $q\rightarrow 1$, and thus
$\gamma\xrightarrow{p\rightarrow\infty}\alpha+\beta$,
which is  the desired behavior for minimum cut.
Moreover, note that as $p\rightarrow 1$
then $q\rightarrow\infty$.
Then $\gamma\xrightarrow{p\rightarrow 1}\max\curprs{\alpha,\beta}$,
and hence $\frac{1}{\gamma}\xrightarrow{p\rightarrow 1}
\min\curprs{\frac{1}{\alpha},\frac{1}{\beta}}$,
which is again the desired behavior since
$d_1$ coincides with
the shortest-path metric in the graph with inverse edge weights.
       
We present here a proof via the dual problem.
We give an additional proof via flows 
in Appendix \ref{appendix section - Another proof for the parallel edges reduction via flows}.
\begin{proof}
We  consider the dual problem.
Recall that $d_p(s,t)=\prs{\Bar{d}_p(s,t)}^{-1}$
(Claim \ref{claim: d_p = 1 / bar(d)_p}),
where
\[
    \Bar{d}_p(s,t)=\min_{\varphi_s-\varphi_t=1}
    \norm{WB\varphi}_q.
\]
Thus, we can finish the proof by showing that the $\Bar{d}_p$
metric is preserved in the new graph $G'$.
Suppose we have potentials $\varphi_a,\varphi_b$,
and we wish to find $\gamma$ such that the potential
difference between $a$ and $b$ is preserved, i.e.
\[
    \gamma^q\cdot\abs{\varphi_a-\varphi_b}^q
    = \alpha^q\cdot\abs{\varphi_a-\varphi_b}^q
    +\beta^q\cdot\abs{\varphi_a-\varphi_b}^q
    =\prs{ \alpha^q+\beta^q}\cdot\abs{\varphi_a-\varphi_b}^q.
\]
and thus it is easy to see that setting
$\gamma=\prs{\alpha^q+\beta^q}^{1/q}$
will satisfy our requirement.
\end{proof}

\subsubsection{Non-Existence of Y-\texorpdfstring{$\Delta$}{Delta} Transform}\label{section: non existence of delta-wye transform}
Suppose that a graph
has a vertex of degree 3, 
and we wish to remove it
while preserving the $d_p$
metric between the remaining
vertices,
thus obtaining a smaller equivalent instance
to work with.
The way to do this
for the resistance distance is
via the well known Y-$\Delta$
transform
(as shown in figure \ref{fig:Y-Delta-transform}),
and we wish to generalize it to all flow metrics.
It is known that 
such a transformation  exists for
$p=1$ (shortest-path metric)
and for $p=\infty$ (specifically for the case
where the middle vertex has degree 3 - see
Theorem 1 in \cite{chaudhuri2000mimickingNetworks}, 
we elaborate on this in subsection \ref{subsection: p=infinity}).
In addition, it is important to note that the Y-$\Delta$
transform for effective resistance,
depends solely on the weights of the 
edges incident to the vertex of degree 3,
i.e. it is a local transformation
that does not depend on the rest of the graph.
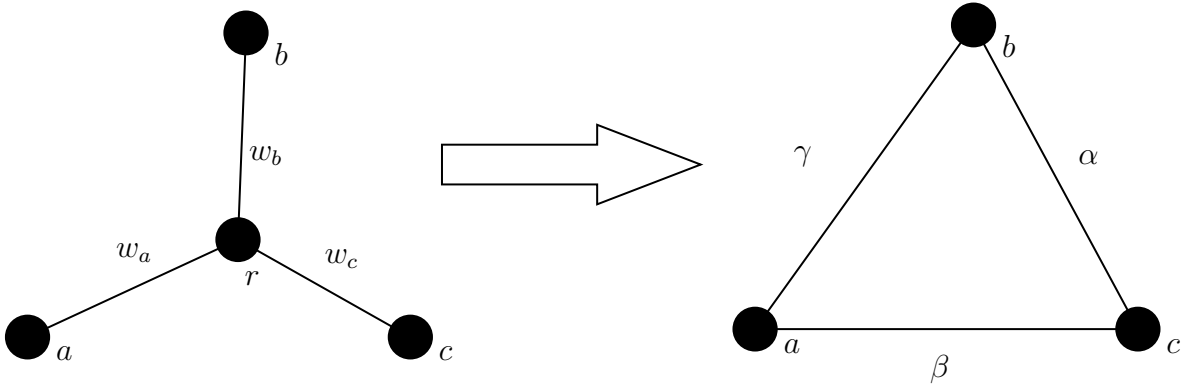
\begin{figure}[hbt]
    \centering

\tikzset{every picture/.style={line width=0.75pt}} 

\begin{tikzpicture}[x=0.75pt,y=0.75pt,yscale=-1,xscale=1]

\draw  [fill={rgb, 255:red, 0; green, 0; blue, 0 }  ,fill opacity=1 ] (51.44,200.78) .. controls (51.44,194.83) and (56.27,190) .. (62.22,190) .. controls (68.17,190) and (73,194.83) .. (73,200.78) .. controls (73,206.73) and (68.17,211.56) .. (62.22,211.56) .. controls (56.27,211.56) and (51.44,206.73) .. (51.44,200.78) -- cycle ;
\draw  [fill={rgb, 255:red, 0; green, 0; blue, 0 }  ,fill opacity=1 ] (160.44,47.78) .. controls (160.44,41.83) and (165.27,37) .. (171.22,37) .. controls (177.17,37) and (182,41.83) .. (182,47.78) .. controls (182,53.73) and (177.17,58.56) .. (171.22,58.56) .. controls (165.27,58.56) and (160.44,53.73) .. (160.44,47.78) -- cycle ;
\draw  [fill={rgb, 255:red, 0; green, 0; blue, 0 }  ,fill opacity=1 ] (242.44,200.78) .. controls (242.44,194.83) and (247.27,190) .. (253.22,190) .. controls (259.17,190) and (264,194.83) .. (264,200.78) .. controls (264,206.73) and (259.17,211.56) .. (253.22,211.56) .. controls (247.27,211.56) and (242.44,206.73) .. (242.44,200.78) -- cycle ;
\draw  [fill={rgb, 255:red, 0; green, 0; blue, 0 }  ,fill opacity=1 ] (156.44,151.78) .. controls (156.44,145.83) and (161.27,141) .. (167.22,141) .. controls (173.17,141) and (178,145.83) .. (178,151.78) .. controls (178,157.73) and (173.17,162.56) .. (167.22,162.56) .. controls (161.27,162.56) and (156.44,157.73) .. (156.44,151.78) -- cycle ;
\draw  [fill={rgb, 255:red, 0; green, 0; blue, 0 }  ,fill opacity=1 ] (414.44,196.78) .. controls (414.44,190.83) and (419.27,186) .. (425.22,186) .. controls (431.17,186) and (436,190.83) .. (436,196.78) .. controls (436,202.73) and (431.17,207.56) .. (425.22,207.56) .. controls (419.27,207.56) and (414.44,202.73) .. (414.44,196.78) -- cycle ;
\draw  [fill={rgb, 255:red, 0; green, 0; blue, 0 }  ,fill opacity=1 ] (523.44,43.78) .. controls (523.44,37.83) and (528.27,33) .. (534.22,33) .. controls (540.17,33) and (545,37.83) .. (545,43.78) .. controls (545,49.73) and (540.17,54.56) .. (534.22,54.56) .. controls (528.27,54.56) and (523.44,49.73) .. (523.44,43.78) -- cycle ;
\draw  [fill={rgb, 255:red, 0; green, 0; blue, 0 }  ,fill opacity=1 ] (605.44,196.78) .. controls (605.44,190.83) and (610.27,186) .. (616.22,186) .. controls (622.17,186) and (627,190.83) .. (627,196.78) .. controls (627,202.73) and (622.17,207.56) .. (616.22,207.56) .. controls (610.27,207.56) and (605.44,202.73) .. (605.44,196.78) -- cycle ;
\draw    (167.22,151.78) -- (253.22,200.78) ;
\draw    (167.22,151.78) -- (62.22,200.78) ;
\draw    (167.22,151.78) -- (171.22,47.78) ;
\draw    (425.22,196.78) -- (534.22,43.78) ;
\draw    (616.22,196.78) -- (534.22,43.78) ;
\draw    (616.22,196.78) -- (425.22,196.78) ;
\draw   (269,104) -- (346.43,104) -- (346.43,94) -- (398.06,114) -- (346.43,134) -- (346.43,124) -- (269,124) -- cycle ;

\draw (75,204.18) node [anchor=north west][inner sep=0.75pt]    {$a$};
\draw (184,51.18) node [anchor=north west][inner sep=0.75pt]    {$b$};
\draw (266,204.18) node [anchor=north west][inner sep=0.75pt]    {$c$};
\draw (629,200.18) node [anchor=north west][inner sep=0.75pt]    {$c$};
\draw (438,200.18) node [anchor=north west][inner sep=0.75pt]    {$a$};
\draw (547,47.18) node [anchor=north west][inner sep=0.75pt]    {$b$};
\draw (169.22,165.96) node [anchor=north west][inner sep=0.75pt]    {$r$};
\draw (105,152.18) node [anchor=north west][inner sep=0.75pt]    {$w_a$};
\draw (171.22,103.18) node [anchor=north west][inner sep=0.75pt]    {$w_b$};
\draw (209,155.18) node [anchor=north west][inner sep=0.75pt]    {$w_c$};
\draw (585,105.18) node [anchor=north west][inner sep=0.75pt]    {$\alpha$};
\draw (511,208.18) node [anchor=north west][inner sep=0.75pt]    {$\beta$};
\draw (443,103.18) node [anchor=north west][inner sep=0.75pt]    {$\gamma$};

\end{tikzpicture}
    \caption{Transforming ``Y" (or a 3-star)
    into a ``$\Delta$" (a triangle)
    by deleting the middle vertex $r$,
    and creating new edges while preserving the
    effective resistance on the graph.
    In the general case, an $n$-star transforms 
    into a clique $K_n$.
    }
    \label{fig:Y-Delta-transform}
\end{figure}

In the case of effective resistance,
the transformation is as follows
(the notations of the weights are as
presented in figure
\ref{fig:Y-Delta-transform}).
\begin{align*}
    \alpha&=\frac{w_b\cdot w_c}{w_a+w_b+w_c},
    &
    \beta&=\frac{w_a\cdot w_c}{w_a+w_b+w_c},
    &
    \gamma&=\frac{w_a\cdot w_b}{w_a+w_b+w_c}
\end{align*}
Unfortunately, such a transformation does not exist when $p\neq
1,2,\infty$,
as we show next.

We first define the transform for general $k$.
Note that for general values of  $k$,
the transform is called a $k$-star-mesh transform,
where in this subsection we focus on the special
case of $k=3$ (and we will continue to call
it $Y$-$\Delta$-transform).
\begin{definition}[$k$-star-mesh transform]
A  $k$-star-mesh transform 
is an operation that given a graph $G$ and 
a vertex $r\in V(G)$ of degree $k$,
removes $r$ from $G$ and
replaces the $k$-star formed by $r$ and its neighbors
with a (possibly weighted) clique
on the neighbor set $N(r)$.
\end{definition}
\begin{definition}[local \textit{k}-star-mesh transform]
A $k$-star-mesh transform is called \textbf{local}
if the edge weights inside the clique $N(r)$
in the transformed graph
depend only on the edge weights of the
$k$-star in $G$
(obliviously to the rest of the graph $G$).
\end{definition}
\begin{definition}[local \textit{k}-star-mesh transform
preserving $d_p$]
We say that a local \textit{k}-star-mesh
transform \textbf{preserves} $d_p$ if
for every graph $G=(V,E,w)$
and a vertex $r\in V$ of degree $k$,
applying the transform on 
$G$ and $r$ yields a graph $G'$ where
\begin{equation*}
    \forall s,t\in V\backslash\{r\},
    \quad d_{p,G'}(s,t)
    = d_{p,G}(s,t).
\end{equation*}
\end{definition}
We can now state our main theorem of this subsection.
\begin{theorem}\label{thm: non existence of Y-Delta transform}
For every $p\neq1,2,\infty$,
there is no local Y-$\Delta$
transform  
that preserves the $d_p$ 
metric.
\end{theorem}
We prove this by showing two graphs
that contain the same unweighted 3-star 
(``Y")
as an induced subgraph,
on which such a  transformation,
if one existed,
must have acted differently.
Of course, this is not possible,
since the transformation
should be local and thus the same in the two graphs.
Before we present the claim formally, let us add some notation.
Denote by $V_T=\curprs{a,b,c}$ a set of 3 vertices we will call
``terminals",
denote $V_Y=V_T\cup\{r\}$,
and denote by $G_Y=\prs{V_Y,E_Y}$ the 3-star over 
the terminals, i.e. $E_Y=\curprs{\{a,r\},\{b,r\},\{c,r\}}$.
In addition, denote by $G_\Delta=\prs{V_T,E_\Delta}$
the triangle over the terminals,
i.e. $E_\Delta=\{\{a,b\},\{b,c\},\{c,a\}\}$.

\begin{claim}\label{claim: existence of graphs that contradict the existence of Y-Delta transform}
Let $p\in(1,2)\cup(2,\infty)$.
There exist graphs $G_1=\prs{V_1,E_1,w_1}, G_2=\prs{V_2,E_2,w_2}$,
with $V_Y\subseteq V_i$, and $G_i\brs{V_Y}=G_Y$ for $i=1,2$,
that satisfy the following.
Suppose that we have graphs $G'_1=\prs{V'_1,E'_1,w'_1},
G_2=\prs{V'_2,E'_2,w'_2}$
on which  the ``Y" transformed into a ``$\Delta$",
i.e. satisfying for $i=1,2$,
\begin{align}
    V_i' & = V_i\backslash\{r\}\\
    E\prs{G'_i\brs{V_T}} & = E_\Delta\\
    E'_i\backslash E_\Delta & = E_i\backslash E_Y\\
    w'_i\mid_{E'_i\backslash E_\Delta} & = w_i\mid_{E_i\backslash E_Y}
\end{align}
If in addition we have that
\begin{equation}
    \forall s,t\in V_T,\quad d_{p,G_i'}(s,t)  = d_{p,G_i}(s,t)
\end{equation}
for $i=1,2$,
then there exists $e\in E_\Delta$ such that
$w'_1(e)\neq w'_2(e)$.
\end{claim}
Observe that Claim \ref{claim: existence of graphs that contradict the existence of Y-Delta transform}
immediately  gives
Theorem \ref{thm: non existence of Y-Delta transform}
as a corollary.
\begin{proof}
We will take $G_1=G_Y$ (with unit weights),
and $G_2$ to be $G_Y$ with the addition of a new vertex
and two edges as presented in figure
\ref{fig:new graph G_2 for Y-Delta transform}
below.
\begin{figure}[ht!]
    \centering

\tikzset{every picture/.style={line width=0.75pt}} 

\begin{tikzpicture}[x=0.75pt,y=0.75pt,yscale=-1,xscale=1]

\draw  [fill={rgb, 255:red, 0; green, 0; blue, 0 }  ,fill opacity=1 ] (186.44,195.78) .. controls (186.44,189.83) and (191.27,185) .. (197.22,185) .. controls (203.17,185) and (208,189.83) .. (208,195.78) .. controls (208,201.73) and (203.17,206.56) .. (197.22,206.56) .. controls (191.27,206.56) and (186.44,201.73) .. (186.44,195.78) -- cycle ;
\draw  [fill={rgb, 255:red, 0; green, 0; blue, 0 }  ,fill opacity=1 ] (295.44,42.78) .. controls (295.44,36.83) and (300.27,32) .. (306.22,32) .. controls (312.17,32) and (317,36.83) .. (317,42.78) .. controls (317,48.73) and (312.17,53.56) .. (306.22,53.56) .. controls (300.27,53.56) and (295.44,48.73) .. (295.44,42.78) -- cycle ;
\draw  [fill={rgb, 255:red, 0; green, 0; blue, 0 }  ,fill opacity=1 ] (377.44,195.78) .. controls (377.44,189.83) and (382.27,185) .. (388.22,185) .. controls (394.17,185) and (399,189.83) .. (399,195.78) .. controls (399,201.73) and (394.17,206.56) .. (388.22,206.56) .. controls (382.27,206.56) and (377.44,201.73) .. (377.44,195.78) -- cycle ;
\draw  [fill={rgb, 255:red, 0; green, 0; blue, 0 }  ,fill opacity=1 ] (291.44,146.78) .. controls (291.44,140.83) and (296.27,136) .. (302.22,136) .. controls (308.17,136) and (313,140.83) .. (313,146.78) .. controls (313,152.73) and (308.17,157.56) .. (302.22,157.56) .. controls (296.27,157.56) and (291.44,152.73) .. (291.44,146.78) -- cycle ;
\draw  [fill={rgb, 255:red, 0; green, 0; blue, 0 }  ,fill opacity=1 ] (437.44,54.78) .. controls (437.44,48.83) and (442.27,44) .. (448.22,44) .. controls (454.17,44) and (459,48.83) .. (459,54.78) .. controls (459,60.73) and (454.17,65.56) .. (448.22,65.56) .. controls (442.27,65.56) and (437.44,60.73) .. (437.44,54.78) -- cycle ;
\draw    (302.22,146.78) -- (388.22,195.78) ;
\draw    (302.22,146.78) -- (197.22,195.78) ;
\draw    (302.22,146.78) -- (306.22,42.78) ;
\draw    (306.22,42.78) -- (448.22,54.78) ;
\draw    (388.22,195.78) -- (448.22,54.78) ;

\draw (210,199.18) node [anchor=north west][inner sep=0.75pt]    {$a$};
\draw (275,46.18) node [anchor=north west][inner sep=0.75pt]    {$b$};
\draw (401,199.18) node [anchor=north west][inner sep=0.75pt]    {$c$};
\draw (361,22.18) node [anchor=north west][inner sep=0.75pt]    {$w=\infty $};
\draw (299,160.18) node [anchor=north west][inner sep=0.75pt]    {$r$};
\draw (463.22,28.96) node [anchor=north west][inner sep=0.75pt]    {$v$};
\draw (420.22,128.68) node [anchor=north west][inner sep=0.75pt]    {$w=\infty $};

\end{tikzpicture}
    
    \caption{Illustration of the graph $G_2$.}
    \label{fig:new graph G_2 for Y-Delta transform}
\end{figure}
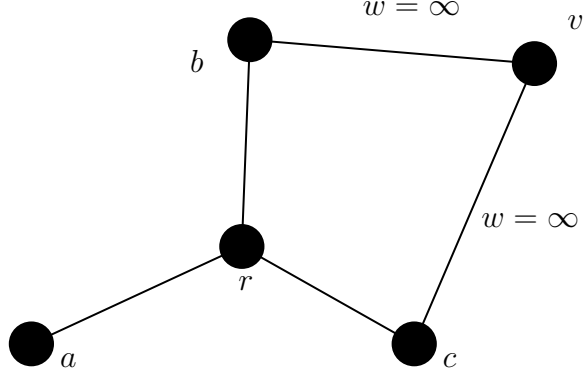
Formally, we define
$V_2=V_Y\cup\{v\}$,
$E_2=E_Y\cup\curprs{\{v,b\},\{v,c\}}$,
and $w_2\prs{\{v,b\}}=
w_2\prs{\{v,c\}}=\infty$,
and the rest of the edges are of unit weight.

Next, suppose that we have transformed the ``Y''
into ``$\Delta$" in the two graphs, and received
$G_1', G_2'$ that satisfy the conditions in the claim,
and assume towards contradiction that 
for all $e\in E_\Delta$,
$w'_1(e)=w'_2(e)$.
Denote the weights by $w_i'\prs{\{b,c\}}=\alpha,
w_i'\prs{\{a,c\}}=\beta,$
and $w_i'\prs{\{a,b\}}=\gamma$
(as presented in figure
\ref{fig:Y-Delta-transform}).

We first focus on the case of $G_1$,
and examine the new weights in $G_1'$.
Note that in our case, $G_1'$
is simply the triangle over the terminals with some new weights.
In addition, note that by symmetry,
the new weights must satisfy $\alpha=\beta=\gamma$.
It is easy to see that for any two pairs
$(s,t),(s',t')\in V_T\times V_T$ with $s\neq t$ and $s'\neq t'$,
it holds that
$d_{p,G_1}\prs{s,t} = d_{p,G_1}\prs{s',t'}$.
Thus, by our assumption,
this also holds in $G'_1$.
In order for this to hold in the triangle,
it must hold that $\alpha=\beta=\gamma$.
Next, let us compute $\bar{d}_{p,G_1}(a,b)$
(recall that $d_p=\prs{\bar{d}_p}^{-1}$ and thus
preserving $\bar{d}_p$
implies preserving $d_p$
and vice versa).
Note that by definition,
$\bar{d}_{p,G_1}(a,b)^q  = 
\underset{\varphi\in\R^{V_1}:
        \varphi_a-\varphi_b=1}{\min}
        \curprs{\abs{\varphi_a-\varphi_r}^q
        +\abs{\varphi_b-\varphi_r}^q
        +\abs{\varphi_c-\varphi_r}^q}$.
Thus, w.l.o.g we can set $\varphi_a=1,\varphi_b=0
        ,\varphi_r=x,\varphi_c=y$
for some $x,y\in \R$
and get that,
\begin{align*}
    \bar{d}_{p,G_1}(a,b)^q 
        & = \min_{x,y\in \R}
        \curprs{\abs{1-x}^q
        +\abs{0-x}^q
        +\abs{y-x}^q}
        \\
        & = \min_{x\in [0,1]}
        \curprs{\prs{1-x}^q + x^q}
        && (\text{by setting }y=x \text{ and }
        x\in[0,1])
\end{align*}
Define $f(x)=\prs{1-x}^q + x^q$,
and we wish to find a minimum
in the interval $[0,1]$.
Hence,
\begin{align*}
    && f'(x) & = 0 \\
    \iff &&
    -\prs{1-x}^{q-1}+x^{q-1}
    & = 0\\
    \iff &&
    x^{q-1}
    & = \prs{1-x}^{q-1}
    \\
    \iff &&
    x 
    & = \frac{1}{2}\\
\end{align*}
Thus, we get that
\begin{equation}
    \bar{d}_{p,G_1}(a,b)^q=
    \prs{1-\frac{1}{2}}^q
    +\prs{\frac{1}{2}-0}^q
    = 2\cdot \prs{\frac{1}{2}}^{q}
    = 2^{1-q}
\end{equation}
On the other hand,
let us compute $\bar{d}_{p,G_1'}(a,b)$
in terms of $\alpha$.
Again, w.l.o.g. we can set $\varphi_a=1,\varphi_b=0,\varphi_c=x$
and see that
\begin{align*}
    \bar{d}_{p,G_1'}(a,b)^q 
        & = \min_{x\in \R}
        \curprs{\alpha^q\cdot\abs{1-0}^q
        +\alpha^q\cdot\abs{0-x}^q
        +\alpha^q\cdot\abs{x-1}^q}
        \\
        & = \min_{x\in [0,1]}
        \curprs{\alpha^q
        +\alpha^q\cdot x^q
        +\alpha^q\cdot\prs{1-x}^q}
\end{align*}
Similarly, we can define $g(x) = \alpha^q
        +\alpha^q\cdot x^q
        +\alpha^q\cdot\prs{1-x}^q$
and find a minimum in the interval $[0,1]$.
By similar computations as above,
we deduce that the minimum
is achieved in $x=\frac{1}{2}$,
and thus we get that
\begin{equation}
    \bar{d}_{p,G'_1}(a,b)^q=\alpha^q\prs{
    1 
    +\prs{\frac{1}{2}}^q
    + \prs{1-\frac{1}{2}}^q
    }
    = \alpha^q\prs{
    1 +2^{1-q}
    }
\end{equation}
Thus, we can derive the value of the new weights $\alpha$.
\begin{equation}
    \begin{split}
        \alpha^q\cdot\prs{1+2^{1-q}} & = 2^{1-q}\\
        \implies \alpha & = \prs{\frac{2^{1-q}}{1+2^{1-q}}}^{1/q}
    \end{split}
\end{equation}
Hence, by examining the $d_p$ metric on $G_1$
and $G'_1$, we deduce that the weights should be 
$\alpha 
= \prs{1+2^{q-1}}^{-1/q}$.

Next, let us focus on the case of $G_2$,
and we start by computing $\bar{d}_{p,G_2}(a,b)$.
Note that in order to not pay the extremely large
weight of the edges
connecting $v$ to $b$ and $c$,
it must hold that  any minimizing potentials vector
will satisfy
$\varphi_b=\varphi_c=\varphi_v$.
Thus, we get that
\begin{align*}
    \bar{d}_{p,G_2}(a,b)^q & = 
    \min_{\varphi\in\R^{V_2}:\,
    \varphi_a-\varphi_b=1, 
    \varphi_b=\varphi_c=\varphi_v}
    \curprs{\abs{\varphi_a-\varphi_r}^q
    +\abs{\varphi_b-\varphi_r}^q
    +\abs{\varphi_c-\varphi_r}^q}\\
    & = \min_{x\in \R}
    \curprs{\abs{1-x}^q
    +\abs{0-x}^q
    +\abs{0-x}^q}
    &&(
    \varphi_a=1,\varphi_b=0,\varphi_r=x)\\
    & = \min_{x\in [0,1]}
    \curprs{\prs{1-x}^q + 2x^q}
\end{align*}
Similarly to the previous case, we can 
define $f(x)=\prs{1-x}^q + 2x^q$,
and we wish to find a minimum
by considering the derivative of $f$
and requiring it to be 0
(because again the minimum
is obtained when $x\in[0,1]$).
\begin{align*}
    && f'(x) & = 0 \\
    \iff &&
    -\prs{1-x}^{q-1}+2x^{q-1}
    & = 0\\
    \iff &&
    2 x^{q-1}
    & = \prs{1-x}^{q-1}\\
    \iff &&
    2^{1/(q-1)}\cdot x
    & = \prs{1-x}\\
    \iff &&
    x\cdot \prs{2^{1/(q-1)}+1}
    & = 1\\
    \iff &&
    x 
    & = \frac{1}{2^{1/(q-1)}+1}\\
\end{align*}
Thus, we get that 
\begin{equation}
    \bar{d}_{p,G_2}(a,c)^q=
    \prs{1-\frac{1}{2^{1/(q-1)}+1}}^q
    +2\cdot\prs{\frac{1}{2^{1/(q-1)}+1}}^q
    = \frac{2^{q/(q-1)}+2}{\prs{2^{1/(q-1)}+1}^q}
    = 2\cdot \prs{2^{1/(q-1)}+1}^{1-q}
\end{equation}
On the other hand,
let us compute $\bar{d}_{p,G_2'}(a,b)$
in terms of $\alpha$
(note that by our assumption all the weights
of edges from $E_\Delta$
must be equal).
Note that again, we do not want to pay the extremely large
weight of the edges
connecting $v$ to $b$ and $c$,
and thus we must set
$\varphi_b=\varphi_c=\varphi_v$.
Moreover, w.l.o.g we can choose $\varphi_a=1,\varphi_b=\varphi_c=0$
and get that,
\begin{align*}
        \bar{d}_{p,G_2'}(a,b)^q 
        & = \alpha^q\cdot\abs{1}^q
                +\alpha^q\cdot\abs{0}^q
                +\alpha^q\cdot\abs{-1}^q \\
        & = 2\alpha^q
\end{align*}
Thus, we can derive the value of the new weights $\alpha$.
\begin{align*}
    2\cdot\alpha^q & = 2\cdot \prs{2^{1/(q-1)}+1}^{1-q}\\
    \alpha & = \prs{2^{1/(q-1)}+1}^{\frac{1-q}{q}}
\end{align*}
Hence, in this case the new weights should be 
$\alpha  = \prs{2^{1/(q-1)}+1}^{-1/p}$.

Thus, we got that according to the first case, the weights
should be $ \alpha  =  \prs{1+2^{q-1}}^{-1/q}$,
but on the other hand they should be
$\alpha=\prs{2^{1/(q-1)}+1}^{-1/p}$.
Thus we ask whether
\begin{equation}\label{eq: 2 terms, are they identical? only when p=q=2}
    \prs{2^{1/(q-1)}+1}^{-1/p} \stackrel{?}{=}\prs{1+2^{q-1}}^{-1/q}
\end{equation}
It is easy to see that the two terms equal when $p=q=2$.
We will now show 
that they are different when $q\neq 2$.
First, note that (\ref{eq: 2 terms, are they identical? only when p=q=2})
is equivalent to
\begin{equation}\label{eq: 2 terms equiv 1 (power p)}
    2^{1/(q-1)}+1  \stackrel{?}{=}\prs{1+2^{q-1}}^{p-1}
\end{equation}
by taking power $p$.

Next, assume that $p\geq 2$, 
and focus on the RHS.
By applying $f\prs{\frac{a+b}{2}}\leq \frac{f(a)+f(b)}{2}$
for $f(x)=x^{p-1}$ (since $p\geq 2$ it is convex)
we can see that,
\begin{align*}
    \prs{1+2^{q-1}}^{p-1}
    & \leq 2^{p-1}\cdot\prs{ \prs{\frac{1}{2}}^{p-1}
    +2^{(q-2)\cdot(p-1)}
    }\\
    & = 1+2^{p -2(p-1) + (p-1)}\\
    & = 3
\end{align*}
Now, plugging this into (\ref{eq: 2 terms equiv 1 (power p)})
we get
\begin{equation}
    \begin{split}
        2^{1/(q-1)} & \leq  2\\
        \iff \frac{1}{q-1}  & \leq 1 \\
        \iff 1 & \leq q-1 \\
        \iff 2 & \leq q \\
        \iff 2 & \geq p
    \end{split}
\end{equation}
Recall that we assumed that $p\geq 2$ (in order to apply Jensen's
inequality) and reached the conclusion that $p\leq 2$,
thus it is impossible that the two terms in 
(\ref{eq: 2 terms, are they identical? only when p=q=2})
equal when $p>2$.

Next, we assume that $p\leq 2$ (which implies that $q\geq 2$),
and note that (\ref{eq: 2 terms, are they identical? only when p=q=2})
is equivalent to
\begin{equation}\label{eq: 2 terms equiv 2 (power q)}
    \prs{2^{1/(q-1)}+1}^{q-1}  \stackrel{?}{=}1+2^{q-1}
\end{equation}
by taking power $q$.
This time we focus on the LHS and similarly to the previous case,
we apply Jensen's inequality on it and get that
\begin{align*}
    \prs{2^{1/(q-1)}+1}^{q-1} & \leq 
    2^{q-1}\cdot\prs{2^{\prs{\frac{1}{q-1} - 1}\cdot(q-1)}+\prs{
    \frac{1}{2}}^{q-1}}\\
    & = 2^{1-(q-1)+(q-1)}+1\\
    & = 3
\end{align*}
Now, plugging this into (\ref{eq: 2 terms equiv 2 (power q)})
we get
\begin{equation}
    \begin{split}
        2^{q-1} & \leq  2\\
        \iff q-1  & \leq 1 \\
        \iff q & \leq 2 \\
        \iff p & \geq 2
    \end{split}
\end{equation}
and again, recall that we assumed that $p\leq 2$,
and reached the conclusion that $p\geq 2$,
which implies that the two terms in
(\ref{eq: 2 terms, are they identical? only when p=q=2})
equal if and only if $p=q=2$.
Thus, we have reached a contradiction,
and the claim follows.
\end{proof}

Note that a corresponding example will also
contradict the other direction, i.e. that there
is no valid $\Delta-Y$ transform for any $p\neq 2$.
            
We see that our counter example does not
contradict the case of $p=q=2$, 
since it holds that in both cases
$\alpha=\frac{1}{ \sqrt{3}}$.
As a matter of fact, in this case, there
exists a proper transform (similarly to the known
transform for the effective resistance),
and we show it in  Appendix \ref{appendix section - Proof of Delta-Y transform for p=2}.

\subsubsection{The Case of Minimum Cuts
(\texorpdfstring{$p=\infty$}{\textit{p}=infinity})}
\label{subsection: p=infinity}
In this subsection we elaborate on the case of $p=\infty$,
i.e. minimum cuts.
The general case of the $Y-\Delta$-transform 
is called a $k$-star-mesh transform, 
in which a $k$-star (a root vertex 
with $k$ neighbors) is
transformed into a clique (``mesh") over $k$
vertices, formally defined as 
shown in subsection \ref{section: non existence of delta-wye transform}.

By Theorem 1 from \cite{chaudhuri2000mimickingNetworks},
and also directly from basic principles,
it is easy to conclude that there exists such a transform
for $k=3$ and for $k=2$
(which is essentially the sequential edges
reduction case).
However,
for every $k>3$,
using Lemma 4 from \cite{chaudhuri2000mimickingNetworks},
we can deduce that there does not exist
such a transform for removing a vertex of degree $k$,
stated as follows.
\begin{theorem}\label{proposition: non existence of star-mesh transform for p=infinity and k>3}

For every $k>3$,
there does not exist a local
$k$-star-mesh transform that preserves $d_\infty$.
\end{theorem}

This is in fact a corollary from the following
Lemma presented in \cite{chaudhuri2000mimickingNetworks}.
\begin{lemma}[Lemma 4 in \cite{chaudhuri2000mimickingNetworks}]
There exists a $k$-terminal network for which every
mimicking network must have at least one non-terminal
vertex (in addition to the $k$ terminals).
\end{lemma}

Its proof in \cite{chaudhuri2000mimickingNetworks}
actually shows the following Lemma 
(restated using our terminology).
\begin{lemma}\label{lemma: not possible to remove vertex of degree >3 and preserve all the cuts between the neighbors}
Let $k>3$, 
and let $G=(V,E)$ be an unweighted
star with a (root) vertex $r$  of degree $k$.
Denote by $G'$  the graph obtained 
after applying a $k$-star-mesh transform on $G$.
Then, there must exist  $S\subset N(r)$
such that
\[
    \mincut_{G'}\prs{S,N(r)\backslash(S)}
    \neq \mincut_{G}\prs{S,N(r)\backslash(S)}.
\]
\end{lemma}
Note that it does not immediately imply Theorem \ref{proposition: non existence of star-mesh transform for p=infinity and k>3},
as in order to preserve the $d_\infty$ metric we only 
need to preserve all \textit{minimum $st$-cuts}, rather
than  minimum cuts between \textit{all subsets
of }$N(r)$.
\begin{proof}(of Theorem \ref{proposition: non existence of star-mesh transform for p=infinity and k>3})
Assume towards contradiction that there exists a local 
$k$-star-mesh transform that preserves the $d_\infty$ metric.
Let $G=(V,E)$ be an unweighted $k$-star graph, i.e. a root vertex 
$r$ with $k$ neighbors.
Then, apply the transform on $G$ in order to obtain a 
clique $G'$.
By Lemma \ref{lemma: not possible to remove vertex of degree >3 and preserve all the cuts between the neighbors},
there is a subset $S\subset N(r)$ such that 
\begin{equation}\label{eq: mincuts are not the same on S}
    \mincut_{G'}\prs{S,N(r)\backslash(S)}
    \neq \mincut_{G}\prs{S,N(r)\backslash(S)}.
\end{equation}
Next, consider the following graph $G_S$.
Add to $G$ two vertices $v_S$ and $u_S$,
connect $v_S$ to every vertex in $S$,
and connect $u_S$ to every vertex in $N(r)\backslash S$.
In addition, define an edge-weight function $w$
given by
\[
    w(e) = \begin{cases}
    1   & \text{if } e \text{ is incident to }r;\\
    \infty   & \text{o.w.}
    \end{cases}
\]
i.e. the weights of the star edges remain the same,
and the weights of the newly added edges
are infinite.
An illustration is presented in Figure \ref{fig:G_S illustration}.
\begin{figure}[htb!]
    \centering

\tikzset{every picture/.style={line width=0.75pt}} 

\begin{tikzpicture}[x=0.75pt,y=0.75pt,yscale=-1,xscale=1]

\draw [fill={rgb, 255:red, 0; green, 0; blue, 0 }  ,fill opacity=1 ]  (318.22,150.78) .. controls (318.22,145.93) and (322.15,142) .. (327,142) .. controls (331.85,142) and (335.78,145.93) .. (335.78,150.78) .. controls (335.78,155.63) and (331.85,159.56) .. (327,159.56) .. controls (322.15,159.56) and (318.22,155.63) .. (318.22,150.78) -- cycle ;
\draw [fill={rgb, 255:red, 0; green, 0; blue, 0 }  ,fill opacity=1 ]  (381,141.78) .. controls (381,136.93) and (384.93,133) .. (389.78,133) .. controls (394.63,133) and (398.56,136.93) .. (398.56,141.78) .. controls (398.56,146.63) and (394.63,150.56) .. (389.78,150.56) .. controls (384.93,150.56) and (381,146.63) .. (381,141.78) -- cycle ;
\draw  [fill={rgb, 255:red, 0; green, 0; blue, 0 }  ,fill opacity=1 ] (358.78,94.78) .. controls (358.78,89.93) and (362.71,86) .. (367.56,86) .. controls (372.4,86) and (376.33,89.93) .. (376.33,94.78) .. controls (376.33,99.63) and (372.4,103.56) .. (367.56,103.56) .. controls (362.71,103.56) and (358.78,99.63) .. (358.78,94.78) -- cycle ;
\draw [fill={rgb, 255:red, 0; green, 0; blue, 0 }  ,fill opacity=1 ]  (358.78,213.78) .. controls (358.78,208.93) and (362.71,205) .. (367.56,205) .. controls (372.4,205) and (376.33,208.93) .. (376.33,213.78) .. controls (376.33,218.63) and (372.4,222.56) .. (367.56,222.56) .. controls (362.71,222.56) and (358.78,218.63) .. (358.78,213.78) -- cycle ;
\draw  [fill={rgb, 255:red, 0; green, 0; blue, 0 }  ,fill opacity=1 ] (272,180.78) .. controls (272,175.93) and (275.93,172) .. (280.78,172) .. controls (285.63,172) and (289.56,175.93) .. (289.56,180.78) .. controls (289.56,185.63) and (285.63,189.56) .. (280.78,189.56) .. controls (275.93,189.56) and (272,185.63) .. (272,180.78) -- cycle ;
\draw [fill={rgb, 255:red, 0; green, 0; blue, 0 }  ,fill opacity=1 ]  (279,115.78) .. controls (279,110.93) and (282.93,107) .. (287.78,107) .. controls (292.63,107) and (296.56,110.93) .. (296.56,115.78) .. controls (296.56,120.63) and (292.63,124.56) .. (287.78,124.56) .. controls (282.93,124.56) and (279,120.63) .. (279,115.78) -- cycle ;
\draw  [fill={rgb, 255:red, 0; green, 0; blue, 0 }  ,fill opacity=1 ] (488,124.78) .. controls (488,119.93) and (491.93,116) .. (496.78,116) .. controls (501.63,116) and (505.56,119.93) .. (505.56,124.78) .. controls (505.56,129.63) and (501.63,133.56) .. (496.78,133.56) .. controls (491.93,133.56) and (488,129.63) .. (488,124.78) -- cycle ;
\draw [fill={rgb, 255:red, 0; green, 0; blue, 0 }  ,fill opacity=1 ]  (176,150.78) .. controls (176,145.93) and (179.93,142) .. (184.78,142) .. controls (189.63,142) and (193.56,145.93) .. (193.56,150.78) .. controls (193.56,155.63) and (189.63,159.56) .. (184.78,159.56) .. controls (179.93,159.56) and (176,155.63) .. (176,150.78) -- cycle ;
\draw    (287.78,115.78) -- (327,150.78) ;
\draw    (327,150.78) -- (367.56,213.78) ;
\draw    (327,150.78) -- (389.78,141.78) ;
\draw    (367.56,94.78) -- (327,150.78) ;
\draw    (327,150.78) -- (280.78,180.78) ;
\draw [line width=2.25]    (280.78,180.78) -- (184.78,150.78) ;
\draw [line width=2.25]    (287.78,115.78) -- (184.78,150.78) ;
\draw [line width=3]    (496.78,124.78) -- (367.56,94.78) ;
\draw [line width=3]    (496.78,124.78) -- (389.78,141.78) ;
\draw [line width=2.25]    (496.78,124.78) -- (367.56,213.78) ;
\draw   (254,151.78) .. controls (254,101.09) and (267.33,60) .. (283.78,60) .. controls (300.22,60) and (313.56,101.09) .. (313.56,151.78) .. controls (313.56,202.47) and (300.22,243.56) .. (283.78,243.56) .. controls (267.33,243.56) and (254,202.47) .. (254,151.78) -- cycle ;
\draw   (344.56,150.78) .. controls (344.56,100.09) and (357.89,59) .. (374.33,59) .. controls (390.78,59) and (404.11,100.09) .. (404.11,150.78) .. controls (404.11,201.47) and (390.78,242.56) .. (374.33,242.56) .. controls (357.89,242.56) and (344.56,201.47) .. (344.56,150.78) -- cycle ;

\draw (150,138.4) node [anchor=north west][inner sep=0.75pt]    {$v_{S}$};
\draw (516,114.4) node [anchor=north west][inner sep=0.75pt]    {$u_{S}$};
\draw (276,33.4) node [anchor=north west][inner sep=0.75pt]    {$S$};
\draw (320,115.4) node [anchor=north west][inner sep=0.75pt]    {$r$};
\draw (340,30.4) node [anchor=north west][inner sep=0.75pt]    {$N( r) \backslash S$};
\draw (219,112.4) node [anchor=north west][inner sep=0.75pt]    {$\infty $};
\draw (219,144.4) node [anchor=north west][inner sep=0.75pt]    {$\infty $};
\draw (423.17,147.68) node [anchor=north west][inner sep=0.75pt]    {$\infty $};
\draw (421.17,118.68) node [anchor=north west][inner sep=0.75pt]    {$\infty $};
\draw (423.17,89.68) node [anchor=north west][inner sep=0.75pt]    {$\infty $};

\end{tikzpicture}
    \caption{Illustration of the graph $G_S$.}
    \label{fig:G_S illustration}
\end{figure}
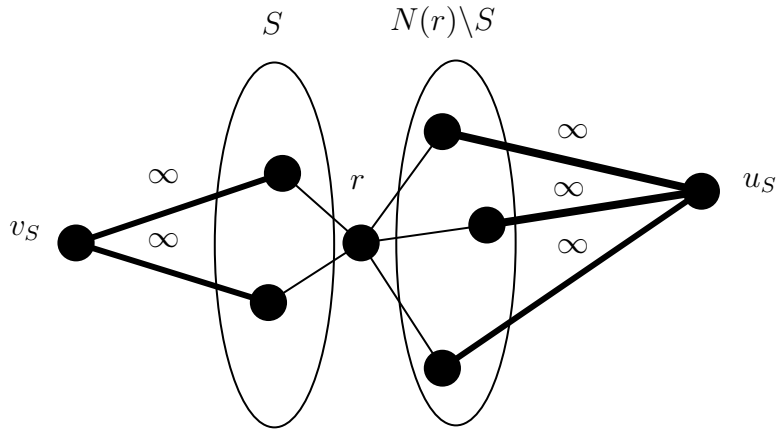

Observe that $d_{\infty,G_S}\prs{v_S,u_S}=
\mincut_{G}\prs{S,N(r)\backslash S}$.
Next, apply the transform on $G_S$ and denote
the obtained graph by $G'_S$.
We assumed the transform preserves $d_\infty$,
and thus in particular 
$d_{\infty,G'_S}\prs{v_S,u_S}=d_{\infty,G_S}\prs{v_S,u_S}$.
But this implies that 
\[
    \mincut_{G'}\prs{S,N(r)\backslash(S)}
    = \mincut_{G}\prs{S,N(r)\backslash(S)},
\]
in contradiction to (\ref{eq: mincuts are not the same on S}).
\end{proof}

Note that the proof in fact shows that there does not
exist a local $k$-star-mesh transform that preserves
$d_\infty$ even for the family of planar graphs,
as $G_S$ is planar.
Moreover, the construction of $G_S$ can be modified
such that $G_S$ will be outer-planar,
which will give the same result even for 
 outer-planar graphs.
However, for trees there does exist such a transform.

\section{Conclusions and Open Questions}\label{chapter: conclusions and open questions}

In this section we discuss  questions
that are left open from our work.

\paragraph{Regarding $d_p$ sparsifiers.}
In subsection \ref{section - lower bound on resistance sparsifiers}
we proved a general lower bound
of $\Omega(n/\sqrt{\varepsilon})$
edges for an $\varepsilon$-resistance sparsifier,
while there is an upper bound of $O(n/\varepsilon)$
edges for a resistance sparsifier for the clique.
We proved this by showing 
that every non-complete graph cannot achieve
better than 
$\prs{1+\frac{1}{O(n^2)}}$-approximation 
of the resistance distance in an $n$-clique.
However, we showed the stronger bound of
$\prs{1+\frac{1}{O(n)}}$-approximation 
for regular graphs,
which intuitively, seem to be the best fit for this task.
This result suggests that 
the stronger bound should hold in general,
which would prove Conjecture 
\ref{intro: conjecture: lower bound on resistance sparsifiers},
stating that in the worst case,
an $\varepsilon$-resistance sparsifier requires
$\Omega(n/\varepsilon)$ edges.
\begin{oq}
Is it true that
for every non-complete graph,
\begin{equation}
    \frac{\max_{x\neq y\in V}\Reff(x,y)}{\min_{x\neq y\in V}\Reff(x,y)}
    \geq  1+ \frac{1}{O(n)}?
\end{equation}
\end{oq}
Another open question
is to extend the lower bound on resistance
sparsifiers to other values of $p$
(the solved cases are $p=1,\infty$ with matching
upper and lower bounds),
as our lower bound for effective resistance
does not extend to other
values of $p$.

Additionally, in subsection \ref{section - flow metric sparsifiers}
we saw that for fixed $p\in (4/3,\infty]$
with H\"older conjugate $q$,
there exists $d_p$-sparsifiers
with 
$|E'|=f(n,\varepsilon,p)$ edges,
where
\begin{equation}\label{eq: number of edges for dp sparsifiers}
    f(n,\varepsilon,p) = 
    \begin{cases}
        n-1     & \text{if }p=\Omega\prs{\varepsilon^{-1} \log n},\\
        \widetilde{O}\prs{n\varepsilon^{-2}} & 
        \text{if }2<p<\infty,\\
        \widetilde{O}\prs{n^{q/2}\varepsilon^{-5}} & 
        \text{if }\frac{4}{3} < p < 2.
    \end{cases}
\end{equation}
One question is regarding
the number of edges needed
to preserve $d_p$ for  $p$ that is  close to $2$.
Recall that for $d_2$
there is the construction of Chu et al
\cite{chu2020shortCycleLongPaper}
for resistance sparsifiers with $\widetilde{O}(n/\varepsilon)$
edges.
In \eqref{eq: number of edges for dp sparsifiers}
we see that in order to preserve 
the $d_p$ metric for $p>2$, then $\widetilde{O}\prs{n\cdot\varepsilon^{-2}}$ edges
suffice, 
i.e. we have a gap in the dependence on $\varepsilon$.
Additionally, for $\frac{3}{4}<p<2$, we 
see that $\widetilde{O}\prs{n^{q/2}\cdot\varepsilon^{-5}}$
edges suffice (where $q/2>1$),
i.e. there is a gap in both $\varepsilon$ and $n$.
Thus, it would be interesting to close these gaps from both
sides of $p=2$.
\begin{oq}
For $p=2\pm 0.01$, does every graph $G$
with $n$ vertices
admits a $d_p$-sparsifier 
that achieves $1+\varepsilon$ approximation
with $\widetilde{O}(n/\varepsilon)$ edges?
\end{oq}
In addition, as $p$ tends to 1, the $d_p$ metric
tends to the shortest-path metric,
where multiplicative spanners are in fact $d_1$-sparsifiers
\cite{peleg1989graph, althofer1993sparse},
for which there exists a lower bound on the number of edges
needed in order to preserve the shortest-path distance.
This leads to the question  whether we can  find 
a lower bound on the number of edges needed 
to preserve the $d_p$ metric for values of $p$
close to 1.
\begin{oq}
For $p=1.1$, can we prove that
a $d_p$-sparsifier requires
$n^{1+\Omega(1)}$
edges in order to achieve 
$1.01$ approximation?
\end{oq} 

Another interesting question is
the gap in the number of needed edges for $d_p$
sparsifiers for $p\in(2,\infty)$.
For instance, for $p=2$ there exists a resistance sparsifier
with $\widetilde{O}\prs{n\cdot\varepsilon^{-1}}$ edges
\cite{chu2020shortCycleLongPaper},
for $p=\Omega(\varepsilon^{-1}\log n)$ 
there exists
a $d_\infty$-sparsifier with $n-1$ edges, which is clearly the
best we can hope for,
and yet for any value $2<p<\varepsilon^{-1}\log n$
the best result we have so far requires
 $\widetilde{O}\prs{n\cdot \varepsilon^{-2}}$ edges.
An explanation for this phenomenon is that 
essentially our proof for flow metric sparsifiers
showed a generalization
of spectral sparsifiers,
since the sparsifiers preserve the norm
of $\norm{WB\varphi}_q$
for every $\varphi\in \R^V$,
which is stronger than what we need.
\begin{oq}
For $2<p<\varepsilon^{-1}\log n$,
there exists $d_p$-sparsifier with
$\widetilde{O}(n/\varepsilon^2)$ edges.
Can we remove the $\poly(\log n)$ factors?
Can we reduce the dependency in 
$\varepsilon$ to (say) $\varepsilon^{-1}$?
Can we show that 
a tree with $n-1$ edges is \textbf{not} sufficient?
\end{oq}

\paragraph{Regarding Delta-Wye transform.}
It is known that for $p=1,2$ and every $k\geq 3$,
there exists
local $k$-star-mesh transforms that preserve $d_p$.
However, in subsection \ref{section - transforms for flow metrics}
we showed that for $k=3$,  \textbf{local}  
$k$-star-mesh-transform that preserves
$d_p$
exists if and only if $p=1,2,\infty$.
Moreover, we showed that for every $k>3$,
there is no \textbf{local}  
$k$-star-mesh-transform that preserves
$d_\infty$.
What about transforms such that the weights
of the new edges may depend on the rest of the graph
(i.e. not local)?
For example in the nature of Schur-Complements.
\begin{oq}
Does there exist a $k$-star-mesh-transform that preserves
$d_p$ for $k=3$ and $p\neq 1,2,\infty$?
for $k>3$ and $p=\infty$?
\end{oq}

\paragraph{Understanding the geometry of the flow metrics.}
An  important tool for understanding the
structure of the flow metrics
is via metric embeddings, 
i.e. mapping a metric space into another one
(specifically in our case into a
normed space)
while preserving the distances -
in which case the mapping is called
an isometry,
or up to some error
 - in which case we say that the mapping has
distortion $>1$,
see e.g. \cite{MatousekDiscreteGeometry,matouvsek1997embedding, matouvsek2013lectureEmbeddings}.
Towards this,
in subsection \ref{section - p-strong triangle inequality}
we showed that the $d_p$
metrics are $p$-strong, which gives some information
about their structure.
For example, for fixed $p\geq 2$,
in any embedding of $d_p$
into $\ell_2$, no 3 points
can lie on the same line.

For the special cases of $p=1,2, \infty$,
it is known that $d_p$ embeds isometrically into
$\ell_q$ (with $q$ being the H\"older conjugate of $p$).
We conjecture that this should hold in general.
\begin{conjecture}\label{conjecture: d_p into l_q embedding conjecture}
Fix $p\in [1,\infty]$ with H\"older conjugate $q$,
and let $G=(V,E,w)$ be a graph.
Then there exists a mapping $\Phi:V\rightarrow \ell_q$
such that
\begin{equation}
    \forall s,t\in V,\quad
    d_p(s,t)=\norm{\Phi(s)-\Phi(t)}_q.
\end{equation}
\end{conjecture}
We remark that the $p$-strong triangle
inequality alone is not enough in order to prove
the above,
and thus there must be additional properties
of the flow metrics that should be used.
We present it formally in Appendix
\ref{section: p strong is not enough for embedding conjecture}.

Furthermore,
by using the connection between the resistance distance
and the graph Laplacian, it was shown
that the resistance distance
is isometrically embedded into $\ell_2^2$
\cite{spielmanSrivastava2011graphSparsificationReff}.
We suspect that this approach can be generalized
by using the connection between $d_p$
and the graph $q$-Laplacian
(presented in Appendix \ref{appendix: connection to the laplacian})
in order to show  that $d_p^p$
can be isometrically embedded
into $\ell_2^{2}$,
which would give some more information about their structure.
\begin{conjecture}\label{conjecture: d_p^p embeds isometricaly into l_2^2}
Fix $p\in (1,\infty)$,
and let $G=(V,E,w)$ be a graph.
Then there exists a mapping $\Psi:V\rightarrow \ell_2^2$
such that
\begin{equation}
    \forall s,t\in V,\quad
    d_p(s,t)^p=\norm{\Psi(s)-\Psi(t)}_2^2.
\end{equation}
\end{conjecture}

\paragraph{Small sketches.}
Once we understand the geometry of the flow metrics,
e.g. which metric spaces they embed into,
we would like to find the best trade-off between
dimension and approximation 
of such embeddings.
One famous example is the Johnson-Lindenstrauss Lemma
\cite{johnsonLindenstrauss1984extensions}
that states that every $n$ points  in $\ell_2$,
can be embedded 
with distortion $1+\varepsilon$ (for every $\varepsilon>0$)
into a subspace of $\ell_2$ of dimension
$O\prs{\varepsilon^{-2}\cdot\logn}$.
Once we reduce the dimension, we can design natural small sketches and 
exploit them  to improve running time and storage requirements
of algorithms.
In particular, Spielman and Srivastava \cite{spielmanSrivastava2011graphSparsificationReff}
utilized such an embedding of the resistance distance
in order to construct a data structure
that given a query  pair of vertices,
returns an approximation of the effective resistance
between them.
If 
Conjecture \ref{conjecture: d_p^p embeds isometricaly into l_2^2}
is true,
their approach can be generalized
and yield a small sketch
for $d_p^{p/2}$, and thus also for $d_p$.

\paragraph{Computing all-pairs distances.}
Another important line of research 
is to compute the distance between all pairs of
vertices simultaneously, or to
construct a data structure
that given as query
 a pair of vertices, returns the exact $d_p$ distance
(or an approximation to it) between them.
Such constructions are known for 
the three special cases.
For $p=\infty$, there is the 
Gomory-Hu tree
\cite{GomoryHu1961multi};
for $p=1$, there are distance oracles
\cite{awerbach1993NearLinearDO,thorupZwick2005approximate, chechik2015ApproximateDO},
All-Pairs Shortest-Path algorithms
\cite{chan2010more, seidel1995all},
and spanners
\cite{peleg1989graph, althofer1993sparse};
for $p=2$, there are constructions by Spielman and Srivastava
\cite{spielmanSrivastava2011graphSparsificationReff},
and later on by Jambulapati and Sidford \cite{jambulapati2018efficient}
for approximating
the effective resistance.
It would be interesting to design algorithms that solve
this problem for other values of $p$,
as well as give lower bounds for
this problem.

\paragraph{Capturing properties of the underlying graphs.}
As mentioned earlier, 
effective resistance captures key properties of the underlying graph.
For example,
the effective resistance between
two  vertices
connected by a unit-weight edge,
equals the probability
that this edge   belongs to a uniformly random
spanning tree of the graph,
and moreover, 
it also equals the commute time between them,
up to scaling by a factor
that depends on the weights in the graph.
A natural direction  is to extend this  characterizations
of the effective resistance ($p=2$)
to other values of $p$,
or to find other properties of the underlying graphs
captured by them.

\small 
\printbibliography[heading=bibintoc, title={Bibliography}]

\appendix

\section{Omitted Proofs from Basic Properties Section}
\subsection{Deriving the Connection to the Dual Problem}\label{appendix: ommited proofs from basic props}
In this subsection we show how Claim \ref{claim: d_p = 1 / bar(d)_p}
is derived from Proposition 4 in \cite{alamgir2011phaseTransitionPresistance},
who consider the following optimization problems
for a graph $G=(V,E,w)$ and fixed $p\in(1,\infty)$
with H\"older conjugate $q$.
\paragraph{Flow problem.}
\begin{equation}
    R_p(s,t)=\min\curprs{\sum_{e\in E}\frac{\abs{f(e)}^p}{w(e)}
    \,:\,B^Tf=\chi_s-\chi_t}.
\end{equation}
\paragraph{Potential problem.}
\begin{equation}
    C_p(s,t)=\min\curprs{
    \sum_{xy\in E}w(xy)^{\frac{q}{p}}\abs{\varphi_x-\varphi_y}^q
    \,:\,\varphi_s-\varphi_t=1}.
\end{equation}
We remark that $d_p^p$ is just $R_p$ on a graph with
$w^p$ as a weight function,
and the same holds with $\Bar{d}_p^q$ and $C_p$.
Moreover, \cite{alamgir2011phaseTransitionPresistance} show that
$R_p$ and $C_p$ are related in the following manner.
\begin{proposition}[Proposition 4\footnote{The original statement in 
\cite{alamgir2011phaseTransitionPresistance} is stated with power $-\frac{q}{p}$.
However,
in the supplementary material they proved the version we presented.} in \cite{alamgir2011phaseTransitionPresistance}]\label{proposition 4 in AvL}
Fix $p>1$ with H\"older conjugate $q$, 
and let $G=(V,E,w)$ be a graph. Then,
\begin{equation}
    \forall s,t\in V,\quad
    R_p(s,t)=\prs{C_p(s,t)}^{-\frac{p}{q}}.
\end{equation}
\end{proposition}
We can now show how Claim \ref{claim: d_p = 1 / bar(d)_p}
is an easy consequence of Proposition \ref{proposition 4 in AvL}.
\begin{proof}(Claim \ref{claim: d_p = 1 / bar(d)_p})
Let $G=(V,E,w)$ be a graph, and let $G^p=(V,E,w^p)$.
Hence, we see that for every $s\neq t\in V$,
\begin{equation}
    d_{p,G}(s,t)^p=R_{p,G^p}(s,t)
    = \prs{C_{p,G^p}(s,t)}^{-\frac{p}{q}}
    = \prs{\bar{d}_{p,G}(s,t)^q}^{-\frac{p}{q}}
    = \bar{d}_{p,G}(s,t)^{-p}.
\end{equation}
\end{proof}

\subsection{Connection to the Graph \textit{p}-Laplacian}\label{appendix: connection to the laplacian}
In this subsection we show how Fact
\ref{fact: closed form solution for d_p}
can be used to relate between the $d_p$ metric
and the graph $p$-Laplacian.

Given a graph $G=(V,E,w)$, and fix $p\in(1,\infty)$
with H\"older conjugate $q$,  denote
the $q$-Laplacian of $G$ by $L_q:\R^V\rightarrow\R^V$,
given by
\begin{equation}
    \forall x\in V,\quad 
    \prs{L_q\varphi}_x = \sum_{y\in V}w(xy)^q\prs{\varphi_x-\varphi_y}\abs{\varphi_x-\varphi_y}^{q-2}.
\end{equation}
This is a non-linear generalization of the ordinary
Laplacian of the graph.
Now, using Fact
\ref{fact: closed form solution for d_p},
we can see
that for a pair $s,t\in V$, and minimizing flow $f^*$
for $d_p(s,t)$,
and a corresponding potentials vector $\varphi^*$,
it holds that
\begin{equation}
    B^Tf^* = L_q\varphi^*.
\end{equation}
In addition, note that $L_q$ satisfies for any $\varphi\in \R^V$
\begin{equation}
    \begin{split}
        \iprod{\varphi}{L_q\varphi}
        & = \sum_{x\in V}\sum_{y\in V}w(xy)^q\prs{\varphi_x-\varphi_y}\abs{\varphi_x-\varphi_y}^{q-2}\cdot\varphi_x\\
        & = \sum_{xy\in E}\prs{
        w(xy)^q\prs{\varphi_x-\varphi_y}\abs{\varphi_x-\varphi_y}^{q-2}\cdot\varphi_x
        +w(xy)^q\prs{\varphi_y-\varphi_x}\abs{\varphi_y-\varphi_x}^{q-2}\cdot\varphi_y}\\
        & = \sum_{xy\in V}w(xy)^q\abs{\varphi_x-\varphi_y}^q\\
        & = \norm{WB\varphi}_q^q.
     \end{split}
\end{equation}
and in particular, for connected graphs,
\begin{equation}\label{eq: quadratic form of L_q is 0 iff x in span(1)}
    \iprod{\varphi}{L_q\varphi} = 0
    \iff
    \varphi \in \text{span}\{\OneVec\}.
\end{equation}
This is also shown as Proposition 3.1 in \cite{buhler2009spectralPLaplacianClustering}.
This motivates the definition of the second eigenvalue
of the $q$-Laplacian.
\begin{definition}
Let $p\in(1,\infty)$ with H\"older conjugate $q$,
and let $G=(V,E,w)$ be a connected graph with $q$-Laplacian $L_q$.
The second smallest eigenvalue of $L_q$, denoted by
$\lambda_q^{(2)}$,
is defined as
\begin{equation}
    \lambda_q^{(2)} = \min_{\varphi\perp \OneVec}\frac{\iprod{\varphi}{L_q\varphi}}{\norm{\varphi}_2^2}.
\end{equation}
\end{definition}
We remark that a consequence of \eqref{eq: quadratic form of L_q is 0 iff x in span(1)}
(and Proposition 3.1 in \cite{buhler2009spectralPLaplacianClustering})
is that for connected graphs,
$\lambda_q^{(2)}>0$.
Next, we show that $d_p$ and $\lambda_q^{(2)}$
are related as follows.
\begin{claim}
Let $p\in(1,\infty)$ with H\"older conjugate $q$,
and let $G=(V,E,w)$ be a connected graph with $q$-Laplacian $L_q$
and second smallest eigenvalue
$\lambda_q^{(2)}$.
Then
\begin{equation}
    \forall s,t\in V,\quad
    d_p(s,t)\leq \prs{\frac{2}{\lambda_q^{(2)}}}^{1/q}.
\end{equation}
\end{claim}
Note that the above claim is tight on the clique
for $p=q=2$,
since $\lambda_2^{(2)}(K_n)=n$
($\chi_s-\chi_t$ is a corresponding eigenvector for any $s,t\in V$),
which yields the exact $d_2$-distance.
\begin{proof}
Fix $s\neq t\in V$,
and let $\varphi^*\in\argmin{\varphi_s-\varphi_t=1}\norm{WB\varphi}_q$.
By subtracting the average of the entries of $\varphi$,
we may assume that $\varphi^*\perp \OneVec$.
Now, we can see that
\begin{align*}
    \frac{1}{d_p(s,t)^q}
        & = \norm{WB\varphi^*}_q^q\\
        & = \iprod{\varphi^*}{L_q\varphi^*}\\
        & \geq \lambda_q^{(2)}\norm{\varphi^*}_2^2
        &&(\text{by }\varphi^*\perp \OneVec
        \text{ and definition of }\lambda_q^{(2)})\\
        & \geq \lambda_q^{(2)}\cdot \norm{\prs{
        \begin{array}{c}
             \varphi^*_s   \\
             -\varphi^*_t
        \end{array}
        }}_2^2\\
        & \geq \lambda_q^{(2)}\cdot \frac{\abs{\iprod{\prs{
        \begin{array}{c}
             \varphi^*_s   \\
             -\varphi^*_t
        \end{array}
        }}{\prs{
        \begin{array}{c}
             1   \\
             1
        \end{array}
        }}}^2}{\norm{\prs{
        \begin{array}{c}
             1   \\
             1
        \end{array}
        }}_2^2}
        && (\text{by Cauchy-Schwartz})\\
        & = \frac{\lambda_q^{(2)}}{2}
        && (\text{by }\varphi_s^*-\varphi_t^*=1).
\end{align*}
\end{proof}

\section{Graph-Size Reductions Appendix}

\subsection{Lower Bound on Resistance Sparsifiers}

\subsubsection{Proof of Symmetric Case via Commute Time}\label{appendix sectoin: Proof of Symmetric Case via Commute Time}
Here we present another proof for Claim
\ref{claim: lower bound on the symmetric case}
by using the relation between
effective resistance and commute time.

First, recall that the hitting time $h(u,v)$ is the 
expected number of steps of a random walk
starting at vertex $u$ to reach vertex $v$ at the fist time.
i.e. $h(u,v)=1+\Ex{x\in N(u)}h(x,v)$.
The commute time $C(u,v)$ is the expected time
that a random walk starting at $u$ will reach $v$
and get back to $u$, i.e.
$C(u,v)=h(u,v)+h(v,u)$.
In addition, we know that $C(u,v)=2w(E)\Reff(u,v)$
where $w(E)=\sum_{e\in E}w(e)$.
Thus, instead of considering effective resistance,
we can work with commute time.

\begin{proof}(of Claim \ref{claim: lower bound on the symmetric case})
Let $x\neq y\in V\backslash \{s,t\}$,
and denote
\begin{align}
    H_0 & = h(s,t) = h(t,s),\\
    H_1 & = h(x,s) = h(x,t),\\
    H_2 & = h(s,x) = h(t,x),\\
    H_3 & = h(x,y).
\end{align}
Note that by symmetry it does not matter which specific
vertices $x,y\in V\backslash\{s,t\}$ we chose
when defining the $H_i$'s.
Now,
we can see that the relations between them are
as follows.
\begin{align*}
    H_0  & = 1+\Ex{z\in N(s)}h(z,t)
     = 1+H_1,\\
    H_1  & = 1+\Ex{z\in N(x)}h(z,t) \\
    & = 1+\frac{\alpha}{2\alpha+(n-3)\beta}H_0+\frac{(n-3)\beta}{2\alpha
    +(n-3)\beta}H_1,\\
    H_2 & = 1+\Ex{z\in N(s)}h(z,x)\\
    & = 1+ \frac{(n-3)\alpha}{(n-2)\alpha}H_3,\\
    H_3  & = 1+\Ex{z\in N(y)}h(z,y)\\
    & = 1 + \frac{2\alpha}{2\alpha+(n-3)\beta}H_2 + \frac{(n-4)\beta}{2\alpha+
    (n-3)\beta}H_3.
\end{align*}
Let us start with computing $H_0$ and $H_1$.
\begin{align*}
    && H_1 & = 1 + \frac{\alpha}{2\alpha+(n-3)\beta}\prs{1+H_1}+\frac{(n-3)\beta}{2\alpha
    +(n-3)\beta}H_1\\
    && & = 1 + \frac{\alpha}{2\alpha+(n-3)\beta} + H_1\cdot \prs{\frac{\alpha+(n-3)\beta}{2\alpha
    +(n-3)\beta}}\\
    \implies &&  \frac{\alpha}{2\alpha+(n-3)\beta}\cdot H_1
    & = \frac{3\alpha+(n-3)\beta}{2\alpha+(n-3)\beta}\\
    \implies && 
    H_1 & = 3+(n-3)\frac{\beta}{\alpha},\\
    &&  H_0 & = 4+ (n-3)\frac{\beta}{\alpha}.
\end{align*}
Let us move on to $H_2$ and $H_3$.
\begin{align*}
    H_3 & = 1 + \frac{2\alpha}{2\alpha+(n-3)\beta}
    \prs{1+ \frac{(n-3)}{(n-2)}H_3}
    + \frac{(n-4)\beta}{2\alpha+(n-3)\beta}H_3\\
    & = 1 + \frac{2\alpha}{2\alpha+(n-3)\beta}
    +H_3\prs{\prs{1-\frac{1}{n-2}}\cdot\frac{2\alpha}{2\alpha+(n-3)\beta}
    +\frac{(n-4)\beta}{2\alpha+(n-3)\beta}}\\
    & = \frac{4\alpha+(n-3)\beta}{2\alpha+(n-3)\beta}
    + H_3\prs{
    \frac{2\alpha + (n-4)\beta}{2\alpha+(n-3)\beta}
    -\frac{2\alpha}{(n-2)(2\alpha+(n-3)\beta)}}\\
    & = 
    \frac{4\alpha+(n-3)\beta}{2\alpha+(n-3)\beta}
    +H_3\prs{1-\frac{\beta+\frac{2\alpha}{n-2}}{2\alpha+(n-3)\beta}}.
\end{align*}
Thus,
\begin{align*}
    \implies H_3 
    & = \frac{4\alpha+(n-3)\beta}{\beta+\frac{2\alpha}{n-2}}\\
    & = (n-2)\prs{\frac{4+(n-3)\frac{\beta}{\alpha}}{2+(n-2)\frac{\beta}{\alpha}}},\\
    \implies 
    H_2 & = 1+ \frac{n-3}{n-2}\prs{(n-2)\prs{\frac{4+(n-3)\frac{\beta}{\alpha}}{2+(n-2)\frac{\beta}{\alpha}}}}\\
    & =1+ (n-3)\prs{\frac{4+(n-3)\frac{\beta}{\alpha}}{2+(n-2)\frac{\beta}{\alpha}}}.
\end{align*}
Denoted $\gamma = \frac{\beta}{\alpha}$,
and now we conclude that,
\begin{align}
    C(s,t) & = 2H_0 
    = 2(4+(n-3)\gamma),\\
    C(x,y) & = 2H_3 
    = \frac{2(n-2)(4+(n-3)\gamma)}{2+(n-2)\gamma},\\
    C(s,x) & = H_1 + H_2
    = (4 + (n-3)\gamma)\prs{1
    + \frac{n-3}{2+(n-2)\gamma}}\\
    & = \frac{(4+(n-3)\gamma)(n-1+(n-2)\gamma)}{2+(n-2)\gamma}.
\end{align}
Assume towards contradiction that $\frac{\max_{x'\neq y'\in V} \Reff(x',y')}{\min_{x'\neq y'\in V} \Reff(x',y')}<1+\frac{1}{10n}$,
and note that in particular this implies that
for any $u,v,u',v'\in V$,
\[
    \frac{C(u,v)}{C(u',v')} = \frac{2w(E)\Reff(u,v)}{2w(E)\Reff(u',v')}<1+\frac{1}{10n}.
\]
Let us compute the ratios.
\begin{align*}
    \frac{C(s,t)}{C(x,y)}
    & = 2(4+(n-3)\gamma)\cdot \frac{2+(n-2)\gamma}{2(n-2)(4+(n-3)\gamma)}\\
    & = \gamma + \frac{2}{n-2}.
\end{align*}
Thus
\begin{align*}
    \gamma & < 1-\frac{2}{n-2}+\frac{1}{10n}.
\end{align*}
But on the other hand we see that,
\begin{align*}
    \frac{C(x,y)}{C(s,x)}
    & = \frac{2(n-2)(4+(n-3)\gamma)}{2+(n-2)\gamma}
    \cdot \frac{2+(n-2)\gamma}{(4+(n-3)\gamma)(n-1+(n-2)\gamma)}\\
    & = \frac{2(n-2)}{n-1+(n-2)\gamma}\\
    & = \frac{2}{1+\frac{1}{n-2}+\gamma}.
\end{align*}
and thus,
\begin{align*}
    2 & < \prs{1+\frac{1}{n-2}}\cdot \prs{1+\frac{1}{10n}}
    +\gamma\prs{1+\frac{1}{10n}}\\
    \implies \gamma & >
    \frac{2-\prs{1+\frac{1}{n-2}}\cdot \prs{1+\frac{1}{10n}}}
    {1+\frac{1}{10n}} \\
    & = 2\cdot\prs{1-\frac{1}{10n+1}}-\prs{1+\frac{1}{n-2}}\\
    & = 1-\frac{1}{n-2}-\frac{2}{10n+1}.
\end{align*}
and thus we conclude that
\begin{equation}
    1-\frac{1}{n-2}-\frac{2}{10n+1}
    < \gamma
    < 1-\frac{2}{n-2}+\frac{1}{10n},
\end{equation}
which is a contradiction.
\end{proof}

\subsection{Transforms for the Flow Metrics}

\subsubsection{Another proof for the parallel edges reduction via flows}\label{appendix  section - Another proof for the parallel edges reduction via flows}
In this subsection we present an alternative
proof for Claim \ref{claim: parallel edges reduction for d_p} via flows.
\begin{proof}(of Claim \ref{claim: parallel edges reduction for d_p})
Denote $f_p(x,t)=\prs{\frac{x}{\alpha}}^p+
\prs{\frac{t-x}{\beta}}^p$.
For any amount of  flow $0\leq t \leq 1$
that is shipped to $a$,
it is the best to minimize $f_p(x,t)$
where $0\leq x\leq t$ (with respect to $x$
where $t$ is fixed),
and thus choosing how much amount of flow to 
ship for the top edge and how much to ship
from the bottom edge.
We will use it to compute the contribution of the
discussed edges to the norm of the minimizing flow,
and then choose a proper weight $\gamma$
which will preserve the norm of the flow.
            
Let us compute the derivative of $f_p(x,t)$
with respect to $x$ and equalize it to 0
in order to find the minimizing flow
(in the interval $x\in [0,t]$).
\begin{align*}
    \frac{d}{dx}f_p(x,t) & = 
    \frac{p}{\alpha}\cdot\prs{\frac{x}{\alpha}}^{p-1}
    -\frac{p}{\beta}\cdot\prs{\frac{t-x}{\beta}}^{p-1}
    \stackrel{\text{want}}{=}0\\
    \Rightarrow\frac{x^{p-1}}{\alpha^p}
    & = \frac{(t-x)^{p-1}}{\beta^p}\\
    \Rightarrow
    x^{p-1} & = \prs{\frac{\alpha}{\beta}}^p
    \cdot(t-x)^{p-1}\\
    \Rightarrow x& = \prs{\frac{\alpha}{\beta}}^{\frac{p}{p-1}}
    \cdot(t-x)\\
    \Rightarrow\prs{1+\prs{\frac{\alpha}{\beta}}^{\frac{p}{p-1}}}\cdot x & = t\cdot \prs{\frac{\alpha}{\beta}}^{\frac{p}{p-1}}\\
    \Rightarrow x_0 & = t\cdot \frac{\prs{\frac{\alpha}{\beta}}^{\frac{p}{p-1}}}{1+\prs{\frac{\alpha}{\beta}}^{\frac{p}{p-1}}}
    = t\cdot \frac{\alpha^{\frac{p}{p-1}}}{\alpha^{\frac{p}{p-1}}+\beta^{\frac{p}{p-1}}}
    = t\cdot\prs{1- \frac{\beta^{\frac{p}{p-1}}}{\alpha^{\frac{p}{p-1}}+\beta^{\frac{p}{p-1}}}}\\
    \Rightarrow \min_{x\in[0,t]}f_p\prs{x,t}
    & = \prs{\frac{t\cdot\alpha^{\frac{1}{p-1}}}{\alpha^{\frac{p}{p-1}}+\beta^{\frac{p}{p-1}}}}^p
    +\prs{\frac{t\cdot\beta^{\frac{1}{p-1}}}{\alpha^{\frac{p}{p-1}}+\beta^{\frac{p}{p-1}}}}^p.
\end{align*}
Thus, we have found the contribution of the discussed
edges to the minimizing flow in the graph
(between specific vertices),
and we wish that the contribution of the new
edge will be the same.
The contribution of the new edge is $\abs{\frac{t}{\gamma}}^p$,
and hence setting
\[
    \gamma=\frac{1}{\prs{\prs{\frac{\alpha^{\frac{1}{p-1}}}{\alpha^{\frac{p}{p-1}}+\beta^{\frac{p}{p-1}}}}^p
    +\prs{\frac{\beta^{\frac{1}{p-1}}}{\alpha^{\frac{p}{p-1}}+\beta^{\frac{p}{p-1}}}}^p}^{1/p}}
    = \prs{\alpha^{\frac{p}{p-1}}+\beta^{\frac{p}{p-1}}}
    ^{\frac{p-1}{p}}.
\]
will give us the desired outcome.
In other words, if we take $q=\frac{p}{p-1}$
the H\"older conjugate of $p$,
we get the following rule:
\[
    \gamma^q=\alpha^q+\beta^q.
\]
which is the same conclusion as before.
\end{proof}

\subsubsection{Proof of Y-\texorpdfstring{$\Delta$}{Delta} transform for \textit{p}=2}\label{appendix section - Proof of Delta-Y transform for p=2}
In this  subsection we show the existence of a Y-$\Delta$
transform analogue for $d_2$.
\begin{claim}\label{claim: appendix - delta-Y transform for d_2}
There exists a local $Y$-$\Delta$ transform
that preserves $d_2$.
\end{claim}
Consider the case presented in figure 
\ref{fig:Y-Delta-transform}.
Recall that in the case of effective resistance,
the rule is as follows:
\begin{equation}
    \begin{split}
        \alpha&=\frac{w_b\cdot w_c}{w_a+w_b+w_c},\\
        \beta&=\frac{w_a\cdot w_c}{w_a+w_b+w_c},\\
        \gamma&=\frac{w_a\cdot w_b}{w_a+w_b+w_c}.
    \end{split}
\end{equation}
But in our case, 
recall that $d_2$ is in fact
the squareroot of the resistance distance
in the same graph but with squared edge weights.
Thus, we conclude that for $p=2$ the rule should be
\begin{equation*}
    \begin{split}
        \alpha&=\sqrt{\frac{w_b^2\cdot w_c^2}{w_a^2+w_b^2+w_c^2}},\\
        \beta&=\sqrt{\frac{w_a^2\cdot w_c^2}{w_a^2+w_b^2+w_c^2}},\\
        \gamma&=\sqrt{\frac{w_a^2\cdot w_b^2}{w_a^2+w_b^2+w_c^2}}.
    \end{split}
\end{equation*}

Below are the algebraic computations that show
that the numbers add up.
Throughout the proof we use the same notations
as presented for the proof of Claim \ref{claim: existence of graphs that contradict the existence of Y-Delta transform}.
\begin{proof}(of Claim \ref{claim: appendix - delta-Y transform for d_2})
We will show that this transform
indeed holds for $d_2$.
In fact, we will show that for any 
potential function $\varphi$ over the terminals
$V_T$,
the new weights satisfy
\begin{equation}\label{eq: Y-Delta transform - what we want to find}
    \min_{x\in \R}
    \curprs{w_a\abs{\varphi_a-x}^q
    +w_b\abs{\varphi_b-x}^q
    +w_c\abs{\varphi_c-x}^q}
    = 
    \gamma^q\abs{\varphi_a-\varphi_b}^q
    +\alpha^q\abs{\varphi_b-\varphi_c}^q
    +\beta^q\abs{\varphi_c-\varphi_a}^q.
\end{equation}
In particular, this will hold for the 
minimizing potential function, and will
lead to the desired outcome.

Let us compute the LHS of (\ref{eq: Y-Delta transform - what we want to find}), 
in the case where $q=p=2$ to verify the rule.
Note that the minimizing $x$ in the LHS of (\ref{eq: Y-Delta transform - what we want to find})
is the weighted average of the potentials (weighted by the weights of
the edges that connect them to the center of the star).
Define a random variable $X$ by 
\[
    \Pr\prs{X=x} = \begin{cases}
    \frac{w_a^2}{w_a^2+w_b^2+w_c^2}   &   \text{if }x=\varphi_a,   \\
    \frac{w_b^2}{w_a^2+w_b^2+w_c^2}   &   \text{if }x=\varphi_b,   \\
    \frac{w_c^2}{w_a^2+w_b^2+w_c^2}   &   \text{if }x=\varphi_c;
    \end{cases}
\]
and now we can
 can view the LHS of (\ref{eq: Y-Delta transform - what we want to find})
 as
\[
    \min_{x\in\R}\curprs{w_a^2\cdot\abs{\varphi_a-x}^2
    +w_b^2\cdot\abs{\varphi_b-x}^2
    +w_c^2\cdot\abs{\varphi_c-x}^2}
    = \min_{x\in\R}\curprs{\prs{w_a^2+w_b^2+w_c^2}\cdot\E\brs{\abs{X-x}^2}}.
\]
The minimum of the RHS in the above is exactly the variance of $X$,
and thus 
\[
    x=\E[X] =  \frac{w_a^2}{w_a^2+w_b^2+w_c^2}\cdot \varphi_a   +
    \frac{w_b^2}{w_a^2+w_b^2+w_c^2}\cdot\varphi_b   +
    \frac{w_c^2}{w_a^2+w_b^2+w_c^2}\cdot\varphi_c.
\]
            
Let us now plug this into the LHS of (\ref{eq: Y-Delta transform - what we want to find}) and compute.
\begin{footnotesize}
\begin{align*}
    & LHS(\ref{eq: Y-Delta transform - what we want to find})= \\
    & = w_a^2\cdot\abs{\varphi_a-\frac{w_a^2\cdot\varphi_a+w_b^2\cdot\varphi_b+w_c^2\cdot\varphi_c}{w_a^2+w_b^2+w_c^2}}^2\\
    &+w_b^2\cdot\abs{\varphi_b-\frac{w_a^2\cdot\varphi_a+w_b^2\cdot\varphi_b+w_c^2\cdot\varphi_c}{w_a^2+w_b^2+w_c^2}}^2\\
    &+w_c^2\cdot\abs{\varphi_c-\frac{w_a^2\cdot\varphi_a+w_b^2\cdot\varphi_b+w_c^2\cdot\varphi_c}{w_a^2+w_b^2+w_c^2}}^2\\
    & = \prs{\frac{w_a}{w_a^2+w_b^2+w_c^2}}^2 \cdot\abs{(w_b^2+w_c^2)\cdot\varphi_a-w_b^2\cdot\varphi_b-w_c^2\cdot\varphi_c}^2\\
    &+\prs{\frac{w_b}{w_a^2+w_b^2+w_c^2}}^2\cdot\abs{(w_a^2+w_c^2)\cdot\varphi_b-w_a^2\cdot\varphi_a-w_c^2\cdot\varphi_c}^2\\
    &+\prs{\frac{w_c}{w_a^2+w_b^2+w_c^2}}^2\cdot\abs{(w_a^2+w_b^2)\cdot\varphi_c-w_a^2\cdot\varphi_a-w_b^2\cdot\varphi_b}^2\\
    & = \prs{\frac{w_a}{w_a^2+w_b^2+w_c^2}}^2 \cdot\abs{w_b^2\cdot\prs{\varphi_a-\varphi_b}+w_c^2\cdot\prs{\varphi_a-\varphi_c}}^2\\
    &+\prs{\frac{w_b}{w_a^2+w_b^2+w_c^2}}^2\cdot\abs{w_a^2\cdot\prs{\varphi_b-\varphi_a}+w_c^2\cdot\prs{\varphi_b-\varphi_c}}^2\\
    &+\prs{\frac{w_c}{w_a^2+w_b^2+w_c^2}}^2\cdot\abs{w_a^2\cdot\prs{\varphi_c-\varphi_a}+w_b^2\cdot\prs{\varphi_c-\varphi_b}}^2\\
    & = \prs{\frac{w_a}{w_a^2+w_b^2+w_c^2}}^2 \cdot\prs{w_b^4\cdot\prs{\varphi_a-\varphi_b}^2
    +2\cdot w_b^2\cdot w_c^2\cdot\prs{\varphi_a-\varphi_b}\cdot 
    \prs{\varphi_a-\varphi_c}
    +w_c^4\cdot\prs{\varphi_a-\varphi_c}^2}\\
    &+\prs{\frac{w_b}{w_a^2+w_b^2+w_c^2}}^2 \cdot\prs{w_a^4\cdot\prs{\varphi_b-\varphi_a}^2
    +2\cdot w_a^2\cdot w_c^2\cdot\prs{\varphi_b-\varphi_a}\cdot 
    \prs{\varphi_b-\varphi_c}
    +w_c^4\cdot\prs{\varphi_b-\varphi_c}^2}\\
    &+\prs{\frac{w_c}{w_a^2+w_b^2+w_c^2}}^2 \cdot\prs{w_a^4\cdot\prs{\varphi_c-\varphi_a}^2
    +2\cdot w_a^2\cdot w_b^2\cdot\prs{\varphi_c-\varphi_a}\cdot 
    \prs{\varphi_c-\varphi_b}
    +w_b^4\cdot\prs{\varphi_c-\varphi_b}^2}\\
    & = \frac{w_a^2\cdot w_b^2\cdot \prs{w_a^2+w_b^2}\cdot\prs{\varphi_a
    -\varphi_b}^2
    +w_a^2\cdot w_c^2\cdot \prs{w_a^2+w_c^2}\cdot\prs{\varphi_a
    -\varphi_c}^2
    +w_b^2\cdot w_c^2\cdot \prs{w_b^2+w_c^2}\cdot\prs{\varphi_b
    -\varphi_c}^2}{\prs{w_a^2+w_b^2+w_c^2}^2}
    \\
    & + 2\cdot w_a^2\cdot w_b^2\cdot w_c^2
    \cdot\frac{\prs{\varphi_a-\varphi_b}\cdot \prs{\varphi_a-\varphi_c}+ 
    \prs{\varphi_b-\varphi_a}\cdot \prs{\varphi_b-\varphi_c}+ 
    \prs{\varphi_c-\varphi_a}\cdot \prs{\varphi_c-\varphi_b} }{\prs{w_a^2+w_b^2+w_c^2}^2}\\
    & = \frac{w_a^2\cdot w_b^2\cdot \prs{w_a^2+w_b^2}\cdot\prs{\varphi_a
    -\varphi_b}^2
    +w_a^2\cdot w_c^2\cdot \prs{w_a^2+w_c^2}\cdot\prs{\varphi_a
    -\varphi_c}^2
    +w_b^2\cdot w_c^2\cdot \prs{w_b^2+w_c^2}\cdot\prs{\varphi_b
    -\varphi_c}^2}{\prs{w_a^2+w_b^2+w_c^2}^2}
    \\
    & + 2\cdot w_a^2\cdot w_b^2\cdot w_c^2
    \cdot\frac{\varphi_a^2-\varphi_a\cdot \varphi_c-\varphi_a
    \cdot \varphi_b+\varphi_b\cdot \varphi_c+
    \varphi_b^2-\varphi_b\cdot \varphi_c-\varphi_a
    \cdot \varphi_b+\varphi_a\cdot \varphi_c+
    \varphi_c^2-\varphi_a\cdot \varphi_c-\varphi_b
    \cdot \varphi_c+\varphi_a\cdot \varphi_b
    }{\prs{w_a^2+w_b^2+w_c^2}^2}\\
    & = \frac{w_a^2\cdot w_b^2\cdot \prs{w_a^2+w_b^2}\cdot\prs{\varphi_a
    -\varphi_b}^2
    +w_a^2\cdot w_c^2\cdot \prs{w_a^2+w_c^2}\cdot\prs{\varphi_a
    -\varphi_c}^2
    +w_b^2\cdot w_c^2\cdot \prs{w_b^2+w_c^2}\cdot\prs{\varphi_b
    -\varphi_c}^2}{\prs{w_a^2+w_b^2+w_c^2}^2}
    \\
    & + 2\cdot w_a^2\cdot w_b^2\cdot w_c^2
    \cdot\frac{\varphi_a^2+\varphi_b^2
    +\varphi_c^2
    -\varphi_a\cdot \varphi_c
    -\varphi_b\cdot \varphi_c
    -\varphi_a\cdot \varphi_b
    }{\prs{w_a^2+w_b^2+w_c^2}^2}.
\end{align*}
\end{footnotesize}
            
Now let us compute the RHS of (\ref{eq: Y-Delta transform - what we want to find})
when applying the rule.
\begin{footnotesize}
\begin{align*}
    & RHS(\ref{eq: Y-Delta transform - what we want to find}) = \\
    & = \prs{\sqrt{\frac{w_a^2\cdot w_b^2}{w_a^2+w_b^2+w_c^2}}}^2\cdot\abs{\varphi_a-\varphi_b}^2
    +\prs{\sqrt{\frac{w_b^2\cdot w_c^2}{w_a^2+w_b^2+w_c^2}}}^2\cdot\abs{\varphi_b-\varphi_c}^2
    +\prs{\sqrt{\frac{w_a^2\cdot w_c^2}{w_a^2+w_b^2+w_c^2}}}^2\cdot\abs{\varphi_c-\varphi_a}^2\\
    & = \frac{w_a^2\cdot w_b^2}{w_a^2+w_b^2+w_c^2}\cdot\abs{\varphi_a-\varphi_b}^2
    +\frac{w_b^2\cdot w_c^2}{w_a^2+w_b^2+w_c^2}\cdot\abs{\varphi_b-\varphi_c}^2
    +\frac{w_a^2\cdot w_c^2}{w_a^2+w_b^2+w_c^2}\cdot\abs{\varphi_c-\varphi_a}^2\\
    & = \frac{{w_a^2\cdot w_b^2}\cdot\prs{w_a^2+w_b^2+w_c^2}\cdot\abs{\varphi_a-\varphi_b}^2
    +{w_b^2\cdot w_c^2}\cdot\prs{w_a^2+w_b^2+w_c^2}\cdot\abs{\varphi_b-\varphi_c}^2
    +{w_a^2\cdot w_c^2}\cdot\prs{w_a^2+w_b^2+w_c^2}\cdot\abs{\varphi_c-\varphi_a}^2}{\prs{w_a^2+w_b^2+w_c^2}^2}\\
    & = \frac{w_a^2\cdot w_b^2\cdot \prs{w_a^2+w_b^2}\cdot\prs{\varphi_a
    -\varphi_b}^2
    +w_a^2\cdot w_c^2\cdot \prs{w_a^2+w_c^2}\cdot\prs{\varphi_a
    -\varphi_c}^2
    +w_b^2\cdot w_c^2\cdot \prs{w_b^2+w_c^2}\cdot\prs{\varphi_b
    -\varphi_c}^2}{\prs{w_a^2+w_b^2+w_c^2}^2}
    \\
    & + w_a^2\cdot w_b^2\cdot w_c^2
    \cdot\frac{\prs{\varphi_a
    -\varphi_b}^2
    +\prs{\varphi_a
    -\varphi_c}^2
    +\prs{\varphi_b
    -\varphi_c}^2
    }{\prs{w_a^2+w_b^2+w_c^2}^2}\\
    & = \frac{w_a^2\cdot w_b^2\cdot \prs{w_a^2+w_b^2}\cdot\prs{\varphi_a
    -\varphi_b}^2
    +w_a^2\cdot w_c^2\cdot \prs{w_a^2+w_c^2}\cdot\prs{\varphi_a
    -\varphi_c}^2
    +w_b^2\cdot w_c^2\cdot \prs{w_b^2+w_c^2}\cdot\prs{\varphi_b
    -\varphi_c}^2}{\prs{w_a^2+w_b^2+w_c^2}^2}
    \\
    & + w_a^2\cdot w_b^2\cdot w_c^2
    \cdot\frac{\varphi_a^2
    -2\cdot\varphi_a\cdot\varphi_b+\varphi_b^2
    +\varphi_a^2
    -2\cdot\varphi_a\cdot\varphi_c+\varphi_c^2
    +\varphi_b^2
    -2\cdot\varphi_b\cdot\varphi_c+\varphi_c^2
    }{\prs{w_a^2+w_b^2+w_c^2}^2}\\
    & = \frac{w_a^2\cdot w_b^2\cdot \prs{w_a^2+w_b^2}\cdot\prs{\varphi_a
    -\varphi_b}^2
    +w_a^2\cdot w_c^2\cdot \prs{w_a^2+w_c^2}\cdot\prs{\varphi_a
    -\varphi_c}^2
    +w_b^2\cdot w_c^2\cdot \prs{w_b^2+w_c^2}\cdot\prs{\varphi_b
    -\varphi_c}^2}{\prs{w_a^2+w_b^2+w_c^2}^2}
    \\
    & +2\cdot  w_a^2\cdot w_b^2\cdot w_c^2
    \cdot\frac{\varphi_a^2+\varphi_b^2+\varphi_c^2
    -\varphi_a\cdot\varphi_b
    -\varphi_a\cdot\varphi_c
    -\varphi_b\cdot\varphi_c
    }{\prs{w_a^2+w_b^2+w_c^2}^2}.
\end{align*}
\end{footnotesize}
Thus we got that both quantities are the same,
and hence we have a proper $Y-\Delta$ transform
for $d_2$.
\end{proof}

\section{Embedding Conjecture Appendix}

In section \ref{chapter: conclusions and open questions}, we discussed the geometry
of the flow metrics, which we would like to better understand.
We mentioned Conjecture \ref{conjecture: d_p into l_q embedding conjecture}, which we restate here.
\begin{conjecture}\label{conjecture: embedding dp into lq}
Let $p\in [1,\infty]$ with H\"older conjugate $q$,
and let $G=(V,E,w)$ be a graph.
Then, there exists a mapping $\Phi:V\rightarrow \ell_q$
such that,
\begin{equation}
    \forall s,t\in V,\quad
    d_p(s,t)=\norm{\Phi(s)-\Phi(t)}_q.
\end{equation}
\end{conjecture}
We remark that for the special cases $p=1,2,\infty$
it is known to hold.
\paragraph{Effective Resistance.}
Spielman and Srivastava \cite{spielmanSrivastava2011graphSparsificationReff}
showed that the resistance distance can be 
isometrically embedded into $\ell_2^2$.
Thus, since $d_2$ is the squareroot
of the effective resistance,
we conclude that $d_2$ can be isometrically embedded
into $\ell_2$ as desired (since $p=q=2$ in this special
case).

\paragraph{Shortest-path.}
Note that every finite metric space embeds
isometrically into $\ell_\infty$
\cite{MatousekDiscreteGeometry},
and thus $d_1$ does as well.

\paragraph{Minimum Cuts.}
We remark that $d_\infty$ is in fact an ultrametric,
and thus it embeds isometrically into $\ell_1$.

\subsection{The \textit{p}-strong Triangle Inequality is not Enough}\label{section: p strong is not enough for embedding conjecture}

In this subsection we show that even though the flow
metrics satisfy a stronger version of the triangle inequality
(Theorem \ref{theorem: p-strong tirangle inequality}),
it is not enough in order to prove 
Conjecture \ref{conjecture: embedding dp into lq}.
We first recall the relevant definitions.
\begin{definition}
Let $\prs{X,d_X}$ and $\prs{Y,d_Y}$ be metric spaces,
and let $\Phi:X\rightarrow Y$ be a mapping (which
we call an embedding).
The distortion of $\Phi$
is the minimum $D\geq 1$ for which
there exists a scaling factor $\alpha>0$,
such that
\begin{equation}
    \forall x,x'\in X,\quad
    d_X(x,x')\leq \alpha\cdot d_Y(\Phi(x),\Phi(x'))
    \leq D\cdot d_X(x,x').
\end{equation}
\end{definition}
\begin{definition}
Let $\prs{X,d_X}$ and $\prs{Y,d_Y}$ be metric spaces,
and let $D\geq 1$.
We say that $\prs{X,d_X}$ $D$-embeds into $\prs{Y,d_Y}$
if there exists an embedding of $\prs{X,d_X}$ into
$\prs{Y,d_Y}$ with distortion $D$.
\end{definition}
Our focus in this subsection is to show the following claim.
\begin{claim}\label{claim: what we want to show in section 1}
For all $p\in [1,\infty)$, $q\in [1,\infty)$,
and  $n\geq 2$, there exists an $n$-point metric space
$(X,d)$, that satisfies the $p$-strong triangle inequality,
but $D$-embeds into $\ell_q$
only for $D=\Omega\prs{ \frac{1}{q}\cdot\prs{\logn}^{1/p}}$.
\end{claim}
In other words, the above claim says that the fact that a metric
space  is  $p$-strong isn't enough to guarantee
isometric embedding into $\ell_q$, for
any $q\in [1,\infty)$.
We remark that $p$ and $q$ in the statement
are not necessarily H\"older
conjugates of each other.
But, in the specific case where they are,
and $p\in[2,\infty)$,
it holds that $q\in[1,2]$.
Thus,  $D=\Omega\prs{\prs{\logn}^{1/p}}$
and as a consequence, 
we will need to use  additional properties of 
the family of the flow-metrics if we desire to prove the Conjecture \ref{conjecture: embedding dp into lq}.

Our proof of Claim \ref{claim: what we want to show in section 1} relies on the following theorem,
presented by 
Matou{\v{s}}ek \cite{matouvsek1997embedding}.
\begin{theorem}\label{theorem: expanders embed in lp with distortion omega(logn/p)}
There exists constants 
$c_1>0$ and $n_0\in\N$ such that for any $p\geq 1$
and any $n\geq n_0$ there exists an $n$-point metric space
which $D$-embeds into $\ell_p$
only for $D\geq \frac{c_1}{p}\cdot\logn$.
\end{theorem}
\begin{proof}(Claim \ref{claim: what we want to show in section 1})
Theorem \ref{theorem: expanders embed in lp with distortion omega(logn/p)}
was proved via expanders - the metric space
that is promised to exist in the theorem arises from this family of graphs.
Let $G$ be an expander on $n$ vertices, and let $(X,d)$
be the metric space that is derived from the shortest-path
metric on $G$.
Let $p\in(1,\infty)$, and define $d':X\times X\rightarrow \R^+$ by 
\[
\forall x,y\in X,\quad d'(x,y)=\prs{d(x,y)}^{1/p}.
\]
We would like to show the following.
\begin{enumerate}
    \item   $d'$ is a metric.
    \item   $d'$ is $p$-strong.
    \item   Embedding $(X,d')$ in $\ell_q$ requires large distortion
            (or at least larger than 1).
\end{enumerate}
First, note that property 2 is  immediate, since
it simply says that
\begin{align*}
    \forall x,y,z\in X,&&\prs{d'(x,y)}^p&\leq 
    \prs{d'(x,z)}^p+\prs{d'(z,y)}^p \\
     &&\iff d(x,y)&\leq 
    d(x,z)+d(z,y).
\end{align*}
where the last line is just the triangle inequality which is
clearly satisfied since $(X,d)$ is a metric space.

\noindent
For showing property 1, we recall Claim
\ref{claim: ap+bp leq (a+b)p}.
 \begin{claim}
 Let $a_1,\dots,a_n\geq 0$ and let $p>0$, then:
 \begin{enumerate}
     \item   if $p\leq 1$:
     \begin{equation*}
         \sum_{i=1}^na_i^p\geq \prs{\sum_{i=1}^na_i}^p.
     \end{equation*}
     \item   if $p\geq 1$:
     \begin{equation*}
         \sum_{i=1}^na_i^p\leq \prs{\sum_{i=1}^na_i}^p.
     \end{equation*}
 \end{enumerate}
 \end{claim}
To prove property 1,
let $x,y,z\in X$, and observe that
\begin{align*}
        d'(x,y) & = \prs{d(x,y)}^{\frac{1}{p}}\\
        & \leq \prs{d(x,z)+d(z,y)}^{\frac{1}{p}}
        && (d\text{ is a metric})\\
        & \leq \prs{d(x,z)}^{\frac{1}{p}}+\prs{d(z,y)}^{\frac{1}{p}}
        && (1\leq p \text{ and Claim }\ref{claim: ap+bp leq (a+b)p})\\
        & = d'(x,z) + d'(z,y).
\end{align*}
We are ready to show that property 3 holds.
Assume that $\Phi:(X,d')\rightarrow \ell_q$ is a $D$-embedding.
Thus, there exists a scaling factor $\alpha>0$ such that,
\[\forall x,y\in X,
\quad d'(x,y)\leq \alpha\cdot\norm{\Phi(x)-\Phi(y)}_q\leq D\cdot d'(x,y).
\]
Hence
\begin{equation}\label{eq: f is D-embedding of d^(1/p)}
    \forall x,y\in X,\quad 
d(x,y)^{1/p}\leq \alpha\cdot\norm{\Phi(x)-\Phi(y)}_q\leq D\cdot d(x,y)^{1/p}
\end{equation}
Next, denote  $\rho = diam(G)$,
and let us  consider an embedding $\Psi$
defined by
\[\Psi(x)=\rho^{\frac{p-1}{p}}\cdot\alpha\cdot \Phi(x).
\]
Fix $x,y\in X$, thus
\[\norm{\Psi(x)-\Psi(y)}_q 
={\rho^{\frac{p-1}{p}}}\cdot\alpha\cdot \norm{\Phi(x)-\Phi(y)}_q.
\]
But now we can see that on the one hand
\begin{equation}\label{eq: lower bound on g}
    \begin{split}
        \norm{\Psi(x)-\Psi(y)}_q 
        & \geq \rho^{\frac{p-1}{p}}\cdot \prs{d(x,y)}^{1/p}\\
        & \geq \prs{d(x,y)^{p-1}\cdot d(x,y)}^{1/p}\\
        & = d(x,y),
    \end{split}
\end{equation}
and on the other hand, since $d(x,y)\geq 1$,
\begin{equation}\label{eq: upper bound on g}
    \begin{split}
        \norm{\Psi(x)-\Psi(y)}_q 
        & \leq \rho^{\frac{p-1}{p}}\cdot \prs{D\cdot \prs{d(x,y)}^{1/p}}\\
        & \leq \prs{D\cdot \rho^{\frac{p-1}{p}}}\cdot d(x,y).
    \end{split}
\end{equation}
Combining equations (\ref{eq: lower bound on g}) and 
(\ref{eq: upper bound on g})
together, we can conclude that $\Psi$ is a $\prs{D\cdot \rho^{\frac{p-1}{p}}}$-embedding of $(X,d)$ 
(the original metric space) into $\ell_q$.
But according to Theorem \ref{theorem: expanders embed in lp with distortion omega(logn/p)},
 ${D\cdot \rho^{\frac{p-1}{p}}}\geq \frac{c_1}{q}\cdot\logn$
 and $\rho=diam(G)=O(\logn)$,
which implies that
$D=\Omega\prs{\prs{\logn}^{1/p}/q}$
as claimed.
\end{proof}

\end{document}